\documentclass{sig-alternate-10}
\usepackage{amsfonts,amssymb}
\usepackage{balance}  % for  \balance command ON LAST PAGE  (only there!)
\usepackage{graphicx, MnSymbol, hyperref, subfigure}
\usepackage{multirow, framed, float, tabularx}
\usepackage{aliascnt}
\usepackage{relsize}
\usepackage{listings}
\usepackage{colortbl}
\usepackage{cite}
\usepackage{tikz,pgf}
\usepackage{mathrsfs}
\usetikzlibrary{arrows}
\newcommand{\cut}[1]{}

\def\sabc{\signature_{\mathrm{rl}}}

\def\stk{\signature_{2\mathrm{fd}}}

\def\stfd{\signature_{2\mathrm{r}}}

\def\str{\signature_{\mathrm{tr}}}

\def\rtk{R_{2\mathrm{fd}}}

\newcommand{\attr}{{\mathit{attr}}}

\newcommand{\val}{{\mathtt{val}}}

\definecolor{LightCyan}{rgb}{0.88,1,1}
\definecolor{Gray}{gray}{0.9}

\newcommand{\sudeepa}[1]{{{\tt \color{blue} Sudeepa: [{#1}]}}}
\newcommand{\red}[1]{{\textbf {\tt {\color{red} [{#1}]}}}}

\newcommand{\dom}{{\mathsf{Val}}}

\newcommand{\expl}{{\phi}}

\newcommand{\true}{{\tt true}}
\newcommand{\false}{{\tt false}}

\newcommand{\proj}[1]{{\Pi}}
\newcommand{\sel}[1]{{\sigma}}

\newcommand{\cutfull}[1]{}

\newcommand{\commentresolved}[1]{}

\newcommand{\ie}{{i.e.}} %\xspace}
\newtheorem{theorem}{Theorem}[section]          	% Theorem environment.
\newaliascnt{lemma}{theorem}				% 1 alias counter
\newtheorem{lemma}[lemma]{Lemma}              	% Lemma environment.
\aliascntresetthe{lemma}  					% 3 set
\newaliascnt{conjecture}{theorem}			% 1 alias counter
\aliascntresetthe{conjecture}  				% 3 set
\newaliascnt{remark}{theorem}				% 1 alias counter
              
\aliascntresetthe{remark}  					% 3 set
\newaliascnt{corollary}{theorem}			% 1 alias counter
\newtheorem{corollary}[corollary]{Corollary}      % Corollary environment.
\aliascntresetthe{corollary}  				% 3 set
\newaliascnt{definition}{theorem}			% 1 alias counter
\aliascntresetthe{definition}  				% 3 set
\newaliascnt{proposition}{theorem}			% 1 alias counter
\newtheorem{proposition}[proposition]{Proposition}  % proposition environment.
\aliascntresetthe{proposition}  				% 3 set
\newaliascnt{example}{theorem}			% 1 alias counter
\aliascntresetthe{example}  				% 3 set
\newaliascnt{observation}{theorem}			% 1 alias counter
\aliascntresetthe{observation}  				% 3 set
\cut{

}

\usepackage{algorithm,algorithmicx}
\usepackage[noend]{algpseudocode}
\usepackage{xspace}
\usepackage{caption}
\usepackage{stackengine}
\usepackage{paralist}

\usepackage{ifmtarg}% http://ctan.org/pkg/ifmtarg
\makeatletter
\newcommand{\toprule}{\hrule height.8pt depth0pt \kern2pt} % Caption top horizontal rule+skip
\newcommand{\midrule}{\kern2pt\hrule\kern2pt} % Caption bottom (or mid) horizontal rule+skip
\newcommand{\bottomrule}{\kern2pt\hrule\relax}% Algorithm bottom rule
\newcommand{\algcaption}[2][]{%
  \refstepcounter{algorithm}%
  \@ifmtarg{#1}
    {\addcontentsline{loa}{figure}{\protect\numberline{\thealgorithm}{\ignorespaces #2}}}
    {\addcontentsline{loa}{figure}{\protect\numberline{\thealgorithm}{\ignorespaces #1}}}%
  \toprule
  \textbf{\fname@algorithm~\thealgorithm}\ #2\par % Caption
  \midrule
}
\makeatother

\def\e#1{\emph{#1}}

\newcommand{\eat}[1]{}

\newcommand{\set}[1]{\{#1\}}
\def\eqdef{\stackrel{\textsf{\tiny def}}{=}}

\newcommand{\algname}[1]{{\sf #1}}
\def\myrulewidth{3.20in}

\def\therule{\rule{\myrulewidth}{0.2pt}}

\newenvironment{insidecode}[3]
{
\begin{tabular}{p{\myrulewidth}}
\multicolumn{1}{c}{\rule{0mm}{3mm}{\bf #3} $\algname{#1}(\mbox{#2})$\vspace{-0.6em}}\\
\therule\vskip-0.8em\therule
\vspace{0em}
\begin{algorithmic}[1]}
{\end{algorithmic}
\vskip-0.3em\therule
\end{tabular}}

\def\consts{\mathsf{Const}}

\newcommand{\depset}{\mathrm{\Delta}}
\newcommand{\tup}[1]{\mathbf{#1}}

\def\phi{\varphi}
\def\signature{\mathcal{S}}

\newenvironment{repeatresult}[2]
{\vskip0.5em\par\textsc{#1} #2.\em}
{\vskip1em}

\newenvironment{repproposition}[1]{\begin{repeatresult}{Proposition}{#1}}{\end{repeatresult}}
\newenvironment{reptheorem}[1]{\begin{repeatresult}{Theorem}{#1}}{\end{repeatresult}}

\newtheorem{examplethm}[theorem]{Example}
\newenvironment{example}{\begin{examplethm}\em}
{\qed\end{examplethm}}

\newtheorem{commentthm}[theorem]{Comment}
\newenvironment{comment}{\begin{commentthm}\em}{\qed\end{commentthm}}

\newenvironment{citedtheorem}[1]
{\begin{theorem}\hskip-0.2em\e{\cite{#1}}\,\,}
{\end{theorem}}

\newenvironment{citedlemma}[1]
{\begin{lemma}\hskip-0.2em\e{\cite{#1}}\,\,}
{\end{lemma}}

\def\true{\mathbf{true}}
\def\false{\mathbf{false}}

\def\closure{\mathit{cl}}

\def\partitle#1{\vskip0.3em\par\noindent\textbf{#1.}\,\,}

\def\val#1{\texttt{#1}}

\def\rel#1{\textsc{#1}}
\def\att#1{\textsf{#1}}

\newenvironment{citemize}
{\begin{compactitem}}
{\end{compactitem}}

\newcommand{\mathsc}[1]{{\normalfont\textsc{#1}}}

\def\fdparam#1#2{\langle #1,#2\rangle}

\def\maxvrepfd#1#2{\mathsc{MinVRep}\fdparam{#1}{#2}}

\DeclareMathOperator*{\argmax}{argmax}

\def\mc{\mathit{mlc}}

\def\ids{\mathit{ids}}
\def\ra{\rightarrow}

\def\dists{\mathit{dist}_{\mathsf{sub}}}
\def\distu{\mathit{dist}_{\mathsf{upd}}}

\def\mComment#1{\Comment{\textit{#1}}}

\newenvironment{subroutine}
{\begin{algorithm}\floatname{algorithm}{Subroutine}}
{\end{algorithm}}

\newcounter{subroutine}

\newenvironment{msubroutine}[2]
{
\stepcounter{subroutine}
\addtocounter{algorithm}{-1}
\renewcommand{\thealgorithm}{\arabic{subroutine}} 
\floatname{algorithm}{Subroutine}
\algcaption{#1}\label{#2}
\begin{algorithmic}[1]
}
{
\end{algorithmic}
\bottomrule
}

\newenvironment{malgorithm}[2]
{
\algcaption{#1}\label{#2}
\begin{algorithmic}[1]
}
{
\end{algorithmic}
\bottomrule
}

\def\asn{\mathrel{{:}{=}}}

\def\cost{\mathrm{cost}}

\allowdisplaybreaks
\begin{document}
\title{Computing Optimal Repairs for Functional Dependencies} %
%\titlenote{This work was supported in part by NSF awards IIS-1408846
%  and IIS-1552538, and a Google Faculty Research Award.}}
\numberofauthors{3}
\author{
\alignauthor Ester Livshits\\
\affaddr{Technion}\\
\affaddr{Haifa 32000, Israel}
\alignauthor Benny Kimelfeld \\
\affaddr{Technion}\\
\affaddr{Haifa 32000, Israel}
\alignauthor Sudeepa Roy \\
\affaddr{Duke University}\\
\affaddr{Durham, NC 27708, USA}
}

\maketitle
\pagestyle{plain}
\pagenumbering{arabic}

%\begin{sloppypar}
\begin{abstract}
We investigate the complexity of computing an optimal repair of an
inconsistent database, in the case where integrity constraints are
Functional Dependencies (FDs). We focus on two types of repairs: an
optimal subset repair (optimal S-repair) that is obtained by a minimum
number of tuple deletions, and an optimal update repair (optimal
U-repair) that is obtained by a minimum number of value (cell)
updates. For computing an optimal S-repair, we present a
polynomial-time algorithm that succeeds on certain sets of FDs and
fails on others. We prove the following about the algorithm. When it
succeeds, it can also incorporate weighted tuples and duplicate
tuples. When it fails, the problem is NP-hard, and in fact,
APX-complete (hence, cannot be approximated better than some
constant). Thus, we establish a dichotomy in the complexity of
computing an optimal S-repair. We present general analysis techniques
for the complexity of computing an optimal U-repair, some based on
the dichotomy for S-repairs.  We also draw a connection to a
past dichotomy in the complexity of finding a ``most probable
database'' that satisfies a set of FDs with a single attribute on the
left hand side; the case of general FDs was left open, and we show how
our dichotomy provides the missing generalization and thereby settles
the open problem.

\end{abstract}

\section{Introduction}\label{sec:introduction}

Database inconsistency arises in a variety of scenarios and for
different reasons. For instance, data may be collected from imprecise
sources (social encyclopedias/networks, sensors attached to
appliances, cameras, etc.) via imprecise procedures (natural-language
processing, signal processing, image analysis, etc.). Inconsistency
may arise when integrating databases of different organizations with
conflicting information, or even consistent information in
conflicting formats.  Arenas et al.~\cite{DBLP:conf/pods/ArenasBC99}
introduced a principled approach to managing inconsistency via the
notions of \e{repairs} and \e{consistent query answering}. An
\e{inconsistent database} is a database $D$ that violates integrity
constraints, a \emph{repair} is a consistent database $D'$ 
obtained from $D$ by a minimal sequence of operations, and the
\e{consistent answers} to a query are the answers given in every
repair $D'$.

Instantiations of the repair framework differ in their definitions of
\e{integrity constraints}, \e{operations}, and
\e{minimality}~\cite{DBLP:conf/icdt/AfratiK09}. Common types of
constraints are denial
constraints~\cite{DBLP:journals/jiis/GaasterlandGM92} that include the
classic functional dependencies (FDs), and inclusion
dependencies~\cite{DBLP:journals/jcss/CasanovaFP84} that include the
referential (foreign-key) constraints.  An operation can be a
\e{deletion} of a tuple, an \e{insertion} of a tuple, and an
\e{update} of an attribute (cell) value. Minimality can be either
\e{local}---no strict subset of the operations achieves consistency,
or \e{global}---no smaller (or cheaper) subset achieves consistency.
For example, if only tuple deletions are allowed, then a \e{subset
  repair}~\cite{DBLP:journals/iandc/ChomickiM05} corresponds to a
local minimum (restoring any deleted tuple causes inconsistency) and a
\e{cardinality repair}~\cite{DBLP:conf/icdt/LopatenkoB07} corresponds
to a global minimum (consistency cannot be gained by fewer tuple
deletions).  The \e{cost} of operations may differ between tuples;
this can represent different levels of trust that we have in the
tuples~\cite{DBLP:conf/icdt/LopatenkoB07,DBLP:conf/icdt/KolahiL09}.

In this paper, we focus on global minima under FDs via tuple deletions
and value updates. Each tuple is associated with a weight that
determines the cost of its deletion or a change of a single value. We
study the complexity of computing a minimum repair in two settings:
\e{(a)} only tuple deletions are allowed, that is, we seek a
(weighted) cardinality repair, \e{and (b)} only value updates are
allowed, that is, we seek what Kolahi and
Lakshmanan~\cite{DBLP:conf/icdt/KolahiL09} refer to as an ``optimum
V-repair.'' We refer to the two challenges as computing an optimal
subset repair (optimal S-repair) and computing an optimal update
repair (optimal U-repair).

The importance of computing an optimal repair arises in the challenge
of \e{data cleaning}~\cite{DBLP:series/synthesis/2012Fan}---eliminate
errors and dirt (manifested as inconsistencies) from the
data\-base. Specifically, our motivation is twofold. The obvious
motivation is in fully automated cleaning, where an optimal repair is
the best candidate, assuming the system is aware of only the
constraints and tuple weights. The second motivation comes from the
more realistic practice of iterative, human-in-the-loop
cleaning~\cite{DBLP:journals/pvldb/BergmanMNT15,DBLP:conf/icde/AssadiMN17,
  DBLP:conf/sigmod/DallachiesaEEEIOT13,DBLP:journals/pvldb/GeertsMPS13}. There,
the cost of the optimal repair can serve as an educated estimate for
the extent to which the database is dirty and, consequently, the
amount of effort needed for completion of cleaning.

As our integrity constraints are FDs, it suffices to consider a
database with a single relation, which we call here a \e{table}. In a
general database, our results can be applied to each relation
individually. A table $T$ conforms to a relational schema
$R(A_1,\dots,A_k)$ where each $A_i$ is an attribute. Integrity is
determined by a set $\depset$ of FDs.  Our complexity analysis focuses
primarily on \e{data complexity}, where $R(A_1,\dots,A_k)$ and
$\depset$ are considered fixed and only $T$ is considered
input. Hence, we have infinitely many optimization problems, one for
each combination of $R(A_1,\dots,A_k)$ and $\depset$.  Table records
have identifiers, as we wish to be able to determine easily which
cells are updated in a repair. Consequently, we allow duplicate tuples
(with distinct identifiers).

We begin with the problem of computing an optimal S-repair. The
problem is known to be computationally hard for denial
constraints~\cite{DBLP:conf/icdt/LopatenkoB07}. As we discuss later,
complexity results can be inferred from prior work~\cite{GVSBUDA14}
for FDs with a single attribute on the left hand side (lhs for
short). For general FDs, we present the algorithm
$\algname{OptSRepair}$ (Algorithm~\ref{alg:osr}). The algorithm seeks
opportunities for simplifying the problem by eliminating attributes
and FDs, until no FDs are left (and then the problem is trivial). For
example, if all FDs share an attribute $A$ on the left hand side, then
we can partition the table according to $A$ and solve the problem
separately on each partition; but now, we can ignore $A$. We refer to
this simplification as ``common lhs.'' Two additional simplifications
are the ``consensus'' and ``lhs marriage.'' Importantly, the algorithm
terminates in polynomial time, even under \e{combined} complexity.

% We apply these simplifications iteratively, until all FDs are
% eliminated.

However, $\algname{OptSRepair}$ may fail by reaching a nonempty set of
FDs where no simplification can be applied.  We prove two properties
of the algorithm. The first is \e{soundness}---if the algorithm
succeeds, then it returns an optimal S-repair. More interesting is the
property of \e{completeness}---if the algorithm fails, then the
problem is NP-hard. In fact, in this case the problem is APX-complete,
that is, for some $\epsilon>0$ it is NP-hard to find a consistent
subset with a cost lower than $(1+\epsilon)$ times the minimum, but
\e{some} $(1+\epsilon')$ is achievable in polynomial time.  More so,
the problem remains APX-complete if we assume that the table does not
contain duplicates, and all tuples have a unit weight (in which case
we say that $T$ is \e{unweighted}). Consequently, we establish the
following dichotomy in complexity for the space of combinations of
schemas $R(A_1,\dots,A_k)$ and FD sets $\depset$.  
\vskip0.3em
\begin{citemize}
\item If we can eliminate all FDs in $\depset$ with the three
  simplifications, then an optimal S-repair can be computed in
  polynomial time using
  $\algname{OptSRepair}$.
  \vskip0.4em 
\item Otherwise, the problem is APX-complete, even for
  unweighted
  tables without duplicates.
\end{citemize}
\vskip0.4em

We then continue to the problem of computing an optimal U-repair. Here
we do not establish a full dichotomy, but we make a substantial
progress. We have found that proving hardness results for updates is
far more subtle than for deletions. We identify conditions where the
complexity of computing an optimal U-repair and that of computing an
optimal S-repair coincide. One such condition is the common lhs (i.e.,
all FDs share a left-hand-side attribute). Hence, in this case, our
dichotomy provides the precise test of tractability.  We also show
decomposition techniques that extend the opportunities of using the
dichotomy. As an example, consider
$\rel{Purchase}(\att{product},\att{price},\att{buyer},\att{email},\att{address})$
and
$\depset_0=\set{\att{product}\ra
  \att{price}\,,\,\att{buyer}\ra\att{email}}$.
We can decompose this problem into
$\depset_1=\set{\att{product}\ra \att{price}}$ and
$\depset_2=\set{\att{buyer}\ra\att{email}}$, and consider each
$\depset_i$, for $i=1,2$, independently. The complexity of each
$\depset_i$ is the same in both variants of optimal repairs, and so,
polynomial time.  Yet, these results do not cover all sets of FDs. For
example, let
$\depset_3=\set{\att{email}\ra \att{buyer}\,,\,\att{buyer}
  \ra\att{address}}$.
Kolahi and Lakshmanan~\cite{DBLP:conf/icdt/KolahiL09} proved that
under $\depset_3$, computing an optimal U-repair is NP-hard. Our
dichotomy shows that it is also NP-hard (and also APX-complete) to
compute an S-repair under $\depset_3$. Yet, this FD set does not fall
in our coincidence cases.

The above defined $\depset_0$ is an example where an optimal U-repair can be
computed in polynomial time, but computing an optimal S-repair is
APX-complete. We also show an example in the reverse direction, namely
$\depset_4=\set{\att{buyer}\ra \att{email}\,,\,\att{email}\ra\att{buyer}\,,\, \att{buyer}\ra\att{address}}$.
This FD set falls in the positive side of our dichotomy for optimal
S-repairs, but computing an optimal U-repair is APX-complete.  The
proof of APX-hardness is inspired by, but considerably more involved
than, the hardness proof of Kolahi and
Lakshmanan~\cite{DBLP:conf/icdt/KolahiL09} for $\depset_3$.

Finally, we consider approximate repairing. For the case of an optimal
S-repair, the problem easily reduces to that of \e{weighted vertex
  cover}, and hence, we get a polynomial-time 2-approximation due to
Bar-Yehuda and Even~\cite{DBLP:journals/jal/Bar-YehudaE81}. To
approximate optimal U-repairs, we show an efficient reduction to
S-repairs, where the loss in approximation is linear in the number of
attributes.  Hence, we obtain a constant-ratio approximation, where
the constant has a linear dependence on  the number of
attributes. Kolahi and Lakshmanan~\cite{DBLP:conf/icdt/KolahiL09} also
gave an approximation for optimal U-repairs, but their worst-case
approximation can be quadratic in the number of attributes. We show an
infinite sequence of FD sets where this gap is actually realized. On
the other hand, we also show an infinite sequence where our
approximation is linear in the number of attributes, but theirs
remains constant. Hence, in general, the two approximations are
incomparable, and we can combine the two by running both
approximations and taking the best.

Stepping outside the framework of repairs, a different approach to
data cleaning is
\e{probabilistic}~\cite{DBLP:journals/pvldb/RekatsinasCIR17,DBLP:conf/icde/AndritsosFM06,GVSBUDA14}. The
idea is to define a probability space over possible clean databases,
where the probability of a database is determined by the extent to
which it satisfies the integrity constraints. The goal is to find a
\e{most probable database} that, in turn, serves as the clean outcome.
As an instantiation, Gribkoff, Van den Broeck, and
Suciu~\cite{GVSBUDA14} identify probabilistic cleaning as the ``Most
Probable Database'' problem (MPD): given a tuple-independent
probabilistic
database~\cite{DBLP:conf/vldb/DalviS04,Suciu:2011:PD:2031527} and a
set of FDs, find the most probable database among those satisfying the
FDs (or, put differently, condition the probability space on the
FDs). They show a dichotomy for unary FDs (i.e., FDs with a single
attribute on the left hand side). The case of general (not necessarily
unary) FDs has been left open. It turns out that there are reductions
from MPD to computing an optimal S-repair and vice
versa. Consequently, we are able to generalize their dichotomy to all
FDs, and hence, fully settle the open problem.

\eat{
The rest of the paper is organized as follows. In
Section~\ref{sec:prelim} we give the basic definitions and problem
statements. We study the problem of computing an optimal S-repair in
Section~\ref{sec:subset-repairs}, where we also discuss the connection
to MPD. In Section~\ref{sec:update-repairs} we study the problem of
computing an optimal U-repair. We conclude and discuss future
directions in Section~\ref{sec:conclusions}.  For lack of space, most
of the proofs are in the appendix.
}
\section{Preliminaries}\label{sec:prelim}

We first present some basic terminology and notation that we use
throughout the paper.

\subsection{Schemas and Tables}
An instance of our data model is a single table where each tuple is
associated with an \e{identifier} and a \e{weight} that states how
costly it is to change or delete the tuple. Such a table corresponds
to a \e{relation schema} that we denote by $R(A_1,\dots,A_k)$, where
$R$ is the \e{relation name} and $A_1$, \dots, $A_k$ are distinct
\e{attributes}. We say that $R(A_1,\dots,A_k)$ is \e{$k$-ary} since it
has $k$ attributes.  When there is no risk of confusion, we may refer
to $R(A_1,\dots,A_k)$ by simply $R$. 

We use capital letters from the beginning of the English alphabet
(e.g., $A$, $B$, $C$), possibly with subscripts and/or superscripts,
to denote individual attributes, and capital letters from the end of
the English alphabet (e.g., $X$, $Y$, $Z$), possibly with subscripts
and/or superscripts, to denote sets of attributes. We follow the
convention of writing sets of attributes without curly braces and
without commas (e.g., $ABC$). 
%By a slight abuse of notation, we allow
%\e{set} operations on attribute sequences, and assume the obvious
%interpretation. In particular, we write $A\in X$ to denote that the
%attribute $A$ occurs in $X$, we write $X\subseteq Y$ to denote that
%every attribute in $X$ occurs in $Y$, and we write $Y\setminus X$ to
%denote the sequence that is obtained from $Y$ by removing every
%(occurrence of an) attribute in $X$.

We assume a countably infinite domain $\dom$ of attribute values.  By
a \e{tuple} we mean a sequence of values in $\dom$.  A \e{table} $T$
over $R(A_1,\dots,A_k)$ has a collection $\ids(T)$ of (tuple)
identifiers and it maps every identifier $i$ to a tuple in $\dom^k$
and a positive weight; we denote this tuple by $T[i]$ and this weight
by $w_T(i)$. For $i\in\ids(T)$ we refer to $T[i]$ as a tuple \e{of}
$T$. We denote by $T[*]$ the set of all tuples of $T$.  We say that
$T$ is:
\begin{itemize}
\item \e{duplicate free} if distinct tuples disagree on at least one
  attribute, that is, $T[i]\neq T[j]$ whenever $i\neq j$;
\item \e{unweighted} if all tuple weights are equal, that is,
  $w_T(i)=w_T(j)$ for all identifiers $i$ and $j$.
\end{itemize}
We use $|T|$ to denote the number of tuple identifiers of $T$, that
is, $|T|\eqdef|\ids(T)|$. Let $\tup t=(a_1,\dots,a_k)$ be a tuple of
$T$. We use $\tup t.A_j$ to refer to the value $a_j$. If $X=A_{i_1},\dots,
A_{i_\ell}$ is a sequence of attributes in $\set{A_1,\dots,A_k}$, then
$\tup t[X]$ denotes the tuple $(\tup t.A_{i_1}.\dots,\tup t.A_{i_1})$.

\begin{example}\label{eg:running}
  Our running example is around the tables of
  Figure~\ref{fig:example}.  The figure shows tables over the schema
  $\rel{Office}(\att{facility},\att{room},\att{floor},\att{city})$,
  describing the location of offices in an organization. For example,
  the tuple $T[1]$ corresponds to an office in room 322, in the third
  floor of the headquarters (HQ) building, located in Paris.  The
  meaning of the yellow background color will be clarified later.  The
  identifier of each tuple is shown on the leftmost (gray shaded)
  column, and its weight on the rightmost column (also gray shaded).
  Note that table $S_2$ is duplicate free and unweighted, table $S_1$
  is duplicate free but not unweighted, and table $U_2$ is neither
  duplicate free nor unweighted.
\end{example}

\subsection{Functional Dependencies (FDs)}
Let $R(A_1,\dots,A_k)$ be a schema. As usual, an FD (over $R$) is an
expression of the form $X\rightarrow Y$ where $X$ and $Y$ are
sequences of attributes of $R$. We refer to $X$ as the \e{left-hand
  side}, or \e{lhs} for short, and to $Y$ as the \e{right-hand side},
of \e{rhs} for short. A table $T$ \e{satisfies}  $X\ra Y$ if
every two tuples that agree on $X$ also agree on $Y$; that is, for all
$\tup t,\tup s\in T[*]$, if $\tup t[X]=\tup s[X]$ then
$\tup t[Y]=\tup s[Y]$. We say that $T$ \e{satisfies} a set $\depset$
of FDs if $T$ satisfies each FD in $\depset$; otherwise, 
$T$ \emph{violates} $\depset$.

{
\def\mlin{\cline{2-5}}
\def\yellowval#1{\multicolumn{1}{>{\columncolor{yellow}}c|}{\val{#1}}}

\def\officebegin{
\begin{tabular}{>{\columncolor{lightgray}}c|c|c|c|c|>{\columncolor{lightgray}}c} 
%\multicolumn{6}{c}{$\rel{Offices}$}\\
\mlin
\multicolumn{1}{c|}{\textit{id}} & $\att{facility}$ & $\att{room}$ & $\att{floor}$ & $\att{city}$ & \multicolumn{1}{c}{$w$}\\\hline\hline
}
\def\officeend{
\mlin
\end{tabular}
}

\begin{figure}[t!]
\small
\renewcommand{\arraystretch}{1.1}
\centering
\subfigure[Table $T$\label{fig:T}]{
\officebegin
1 & \val{HQ} &  \val{322} & \val{3} & \val{Paris} & 2\\
2 & \val{HQ} &  \val{322} & \val{30} & \val{Madrid} & 1\\
3 & \val{HQ} &  \val{122} & \val{1} & \val{Madrid} & 1\\
4 & \val{Lab1} &  \val{B35} & \val{3} & \val{London} & 2\\
\officeend
}
\subfigure[Consistent subset $S_1$\label{fig:S1}]{
\officebegin
%1 & \val{HQ} &  \val{322} & \val{3} & \val{Paris} & 2\\
2 & \val{HQ} &  \val{322} & \val{30} & \val{Madrid} & 1\\
3 & \val{HQ} &  \val{122} & \val{1} & \val{Madrid} & 1\\
4 & \val{Lab1} &  \val{B35} & \val{3} & \val{London} & 2\\\mlin
\end{tabular}
}
\subfigure[Consistent subset $S_2$\label{fig:S2}]{
\officebegin
1 & \val{HQ} &  \val{322} & \val{3} & \val{Paris} & 2\\
%2 & \val{HQ} &  \val{322} & \val{30} & \val{Madrid} & 1\\
%3 & \val{HQ} &  \val{122} & \val{1} & \val{Madrid} & 1\\
4 & \val{Lab1} &  \val{B35} & \val{3} & \val{London} & 2\\
\officeend
}
\subfigure[Consistent subset $S_3$\label{fig:S3}]{
\officebegin
%1 & \val{HQ} &  \val{322} & \val{3} & \val{Paris} & 2\\
%2 & \val{HQ} &  \val{322} & \val{30} & \val{Madrid} & 1\\
3 & \val{HQ} &  \val{122} & \val{1} & \val{Madrid} & 1\\
4 & \val{Lab1} &  \val{B35} & \val{3} & \val{London} & 2\\
\officeend
}
\subfigure[Consistent update $U_1$\label{fig:U1}]{
\officebegin
1 & \yellowval{F01} &  \val{322} & \val{3} & \val{Paris} & 2\\
2 & \val{HQ} &  \val{322} & \val{30} & \val{Madrid} & 1\\
3 & \val{HQ} &  \val{122} & \val{1} & \val{Madrid} & 1\\
4 & \val{Lab1} &  \val{B35} & \val{3} & \val{London} & 2\\
\officeend
}
\subfigure[Consistent update $U_2$\label{fig:U2}]{
\officebegin
1 & \val{HQ} &  \val{322} & \val{3} & \val{Paris} & 2\\
2 & \val{HQ} &  \val{322} & \yellowval{3} & \yellowval{Paris} & 1\\
3 & \val{HQ} &  \val{122} & \val{1} & \yellowval{Paris} & 1\\
4 & \val{Lab1} &  \val{B35} & \val{3} & \val{London} & 2\\
\officeend
}
\subfigure[Consistent update $U_3$\label{fig:U3}]{
\officebegin
1 & \val{HQ} &  \val{322} & \yellowval{30} & \yellowval{Madrid} & 2\\
2 & \val{HQ} &  \val{322} & \val{30} & \val{Madrid} & 1\\
3 & \val{HQ} &  \val{122} & \val{1} & \val{Madrid} & 1\\
4 & \val{Lab1} &  \val{B35} & \val{3} & \val{London} & 2\\
\officeend
}
\caption{\label{fig:example} For
  $\rel{Office}(\att{facility},\att{room},\att{floor},\att{city})$ and
  FDs $\att{facility}\ra\att{city}$ and
  $\att{facility room}\ra\att{floor}$, a table $T$, consistent subsets
  $S_1$, $S_2$ and $S_3$, and consistent updates $U_1$, $U_2$ and
  $U_3$.  Changed values are marked in yellow.  }
\end{figure}
}

An FD $X \rightarrow Y$ is \emph{entailed by} $\depset$, denoted 
$\depset \models X\rightarrow Y$, if every table $T$ that satisfies
$\Delta$ also satisfies the FD $X \rightarrow Y$. The \emph{closure} of
$\depset$, denoted $\closure(\depset)$, is the set of all FDs over $R$
that are entailed by $\depset$. The \e{closure} of an attribute set
$X$ (\e{w.r.t.~$\depset$}), denoted $\closure_{\depset}(X)$, is the set of all attributes $A$
such that the FD $X\rightarrow A$ is entailed by $\depset$. Two sets
$\depset_1$ and $\depset_2$ of FDs are \e{equivalent} if they have the
same closure (or in other words, each FD in $\depset_1$ is entailed by
$\depset_2$ and vice versa, or put differently, every table that
satisfies one also satisfies the other). An FD $X \rightarrow Y$ is
\emph{trivial} if $Y\subseteq X$; otherwise, it is \emph{nontrivial}.
Note that a trivial FD belongs to the closure of every set of FDs
(including the empty one). We say that $\depset$ is \e{trivial} if
$\depset$ does not contain any nontrivial FDs (e.g., it is empty);
otherwise, $\depset$ is \e{nontrivial}.

Next, we give some non-standard notation that we need for this paper.
A \e{common lhs} of an FD set $\depset$ is an attribute $A$ such that
$A\in X$ for all FDs $X\ra Y$ in $\depset$.  An FD set $\depset$ is a
\emph{chain} if for every two FDs $X_1\rightarrow Y_1$ and
$X_2 \rightarrow Y_2$ it is the case that $X_1 \subseteq X_2$ or
$X_2 \subseteq X_1$. Livshits and
Kimelfeld~\cite{DBLP:conf/pods/LivshitsK17} proved that the class of
chain FD sets consists of precisely the FD sets in which the \e{subset
  repairs}, which we define in
Section~\ref{sec:preliminaries:repairs}, can be counted in polynomial
time (assuming $\mbox{P}\neq\mbox{\#P}$). The chain FD sets will arise in this work as
well.

\begin{example}\label{eg:running-FD}
  In our running example (Figure~\ref{fig:example}) the set $\depset$
  consists of the following FDs:
\begin{itemize}
\item $\att{facility}\ra\att{city}$: a facility belongs to a
  single city.  
\item $\att{facility room}\ra\att{floor}$: a room in a
  facility does not go beyond one floor.
\end{itemize}
Note that the FDs allow for the same room number to occur in different
facilities (possibly on different floors, in different cities).  The
attribute $\att{facility}$ is a common lhs. Moreover, $\depset$ is a
chain FD set, since
$\set{\att{facility}}\subseteq\set{\att{facility},\att{room}}$.  Table
$T$ (Figure~\ref{fig:T}) violates $\depset$, and the other tables
(Figures~\ref{fig:S1}--\ref{fig:U3}) satisfy $\depset$.
\end{example}

An FD $X\ra Y$ might be such that $X$ is empty, and then we denote it
by $\emptyset\ra Y$ and call it a \e{consensus} FD.  Satisfying the
consensus FD $\emptyset \rightarrow Y$ means that all tuples agree on
$Y$, or in other words, the column that corresponds to each attribute
in $Y$ consists of copies of the same value. For example,
$\emptyset\ra\att{city}$ means that all tuples have the same city.  A
\e{consensus attribute} (\e{of $\depset$}) is an attribute in
$\closure_{\depset}(\emptyset)$, that is, an attribute $A$ such that
$\emptyset\ra A$ is implied by $\depset$. We say that $\depset$ is
\e{consensus free} if it has no consensus attributes.

\cut{ \sudeepa{adding these two defns:\\}\red{adding fixed attribute
    -- please check the definitions}.  Given a set of FDs $\Delta$
  over $R(A_1, \cdots, A_k)$, if $\Delta$ contains a FD of the form
  $\emptyset \rightarrow Z$ for $Z \subseteq \{A_1, \cdots, A_k\}$,
  then the attributes in $X$ are called \emph{fixed attributes}, since
  for all tuples in any table $T$ of $R$, the values of attributes in
  $X$ have to be the same.  The standard \emph{closure} $\emptyset^+$
  of $\emptyset$ gives the set of all fixed attributes in $R$.  For a
  set of FDs $\Delta$ over $R(A_1, \cdots, A_k)$ that \emph{does not
    contain} any FD of the form $\emptyset \rightarrow Z$ for
  $Z \subseteq \{A_1, \cdots, A_k\}$, we define \emph{attribute cover}
  of $\Delta$ as a subset $C \subseteq \{A_1, \cdots, A_k\}$ such that
  for every FD $X \rightarrow Y$ in $\Delta$,
  $X \cap C \neq \emptyset$. The \emph{minimum attribute cover} of
  $\Delta$ is a cover $C^*$ of minimum size, \ie, for all cover $C$ of
  $\Delta$, $|C^*| \leq |C|$. We denote the size of minimum attribute
  cover of $\Delta$ by $\mc(\Delta)$. In Example~\ref{eg:running-FD},
  $\{facility\}$ is the minimum attribute cover with size
  $\mc(\Delta) = 1$. Clearly, for a set of FDs $\Delta$ over
  $R(A_1, \cdots, A_k)$, $\mc(\Delta) \leq k$.
\par
}

\subsection{Repairs}\label{sec:preliminaries:repairs}
Let $R(A_1,\dots,A_k)$ be a schema, and let $T$ be a table. A
\e{subset} of $T$ is a table $S$ that is obtained from $T$ by
eliminating tuples. More formally, table $S$ is a subset of $T$ if
$\ids(S)\subseteq\ids(T)$ and for all $i\in\ids(S)$ we have
$S[i]=T[i]$ and $w_S(i)=w_T(i)$. If $S$ is a subset of $T$, then the
\e{distance} from $S$ to $T$, denoted $\dists(S,T)$, is the weighted
sum of the tuples missing from $S$; that is,
\[\dists(S,T)\eqdef\sum_{i\in\ids(T)\setminus\ids(S)}\hspace{-2.2em}w_T(i)\,.\]
A \e{value update} of $T$ (or just \e{update} of $T$ for short) is a
table $U$ that is obtained from $T$ by changing attribute values. More
formally, a table $U$ is an update of $T$ if $\ids(U)=\ids(T)$ and for
all $i\in\ids(U)$ we have $w_U(i)=w_T(i)$.  We adopt the definition of
Kolahi and Lakshmanan~\cite{DBLP:conf/icdt/KolahiL09} for the distance
from $U$ to $T$. Specifically, if $\tup u$ and $\tup t$ are tuples of
tables over $R$, then the \e{Hamming distance} $H(\tup u,\tup t)$ is
the number of attributes in which $\tup u$ and $\tup t$ disagree, that
is, $H(\tup u,\tup t)=|\set{j\mid \tup u.A_j\neq\tup t.A_j}|$.  If $U$
is an update of $T$ then the \e{distance} from $U$ to $T$, denoted
$\dists(U,T)$, is the weighted Hamming distance between $U$ and $T$
(where every changed value counts as the weight of the tuple); that
is,
\[\distu(U,T)\eqdef\sum_{i\in\ids(T)}w_T(i)\cdot H(T[i],U[i])\,.\]

Let $R(A_1,\dots,A_k)$ be a schema, let $T$ be table, and let
$\depset$ be a set of FDs. A \e{consistent subset} (\e{of $T$
  w.r.t.~$\depset$}) is a subset $S$ of $T$ such that
$S\models\depset$, and a \e{consistent update} (\e{of $T$
  w.r.t.~$\depset$}) is an update $U$ of $T$ such that
$U\models\depset$.  A \e{subset repair}, or just \e{S-repair} for
short, is a consistent subset that is not strictly contained in any
other consistent subset. An \e{update repair}, or just \e{U-repair}
for short, is a consistent update that becomes inconsistent if any set
of updated values is restored to the original values in $T$.  An
\e{optimal subset repair} of $T$, or just \e{optimal S-repair} for
short, is a consistent subset $S$ of $T$ such that $\dists(S,T)$ is
minimal among all consistent subsets of $T$.  Similarly, an \e{optimal
  update repair} of $T$, or just \e{optimal U-repair} for short, is a
consistent update $U$ of $T$ such that $\distu(U,T)$ is minimal among
all consistent updates of $T$.  When there is risk of ambiguity, we
may stress that the optimal S-repair (or U-repair) is 
\e{of} $T$ and 
\e{under
  $\depset$} or \e{under $R$ and $\depset$}.

Every (S- or U-) optimal repair is a repair, but not necessarily vice
versa. Clearly, a consistent subset
(respectively, update) can be transformed into a (not necessarily optimal) S-repair
(respectively, U-repair), with no increase of distance, in polynomial
time. In fact, we do not really need the concept of a repair per se,
and the definition is given mainly for compatibility with the
literature (e.g.,~\cite{DBLP:conf/icdt/AfratiK09}). Therefore, unless
explicitly stated otherwise, we do not distinguish between an S-repair
and a consistent subset, and between a U-repair and a consistent
update.

We also define \e{approximations} of optimal repairs in the obvious
ways, as follows. For a number $\alpha\geq 1$, an $\alpha$-optimal
S-repair is an S-repair $S$ of $T$ such that
$\dists(S,T)\leq\alpha\dists(S',T)$ for all S-repairs $S'$ of $T$, and
an $\alpha$-optimal U-repair is a U-repair $U$ of $T$ such that
$\distu(U,T)\leq\alpha\distu(U',T)$ for all U-repairs $U'$ of
$T$. In particular, an optimal S-repair (resp., optimal U-repair) is
the same as a $1$-optimal S-repair (resp., $1$-optimal U-repair).

\begin{example}
  In our running example (Figure~\ref{fig:example}), tables $S_1$,
  $S_2$ and $S_3$ are consistent subsets, and $U_1$, $U_2$ and $U_3$
  are consistent updates. For clarity, we marked with yellow shading
  the values that were changed for constructing each $U_i$.  We have
  $\dists(S_1,T)=2$ since the missing tuple (tuple $1$) has the weight
  $2$. We also have $\dists(S_2,T)=2$ and $\dists(S_3,T)=3$. The
  reader can verify that $S_1$ and $S_2$ are optimal
  S-repairs. However, $S_3$ is not an optimal S-repair since its
  distance to $T$ is greater than the minimum. Nevertheless, $S_3$ is
  an $1.5$-optimal S-repair (since $3/2=1.5$). Similarly, we have
  $\distu(U_1,T)=2$, $\distu(U_2,T)=3$, and $\distu(U_3,T)=4$ (since
  $U_3$ is obtained by changing two values from a tuple of weight
  $2$).
\end{example}

It should be noted that the values of an update $U$ of a table $T$ are
not necessarily taken from the \e{active domain} (i.e., values that
occur in $T$). An example is the value $\val{F01}$ of table $U_1$ in
Figure~\ref{fig:U1}. This has implications on the complexity of
computing optimal U-repairs. We discuss a restriction on the allowed
update values in Section~\ref{sec:conclusions}.

\subsection{Complexity}
We adopt the conventional measure of \e{data complexity}, where the
schema $R(A_1,\dots,A_k)$ and dependency set $\depset$ are assumed to
be \e{fixed}, and only the table $T$ is considered \e{input}. In
particular, a ``polynomial'' running time may have an exponential
dependency on $k$, as in $O(|T|^k)$, Hence, each combination of
$R(A_1,\dots,A_k)$ and $\depset$ defines a distinct problem of finding
an optimal repair (of the relevant type), and different combinations
may feature different computational complexities.

For the complexity of approximation, we use the following terminology.
In an optimization problem $P$, each input $x$ has a space of
solutions $y$, each associated with a cost $\cost(x,y)$. Given $x$,
the goal is to compute a solution $y$ with a minimum cost. For
$\alpha\geq 1$, an \e{$\alpha$-approximation} for $P$ is an algorithm
that, for input $x$, produces an \e{$\alpha$-optimal} solution $y$,
which means that $\cost(x,y)\leq \alpha\cdot\cost(x,y')$ for all
solutions $y'$. The complexity class \e{APX} consists of all
optimization problems that have a polynomial-time constant-factor
approximation.  A polynomial-time reduction $f$ from an optimization
problem $Q$ to an optimization problem $P$ is a \e{strict reduction}
if for all $\alpha\geq 1$, any $\alpha$-optimal solution for $f(x)$
can be transformed in polynomial time into an $\alpha$-optimal
solution for $x$~\cite{DBLP:journals/jcss/Krentel88}; it is a \e{PTAS}
(Polynomial-Time Approximation Scheme)
reduction if for all $\alpha>1$ there exists $\beta_\alpha>1$ such
that any $\beta_\alpha$-optimal solution for $f(x)$ can be transformed
in polynomial time into an $\alpha$-optimal solution for $x$.  A
strict reduction is also a PTAS reduction, but not necessarily vice
versa. A problem $P$ is \e{APX-hard} if there is a PTAS reduction to
$P$ from every problem in APX; it is \e{APX-complete} if, in addition,
it is in APX.  If $P$ is APX-hard, then there is a constant
$\alpha_P>1$ such that $P$ cannot be approximated better than
$\alpha_P$, or else P$=$NP.

\def\prob{\mathrm{Pr}}
\def\OSR{OptSRepair\xspace}
\def\top{_{\mathrm{max}}}
\def\marriageplus{\depset_{A\leftrightarrow B\ra C}}
\def\expl#1{\mbox{(#1)}\Rrightarrow}

\section{Computing an Optimal S-Repair}\label{sec:subset-repairs}
In this section, we study the problem of computing an optimal
S-repair. We begin with some conventions.

\paragraph*{Assumptions and Notation}
Throughout this section we assume that every FD has a single attribute
on its right-hand side, that is, it has the form $X\rightarrow
A$. Clearly, this is not a limiting assumption, since replacing $X\ra
YZ$ with $X\ra Y$ and $X\ra Z$  preserves equivalence.

Let $\depset$ be a set of FDs.  If $X$ is a set of attributes, then we
denote by $\depset-X$ the set $\depset'$ of FDs that is obtained from
$\depset$ by removing each attribute of $X$ from every lhs and rhs of
every FD in $\depset$. Hence, no attribute in $X$ occurs in
$\depset-X$. If $A$ is an attribute, then we may write $\depset-A$
instead of $\depset-\set{A}$.

An \e{lhs marriage} of an FD set $\depset$ is a pair $(X_1,X_2)$ of
distinct lhs of FDs in $\depset$ with the following properties.
\begin{itemize}
\item $\closure_{\depset}(X_1)=\closure_{\depset}(X_2)$
\item The lhs of every FD in $\depset$ contains either $X_1$ or $X_2$ (or both).
\end{itemize}

\begin{example}\label{example:common-and-marriage}
  A simple example of an FD set with an lhs marriage is the following
  FD set.
\begin{equation}\label{eq:marriagefd}
\marriageplus\eqdef\set{A\ra B\,,\,B\ra A\,,\,B\ra C} 
\end{equation}
As another example, consider the following FD set.
\begin{align*}
\depset_1\eqdef\{&\att{ssn}\ra\att{first}\,,\,\att{ssn}\ra\att{last}\,,\,\att{first last}\ra\att{ssn}\,,\,
\att{ssn}\ra\\&\att{address}\,,\,\att{ssn office}\ra\att{phone}\,,\,\att{ssn office}\ra\att{fax} \}
\end{align*}
In $\depset_1$ the pair
$(\set{\att{ssn}},\set{\att{first},\att{last}})$ is an lhs marriage.
%Note that a common lhs (and an lhs marriage) of $\depset$ continues to
%be so even if we add to the schema new attributes that do not
%participate in $\depset$, such as \att{OwnedCar} (assuming a person
%can have more than one car, which, in turn, can be owned by others as
%well).
\end{example}

{
\def\algsep{\vskip1.5em}
\begin{figure}[t]
\begin{malgorithm}{$\algname{\OSR}(\depset,T)$}{alg:osr}
\If{$\depset$ is trivial} \mComment{successful termination}
\State \textbf{return} $T$
\EndIf
\State remove trivial FDs from $\depset$ 
\If{$\depset$ has a common lhs}
\State \textbf{return} $\algname{CommonLHSRep}(\depset,T)$\label{line:suc:common}
\EndIf
\If{$\depset$ has a consensus FD}
\State \textbf{return} $\algname{ConsensusRep}(\depset,T)$\label{line:suc:consensus}
\EndIf
\If{$\depset$ has an lhs marriage}
\State \textbf{return} $\algname{MarriageRep}(\depset,T)$\label{line:suc:marriage}
\EndIf
\State \textbf{fail} \mComment{cannot find a minimum repair}
\end{malgorithm}
\algsep
\begin{msubroutine}{$\algname{CommonLHSRep}(\depset,T)$}{alg:commonlhs}
\State $A\asn$ a common lhs of $\depset$
\State \Return $\cup_{(a)\in\pi_AT[*]}\algname{\OSR}(\sigma_{A=a}T,\depset-A)$
\end{msubroutine}
\algsep
\begin{msubroutine}{$\algname{ConsensusRep}(\depset,T)$}{alg:consensus}
\State select a consensus FD $\emptyset\ra A$ in $\depset$
\ForAll{$a\in\pi_AT[*]$}
\State $S_{a}\asn\algname{\OSR}(\sigma_{A=a}T,\depset-A)$
\EndFor
\State $a\top\asn\argmax_{(a)\in\pi_AT[*]}w_T(S_{a})$
\State \Return $S_{a\top}$
\end{msubroutine}
\algsep
\begin{msubroutine}{$\algname{MarriageRep}(\depset,T)$}{alg:marriage}
\State select an lhs marriage $(X_1,X_2)$ of $\depset$
\ForAll{$(\tup a_1,\tup a_2)\in\pi_{X_1X_2}T[*]$}
\State $S_{\tup a_1,\tup a_2}\asn\algname{\OSR}(\sigma_{X_1=\tup a_1,X_2=\tup a_2}T,\depset-X_1X_2)$
\State $w(\tup a_1,\tup a_2)\asn w_T(S_{\tup a_1,\tup a_2})$
\EndFor 
\State $V_i\asn\pi_{X_i}T[*]$ for $i=1,2$
\State $E\asn\set{(\tup a_1,\tup a_2)\mid (\tup a_1,\tup a_2)\in\pi_{X_1X_2}T[*]}$
\State $G\asn$ weighted bipartite graph $(V_1,V_2,E,w)$
\State $E\top\asn$ a maximum matching of $G$
\State \Return $\cup_{(\tup a_1,\tup a_2)\in E\top}S_{\tup a_1,\tup a_2}$
\end{msubroutine}
\end{figure}
}

\begin{figure}
\begin{malgorithm}{$\algname{OSRSucceeds}(\depset)$}{alg:succeeds}
\While{$\depset$ is nontrivial}
\State remove trivial FDs from $\depset$ 
\If{$\depset$ has a common lhs $A$}
\State $\depset\asn\depset-A$
\ElsIf{$\depset$ has a consensus FD $\emptyset\ra A$}
\State $\depset\asn\depset-X$
\ElsIf{$\depset$ has an lhs marriage $(X_1,X_2)$}
\State $\depset\asn\depset-X_1X_2$
\Else 
\State \Return $\false$
\EndIf
\EndWhile
\State \Return $\true$
\end{malgorithm}
\end{figure}

Finally, if $S$ is a subset of a table $T$, then we denote by $w_T(S)$
the sum of weights of the tuples of $S$, that is,
\[w_T(S)\eqdef \sum_{i\in\ids(S)}w_T(i)\]

\subsection{Algorithm} 
We now describe an algorithm for finding an optimal S-repair. The
algorithm terminates in polynomial time, even under combined
complexity, yet it may \e{fail}. If it succeeds, then the result is
guaranteed to be an optimal S-repair. We later discuss the situations
in which the algorithm fails. The algorithm, $\algname{\OSR}$, is
shown as Algorithm~\ref{alg:osr}. The input is a set $\depset$ of FDs
and a table $T$, both over the same relation schema (that we do not
need to refer to explicitly). In the remainder of this section, we fix
$\depset$ and $T$, and describe the execution of $\algname{\OSR}$ on
$\depset$ and $T$.

The algorithm handles four cases. The first is where $\depset$ is
trivial. Then, $T$ is itself an optimal S-repair. The second case is
where $\depset$ has a common lhs $A$. Then, the algorithm groups the
tuples by $A$, finds an optimal S-repair for each group (via a
recursive call to $\algname{OptSRepair}$), this time by ignoring $A$
(i.e., removing $A$ from the FDs of $\depset$), and returning the
union of the optimal S-repairs. The precise description is in the
subroutine $\algname{CommonLHSRep}$ (Subroutine~\ref{alg:commonlhs}).
The third case is where $\depset$ has a consensus FD $\emptyset\ra
A$. Similarly to the second case, the algorithm groups the tuples by
$A$ and finds an optimal S-repair for each group. This time, however,
the algorithm returns the optimal S-repair with the maximal weight.
The precise description is in the subroutine $\algname{ConsensusRep}$
(Subroutine~\ref{alg:consensus}).

The fourth (last) case is the most involved. This is the case where
$\depset$ has an lhs marriage $(X_1,X_2)$. In this case the problem is
reduced to finding a maximum weighted matching of a bipartite graph.
The graph, which we denote by $G=(V_1,V_2,E,w)$, consists of two
disjoint node sets $V_1$ and $V_2$, an edge set $E$ that connects
nodes from $V_1$ to nodes from $V_2$, and a weight function $w$ that
assigns a weight $w(v_1,v_2)$ to each edge $(v_1,v_2)$. For $i=1,2$,
the node set $V_i$ is the set of tuples in the projection of $T$ to
$X_i$.\footnote{In principle, it may be the case that the same tuple
  occurs in both $V_1$ and $V_2$, since the tuple is in both
  projections.  Nevertheless, we still treat the two occurrences of
  the tuple as distinct nodes, and so effectively assume that $V_1$
  and $V_2$ are disjoint.} To determine the weight $w(v_1,v_2)$, we
select from $T$ the subset $T_{v_1,v_2}$ that consists of the tuples
that agree with $v_1$ and $v_2$ on $X_1$ and $X_2$, respectively. We
then find an optimal S-repair for $T_{v_1,v_2}$, after we remove from
$\depset$ every attribute in either $X_1$ or $X_2$. Then, the weight
$w(v_1,v_2)$ is the weight of this optimal S-repair. Next, we find a
maximum matching $E\top$ of $G$. Note that $E\top$ is a subset of $E$
such that no node appears more than once. The returned result is then
the disjoint union of the optimal S-repair of $T_{v_1,v_2}$ over all $(v_1,v_2)$ in
$E\top$.  The precise description is in the subroutine
$\algname{MarriageRep}$ (Subroutine~\ref{alg:marriage}).

The following theorem states the correctness and efficiency of
$\algname{\OSR}$.

\def\thmosr{ Let $\depset$ and $T$ be a set of FDs and a table,
  respectively, over a relation schema $R(A_1,\dots,A_k)$.  If
  $\algname{\OSR}(\depset,T)$ succeeds, then it returns an optimal
  S-repair. Moreover, $\algname{\OSR}(\depset,T)$ terminates in
  polynomial time in $k$, $|\depset|$, and $|T|$.  }
\begin{theorem}\label{thm:osr}
\thmosr
\end{theorem}
What about the cases where $\algname{\OSR}(\depset,T)$ fails? We
discuss it in the next section.

\paragraph*{Approximation}
An easy observation is that the computation of an optimal subset is
easily reducible to the \e{weighted vertex-cover} problem---given a
graph $G$ where nodes are assigned nonnegative weights, find a vertex
cover (i.e., a set $C$ of nodes that intersects with all edges) with a
minimal sum of weights. Indeed, given a table $T$, we construct the
graph $G$ that has $\ids(T)$ as the set of nodes, and an edge between
every $i$ and $j$ such that $T[i]$ and $T[j]$ contradict one or more
FDs in $\depset$. Given a vertex cover $C$ for $G$, we obtain a
consistent subset $S$ by deleting from $T$ every tuple with an
identifier in $C$. Clearly, this reduction is strict. As weighted
vertex cover is 2-approximable in polynomial
time~\cite{DBLP:journals/jal/Bar-YehudaE81}, we conclude the same for
optimal subset repairing.

\begin{proposition}\label{prop:subset-approx}
  For all FD sets $\depset$, a 2-optimal S-repair can be computed in
  polynomial time.
\end{proposition}

While Proposition~\ref{prop:subset-approx} is straightforward, it is
of practical importance as it limits the severity of the lower bounds
we establish in the next section. Moreover, we will later show that
the proposition has implications on the problem of approximating an
optimal U-repair.

\subsection{Dichotomy}
The reader can observe that the success or failure of
$\algname{\OSR}(\depset,T)$ depends only on $\depset$, and not on
$T$. The algorithm $\algname{OSRSucceeds}(\depset)$, depicted as
Algorithm~\ref{alg:succeeds}, tests whether $\depset$ is such that
$\algname{\OSR}$ succeeds by simulating the cases and corresponding
changes to $\depset$. The next theorem shows that, under conventional
complexity assumptions, $\algname{\OSR}$ covers \e{all} sets $\depset$
such that an optimal S-repair can be found in polynomial time. Hence,
we establish a dichotomy in the complexity of computing an optimal
S-repair.

\def\thmdichotomy{
Let $\depset$ be a set of FDs.
\begin{itemize}
\item If $\algname{OSRSucceeds}(\depset)$ returns true, then an optimal
  S-repair can be computed in polynomial time by executing
  $\algname{\OSR}(\depset,T)$ on the input $T$.
\item If $\algname{OSRSucceeds}(\depset)$ returns false, then
  computing an optimal S-repair is APX-complete, and remains
  APX-complete on unweighted, duplicate-free tables.
\end{itemize}
Moreover, the execution of $\algname{OSRSucceeds}(\depset)$ terminates in polynomial
time in $|\depset|$.  }
\begin{theorem}\label{thm:dichotomy}
\thmdichotomy
\end{theorem}
Recall that a problem in APX has a constant factor approximation and,
under the assumption that P$\neq$NP, an APX-hard problem cannot be
approximated better than some constant factor (that may depend on the
problem itself). 

\begin{example}\label{example:dichotomy}
  We now illustrate the application of Theorem~\ref{thm:dichotomy} to
  several FD sets. Consider first the FD set $\depset$ of our running
  example. The execution of $\algname{OSRSucceeds}(\depset)$ transforms
  $\depset$ as follows.
  \begin{align*}
   &\set{\att{facility}\ra\att{city}\,,\,\att{facility
    room}\ra\att{floor}} \\
   \expl{common lhs} &\set{\emptyset\ra\att{city}\,,\,\att{room}\ra\att{floor}}\\
   \expl{consensus}  &\set{\att{room}\ra\att{floor}}\\
   \expl{common lhs} &\set{\emptyset\ra\att{floor}}\\
   \expl{consensus}  &\set{} 
        \end{align*}
        Hence, $\algname{OSRSucceeds}(\depset)$ is true, and hence, an
        optimal S-repair can be found in polynomial time.
   
        Next, consider the FD set $\marriageplus$ from
        Example~\ref{example:common-and-marriage}. 
        $\algname{OSRSucceeds}(\marriageplus)$ executes as follows.
  \begin{align*}
                        &\set{A\ra B,B\ra A,B\ra C}\\
    \expl{lhs marriage} &\set{\emptyset\ra C} \\
    \expl{consensus} & \set{}
        \end{align*}
        Hence, this is again an example of an FD set on the tractable
        side of the dichotomy.   

        As the last positive example we consider the FD set
        $\depset_1$ of Example~\ref{example:common-and-marriage}.
{\small
    \begin{align*}
       &\{\att{ssn}\ra\att{first}\,,\,\att{ssn}\ra\att{last}\,,\,\att{first last}\ra\att{ssn}\,,\,
       \att{ssn}\ra\att{address}\,,\,
       \\&\att{ssn office}\ra\att{phone}\,,\,\att{ssn office}\ra\att{fax}\}\\
       &\expl{lhs marriage}
       \{\emptyset\ra\att{address}\,,\,\att{office}\ra\att{phone}\,,\,\att{office}\ra\att{fax} \}\\
       &\expl{consensus}\{\att{office}\ra\att{phone}\,,\,\att{office}\ra\att{fax} \}\\
       &\expl{common lhs}\{\emptyset\ra\att{phone}\,,\,\emptyset\ra\att{fax} \}\\
       &\expl{consensus} \set{}
 \end{align*}
}
 On the other hand, for $\depset=\set{A\ra B,B\ra C}$, none of the
 conditions of $\algname{OSRSucceeds}(\depset)$ is true, and therefore,
 the algorithm returns false. It thus follows from
 Theorem~\ref{thm:dichotomy} that computing an optimal S-repair is
 APX-complete (even if all tuple weights are the same and there are no
 duplicate tuples). The same applies to $\depset=\set{A\ra B,C\ra D}$.
\end{example}

%\sudeepa{need a proposition on chain FDs here to be referred to in the U-repair section.}
As another example, the following corollary of
Theorem~\ref{thm:dichotomy} generalizes the tractability of our
running example to general chain FD sets.
\begin{corollary}\label{cor:subset-chain-tractable}
  If $\depset$ is a chain FD set, then an optimal S-repair can be
  computed in polynomial time.
\end{corollary}
\begin{proof}
  The reader can easily verify that when $\depset$ is a chain FD
  set, $\algname{OSRSucceeds}(\depset)$ will reduce it to emptiness by
  repeatedly removing consensus attributes and common-lhs, as done in
  our running example.
\end{proof}

\subsection{Proof of Theorem~\ref{thm:dichotomy}}

\def\sabc{\depset_{A\rightarrow B\rightarrow C}}
\def\stfd{\depset_{A\rightarrow C\leftarrow B}}
\def\stk{\depset_{AB\rightarrow C\rightarrow B}}
\def\str{\depset_{AB\leftrightarrow AC\leftrightarrow BC}}

In this section we discuss the proof of Theorem~\ref{thm:dichotomy}.
(The full proof is in the
\hyperref[app:proofs-subset-repair]{Appendix}.)  The positive side is
a direct consequence of Theorem~\ref{thm:osr}.  For the negative side,
membership in APX is due to Proposition~\ref{prop:subset-approx}.  The
proof of hardness is based on the concept of a \e{fact-wise
  reduction}~\cite{DBLP:conf/pods/Kimelfeld12}, as previously done for
proving dichotomies on sets of
FDs~\cite{DBLP:conf/icdt/KimelfeldLP17,DBLP:conf/pods/Kimelfeld12,
  DBLP:conf/pods/FaginKK15,DBLP:conf/pods/LivshitsK17}.  In our setup,
a fact-wise reduction is defined as follows.  Let $R$ and $R'$ be two
relation schemas.  A \e{tuple mapping} from $R$ to $R'$ is a function
$\mu$ that maps tuples over $R$ to tuples over $R'$. We extend $\mu$
to map tables $T$ over $R$ to tables over $R'$ by defining $\mu(T)$ to
be $\set{\mu(t)\mid t\in T}$.  Let $\depset$ and $\depset'$ be sets of
FDs over $R$ and $R'$, respectively. A \e{fact-wise reduction} from
$(R,\depset)$ to $(R',\depset)$ is a tuple mapping $\Pi$ from $R$ to
$R'$ with the following properties:
%\begin{enumerate}
%\item 
\e{(a)} $\Pi$ is injective, that is, for all tuples $\tup t_1$ and
$\tup t_2$ over $R$, if $\Pi(\tup t_1) = \Pi(\tup t_2)$ then $\tup t_1
= \tup t_2$; \e{(b)} $\Pi$ preserves consistency and inconsistency;
that is, $\Pi(T)$ satisfies $\depset'$ if and only if $T$ satisfies
$\depset$; \e{and (c)} 
%\item 
$\Pi$ is computable in polynomial time.
%\end{enumerate}
The following lemma is straightforward.
\begin{lemma}\label{lemma:fact-wise-apx}
  Let $R$ and $R'$ be relation schemas and $\depset$ and
  $\depset'$ FD sets over $R$ and $R'$, respectively.  If there is
  a fact-wise reduction from $(R,\depset)$ to $(R',\depset')$, then
  there is a strict reduction from the problem of computing an optimal
  S-repair under $R$ and $\depset$ to that of computing an optimal
  S-repair under $R'$ and $\depset'$.
\end{lemma}

In the remainder of this section, we describe the way we use
Lemma~\ref{lemma:fact-wise-apx}. Our proof consists of four steps.

{
\begin{table}
\caption{FD sets over $R(A,B,C)$ used in the proof of hardness of Theorem~\ref{thm:dichotomy}.\label{table:special-schemas}}
\def\arraystretch{1.4}
\vskip0.5em
\centering
\begin{tabular}{l|l}
\textbf{Name} & \textbf{FDs}\\\hline
$\sabc$ & $A\rightarrow B$, $B\rightarrow C$\\
$\stfd$ & $A\rightarrow C$, $B\rightarrow C$\\
$\stk$ & $AB\rightarrow C$, $C\rightarrow B$\\
$\str$ & $AB\rightarrow C$, $AC\rightarrow B$, $BC\rightarrow A$\\
\hline
\end{tabular}
\end{table}
}

\begin{enumerate}
\item We first prove APX-hardness for each of the FD sets in
  Table~\ref{table:special-schemas} over $R(A,B,C)$. For 
  $\sabc$ and $\stfd$ we adapt reductions by Gribkoff et
  al.~\cite{GVSBUDA14} in a work that we discuss in
  Section~\ref{sec:mpd}. For $\stk$ we show a reduction from
  MAX-non-mixed-SAT~\cite{DBLP:journals/jacm/Hastad01}. Most intricate
  is the proof for $\str$, where we devise a nontrivial adaptation of
  a reduction by Amini et al.~\cite{DBLP:journals/tcs/AminiPS09} to
  triangle packing in graphs of bounded degree.
\item Next, we prove that whenever $\algname{OSRSucceeds}$ simplifies
  $\depset$ into $\depset'$, there is a fact-wise reduction from
  $(R,\depset')$ to $(R,\depset)$, where $R$ is the underlying
  relation schema.
\item Then, we consider an FD set $\depset$ that cannot be further
  simplified (that is, $\depset$ does not have a common lhs, a
  consensus FD, or an lhs marriage). We show that $\depset$ can be
  classified into one of five certain classes of FD sets (that we
  discuss next).
\item Finally, we prove that for each FD set $\depset$ in one of the
  five classes there exists a fact-wise reduction from one of the four
  schemas of Table~\ref{table:special-schemas}.
\end{enumerate}

\def\xobm{\widehat{X_1}}
\def\xtbm{\widehat{X_2}}

%  $X_1$ and $X_2$ are the lhs of the two local minima in
%  $\depset$. $\widehat{X_i}$ for $i\in\set{1,2}$ is the set
%  $\closure_\depset(X_i)\setminus X_i$. Overlapping solid lines
%  represent two sets with a nonempty intersection, lines that do not
%  overlap represent two sets with an empty intersection and
%  overlapping dashed lines represent two sets for which we do not
%  assume anything about their intersection (that is, both cases are
%  possible).

%\captionsetup[subfigure]{position=left}
%\def\xobm{\widehat{X_1}}
%\def\xtbm{\widehat{X_2}}
%\begin{figure}[t]
%\centering
%\def\xo{$X_1$}
%\def\xt{$X_2$}
%\def\xob{$\xobm$}
%\def\xtb{$\xtbm$}
%
%\renewcommand{\thesubfigure}{(\arabic{subfigure})}
%\parbox{1.7in}{(1)\input{case1.pspdftex}\vskip1em}\vrule\,\,
%\parbox{1.5in}{\input{case2.pspdftex} (2)\vskip1em}  
%\hrule\vskip1em
%(3) \input{case3.pspdftex}\vskip1em
%\hrule
%\parbox{1.2in}{\vskip1em(4)\quad\belowbaseline[0pt]{\input{case4.pspdftex}}}\quad\vrule\quad
%\parbox{1.2in}{\vskip1em(5)\quad\belowbaseline[0pt]{\input{case5.pspdftex}}}\vskip1em
%\caption{\label{fig:classes} Classes of FD sets that cannot be
%  simplified. $X_1$ and $X_2$ are the lhs of the two local minima in
%  $\depset$. $\widehat{X_i}$ for $i\in\set{1,2}$ is the set
%  $\closure_\depset(X_i)\setminus X_i$. Overlapping solid lines
%  represent two sets with a nonempty intersection, lines that do not
%  overlap represent two sets with an empty intersection and
%  overlapping dashed lines represent two sets for which we do not
%  assume anything about their intersection (that is, both cases are
%  possible).}
%\end{figure}

The most challenging part of the proof is identifying the classes of
FD sets in Step~3 in such a way that we are able to build the
fact-wise reductions in Step~4. We first identify that if an FD set
$\depset$ cannot be simplified, then there are at least two distinct
local minima $X_1\rightarrow Y_1$ and $X_2\rightarrow Y_2$ in
$\depset$. By a \e{local minimum} we mean an FD with a set-minimal
lhs, that is, an FD $X\rightarrow Y$ such that no FD $Z\rightarrow W$
in $\depset$ satisfies that $Z$ is a strict subset of $X$. We pick any
two local minima from $\depset$. Then, we divide the FD sets into five
classes based on the relationships between $X_1$, $X_2$,
$\closure_\depset(X_1)\setminus X_1$, which we denote by $\xobm$, and
$\closure_\depset(X_2)\setminus X_2$, which we denote by $\xtbm$. The
classes are illustrated in Figure~\ref{fig:classes}.

Each line in Figure~\ref{fig:classes} represents one of the sets
$X_1$, $X_2$, $\xobm$ or $\xtbm$. If two lines do not overlap, it
means that we assume that the corresponding two sets are disjoint. For
example, the sets $\xobm$ and $\xtbm$ in class $(1)$ have an empty
intersection. Overlapping lines represent sets that have a nonempty
intersection, an example being the sets $\xobm$ and $\xtbm$ in class
$(2)$. If two dashed lines overlap, it means that we do not assume
anything about their intersection. As an example, the sets $X_1$ and
$X_2$ can have an empty or a nonempty intersection in each one of the
classes. Finally, if a line covers another line, it means that the set
corresponding to the first line contains the set corresponding to the
second line. For instance, the set $\xtbm$ in class $(4)$ contains the
set $X_1\setminus X_2$, while in class $(5)$ it holds that
$(X_1\setminus X_2)\not\subseteq \xtbm$. We remark that
Figure~\ref{fig:classes} well covers the important cases that we need
to analyze, but it misses a few cases. (As previously said, full
details are in the \hyperref[app:proofs-subset-repair]{Appendix}.)
%presents an
%example for each one of the classes, and does not provide all the
%details, which are given in the appendix.

\captionsetup[subfigure]{position=left}
\begin{figure}
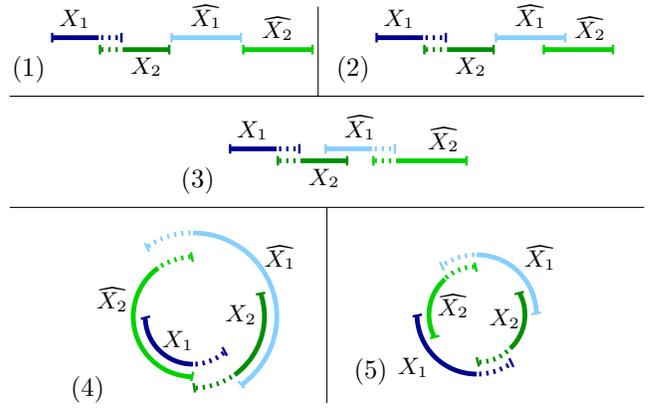

\centering
\def\xo{\small $X_1$}
\def\xt{\small $X_2$}
\def\xob{\small $\widehat{X_1}$}
\def\xtb{\small $\widehat{X_2}$}
\renewcommand{\thesubfigure}{(\arabic{subfigure})}
\parbox{1.6in}{(1)\input{case1.pspdftex}\vskip0.5em}\vrule\,\,
\parbox{1.6in}{(2)\input{case2.pspdftex}\vskip0.5em}  
\hrule\vskip0.8em
(3) \input{case3.pspdftex}\vskip0.5em
\hrule
\parbox{1.2in}{\vskip0.8em(4)\hskip0.8em{\input{case4.pspdftex}}}\quad\vrule\quad
\parbox{1.2in}{\vskip0.8em(5)\hskip0.8em{\input{case5.pspdftex}}}\vskip1em
\caption{\label{fig:classes} Classes of FD sets that cannot be
  simplified.}
\end{figure}

\begin{example}
For each one of the five classes of FD sets from Figure~\ref{fig:classes} we will now give an example of an FD set that belongs to this class.
%\begin{enumerate}
%\item
\partitle{Class 1}
 $\depset_1=\set{A\rightarrow B, C\rightarrow D}$. In this case $X_1=\set{A}$, $X_2=\set{C}$, $\xobm=\set{B}$ and $\xtbm=\set{D}$. Thus, $\xobm\cap X_2=\emptyset$, $\xtbm\cap X_1=\emptyset$ and $\xobm\cap\xtbm=\emptyset$ and indeed the only overlapping lines in $(1)$ are the dashed lines corresponding to $X_1$ and $X_2$.
%\item 
\partitle{Class 2}
$\depset_2=\set{A\rightarrow CD, B\rightarrow CE}$. It holds that $X_1=\set{A}$, $X_2=\set{B}$, $\xobm=\set{C,D}$ and $\xtbm=\set{C,E}$. Hence, $\xobm\cap X_2=\emptyset$ and $\xtbm\cap X_1=\emptyset$, but $\xobm\cap \xtbm\neq\emptyset$, and the difference from $(1)$ is that the lines corresponding to $\xobm$ and $\xtbm$ in $(2)$ overlap.
%\item 
\partitle{Class 3}
$\depset_3=\set{A\rightarrow BC, B\rightarrow D}$. Here, it holds that $X_1=\set{A}$, $X_2=\set{B}$, $\xobm=\set{B,C,D}$ and $\xtbm=\set{D}$. Thus, $\xobm\cap X_2\neq\emptyset$, but $\xtbm\cap X_1=\emptyset$. The difference from $(2)$ is that now the lines corresponding to $X_2$ and $\xobm$ overlap and we do not assume anything about the intersection between $\xobm$ and $\xtbm$.
%\item 
\partitle{Class 4}
$\depset_4=\set{AB\rightarrow C, AC\rightarrow B, BC\rightarrow A}$. In this case we have three local minima. We pick two of them: $AB\rightarrow C$ and $AC\rightarrow B$. Now, $X_1=\set{AB}$, $X_2=\set{AC}$, $\xobm=\set{C}$ and $\xtbm=\set{B}$. Thus, $\xobm\cap X_2\neq\emptyset$ and $\xtbm\cap X_1\neq\emptyset$. The difference from $(3)$ is that now the lines corresponding to $X_1$ and $\xtbm$ overlap. Moreover, the line corresponding to $\xobm$ covers the entire line corresponding to $X_2\setminus X_1$ and the line corresponding to $\xtbm$ covers the entire line corresponding to $X_1\setminus X_2$. This means that we assume that $(X_1\setminus X_2)\subseteq\xtbm$ and $(X_2\setminus X_1)\subseteq\xobm$.
%\item 
\partitle{Class 5}
$\depset_5=\set{AB\rightarrow C, C\rightarrow AD}$. Here, $X_1=\set{A,B}$, $X_2=\set{C}$, $\xobm=\set{C,D}$ and $\xtbm=\set{A,D}$, therefore $\xobm\cap X_2\neq\emptyset$ and $\xtbm\cap X_1\neq\emptyset$. The difference from $(4)$ is that now we assume that $(X_1\setminus X_2)\not\subseteq \xtbm$.
%\end{enumerate}
\end{example}

\subsection{Most Probable Database}\label{sec:mpd}

In this section, we draw a connection to the \e{Most Probable Database
  problem} (MPD)~\cite{GVSBUDA14}. A table in our setting can be
viewed as a relation of a \e{tuple-independent
  database}~\cite{DBLP:conf/vldb/DalviS04} if each weight is in the
interval $[0,1]$. In that case, we view the weight as the probability
of the corresponding tuple, and we call the table a \e{probabilistic
  table}. Such a table $T$ represents a probability space over the
subsets of $T$, where a subset is selected by considering each tuple
$T[i]$ independently and \e{selecting} it with the probability
$w_T(i)$, or equivalently, deleting it with the probability
$1-w_T(i)$. Hence, the probability of a subset $S$, denoted
$\prob_T(S)$, is given by:
\begin{equation}\label{eq:prob}
\prob_T(S)\eqdef\left(\prod_{i\in\ids(S)}\hspace{-0.7em}w_T(i)\right)\times 
\left(\prod_{i\in\ids(T)\setminus\ids(S)}\hspace{-2em}(1-w_T(i))\right)
\end{equation}
Given a constraint $\varphi$ over the schema of $T$, MPD for $\varphi$
is the problem of computing a subset $S$ that satisfies
$\varphi$, and has the maximal probability among all such
subsets. Here, we consider the case where $\varphi$ is a set $\depset$
of FDs. Hence MPD for $\depset$ is the problem of computing
\[\argmax_{\substack{S\subseteq T\,,\, S\models\depset}}\prob_T(S)\,.\]

Gribkoff, Van den Broeck, and Suciu~\cite{GVSBUDA14} proved the
following dichotomy for \e{unary} FDs, which are FDs of the form $A\ra
X$ having a single attribute on their lhs.
\begin{citedtheorem}{GVSBUDA14}\label{thm:buda-dichotomy}
  Let $\depset$ be a set of unary FDs over a relational schema. MPD
  for $\depset$ is either solvable in polynomial time or NP-hard.
\end{citedtheorem}
The question of whether such a dichotomy holds for \e{general} (not
necessarily \e{unary}) FDs has been left open. The following corollary
of Theorem~\ref{thm:dichotomy} fully resolves this question.

\begin{theorem}
  Let $\depset$ be a set of FDs over a relational schema. If
  $\algname{OSRSucceeds}(\depset)$ is true, then MPD for $\depset$ is
  solvable in polynomial time; otherwise, it is NP-hard.
\end{theorem}
\begin{proof}
  We first show a reduction from MPD to the problem of computing an
  optimal S-repair. Let $T$ be an input for MPD.  By a \e{certain
    tuple} we refer to a tuple identifier $i\in\ids(T)$ such that
  $w_T(i)=1$. We assume that the set of certain tuples satisfies $\depset$
  collectively, since otherwise the probability of any consistent
  subset is zero (and we can select, e.g., the empty subset as a most
  likely solution). We can then replace each probability $1$ with a
  probability that is smaller than, yet close enough to $1$, so that
  every consistent subset that excludes a certain fact is less likely
  than any subset that includes all certain facts. In addition, as
  observed by Gribkoff et al.~\cite{GVSBUDA14}, tuples with
  probability at most $0.5$ can be eliminated, since we can always
  remove them from any (consistent) subset without reducing the
  probability. Hence, we assume that $0.5<w_T(i)<1$ for all
  $i\in\ids(T)$. From~\eqref{eq:prob} we conclude the following.
\begin{align*}
\prob_T(S)&=\left(\prod_{i\in\ids(S)}\frac{w_T(i)}{1-w_T(i)}\right)\times 
\left(\prod_{i\in\ids(T)}(1-w_T(i))\right)\\
&\propto \left(\prod_{i\in\ids(S)}\frac{w_T(i)}{1-w_T(i)}\right)
\end{align*}
The reason for the proportionality ($\propto$) is that all consistent
subsets share the same right factor of the first product. Hence, we
construct a table $T'$ that is the same as $T$, except that
$w_{T'}(i)=\log(w_T(i)/(1-w_T(i)))$ for all $i\in\ids(T')$, and then a
most likely database of $T$ is the same\footnote{We do
  not need to make an assumption of infinite precision to work with
  logarithms, since the algorithms we use for computing an optimal
  S-repair can replace addition and subtraction with multiplication
  and division, respectively.} as an optimal S-repair of $T'$.

For the ``otherwise'' part we show a reduction from the problem of
computing an optimal S-repair of an \e{unweighted} table to MPD. The
reduction is straightforward: given $T$, we set
the weight $w_T(i)$ of each tuple to $0.9$ (or any fixed number
greater than $0.5$). From~\eqref{eq:prob} it follows that a consistent
subset is most probable if and only if it has a maximal number of
tuples.
\end{proof}

\begin{comment}
  When considering unary FDs, there is a disagreement between our
  tractability condition (Algorithm~\ref{alg:succeeds}) and that of
  Gribkoff et al.~\cite{GVSBUDA14}. In particular, the FD set
  $\marriageplus$ defined in~\eqref{eq:marriagefd} is classified as
  polynomial time in our dichotomy while NP-hard by Gribkoff et
  al.~\cite{GVSBUDA14}. This is due to a gap in their proof of
  hardness.\footnote{This has been established in a private
    communication with the authors of~\cite{GVSBUDA14}.}
\end{comment}

\section{Computing an Optimal U-Repair}\label{sec:update-repairs}

In this section, we focus on the problem of finding an optimal
U-repair and an approximation thereof. We devise general tools for
analyzing the complexity of this problem, compare it to the problem of
finding an optimal S-repair (discussed in the previous section), and
identify sufficient conditions for efficient reductions between the
two problems. Yet, unlike S-repairs, the existence of a full dichotomy
for computing an optimal U-repair remains an open problem.

\paragraph*{Notation}
%We first introduce some concepts that we use in this section.  The
%\emph{consensus reduction} of a set $\depset$ of FDs, denoted by
%$\conred(\depset)$, is the set
%$\depset-\closure_{\depset}(\emptyset)$; in other words,
%$\conred(\depset)$ is the FD set obtained from $\depset$ by removing
%the consensus attributes from both lhs and rhs of all FDs. 
Let $\depset$ be a set of FDs.  An \emph{lhs cover} of $\depset$ is a
set $C$ of attributes that hits every lhs, that is, $X\cap
C\neq\emptyset$ for every $X \ra Y$ in $\depset$.  We denote the
minimum cardinality of an lhs cover of $\depset$ by
$\mc(\depset)$. For instance, if $\depset$ is nonempty and has a
common lhs (e.g., Figure~\ref{fig:example}), then $\mc(\Delta) = 1$.
For a set $\depset$ of FDs, $\attr(\depset)$ denotes the set of
attributes that appear in $\depset$ (i.e., the union of lhs and
rhs over all the FDs in $\depset$).  Two FD sets $\depset_1$ and
$\depset_2$ (over the same schema) are \emph{attribute disjoint} if
$\attr(\depset_1)$ and $\attr(\depset_2)$ are  disjoint.
%In other words, $\depset_1$ and $\depset_2$ are disjoint if for all
% FDs $X_1 \ra Y_1$ and $X_2\ra Y_2$ in $\depset_1$ and $\depset_2$,
% respectively, we have
% $(X_1\cup Y_1) \cap (X_2 \cup Y_2) = \emptyset$.
For example, $\set{A \ra BC, C \ra D}$ and $\set{E \ra FG}$ are
attribute disjoint.

\subsection{Reductions between FD Sets}
In this section we show two reductions between sets of FDs. The
following theorem implies that to determine the complexity of the
union of two attribute-disjoint FD sets, it suffices to look at each set
separately.

\def\thmdisjoint{
  Suppose that $\depset=\depset_1\cup\depset_2$ where $\depset_1$ and
  $\depset_2$ are attribute disjoint. The following are equivalent for
  all $\alpha\geq 1$.
  \begin{enumerate}
  \item An $\alpha$-optimal U-repair can be computed in polynomial
    time under $\depset$.
  \item An $\alpha$-optimal U-repair can be computed in polynomial
    time under \e{each of} $\depset_1$ and $\depset_2$.
\end{enumerate}
}
\begin{theorem}\label{thm:disjoint}
\thmdisjoint
\end{theorem}

If we have a polynomial-time algorithm to compute an $\alpha$-optimal
solution for $\depset$, it can be used to obtain an $\alpha$-optimal
solution for $\depset_1$ by setting all attributes not in
$\attr(\depset_1)$ to 0, and running the algorithm (similarly for
$\depset_2$). In the reverse direction, it can be shown that simply
composing $\alpha$-optimal solutions of $\depset_1, \depset_2$ gives
an $\alpha$-optimal solution for $\depset$. 
% The full proof of Theorem~\ref{thm:disjoint} is given in the
% Appendix.
\begin{example}\label{example:disjoint-u}
Consider the following set of FDs.
\[\depset\eqdef\set{\att{item}\ra\att{cost}\,,\,\att{buyer}\ra\att{address}}\]
We will later show that if $\depset$ consists of a single FD, then an
optimal U-repair can be computed in polynomial time. Hence, we can
compute an optimal U-repair under
$\depset_1=\set{\att{item}\ra\att{cost}}$ and under
$\depset_2=\set{\att{buyer}\ra\att{phone}}$. Then,
Theorem~\ref{thm:disjoint} implies that an optimal U-repair can
be computed in polynomial time under $\depset$ as well.

Now consider the following set of FDs.
\[\depset'\eqdef\set{\att{item}\ra\att{cost} \,,\, \att{buyer}\ra\att{address} \,,\,\att{address}\ra\att{state}}\]
Kolahi and Lakshmanan~\cite{DBLP:conf/icdt/KolahiL09} proved that it
is NP-hard to compute an optimal U-repair for $\set{A\ra B,B\ra C}$,
by reduction from the problem of finding a minimum vertex cover of a
graph $G$. Their reduction is, in fact, a PTAS reduction if we use
vertex cover in a graph of a bounded
degree~\cite{DBLP:journals/tcs/AlimontiK00}. Hence, computing an
optimal U-repair is APX-hard for this set of
FDs. Theorem~\ref{thm:disjoint} then implies that it is also
APX-hard for $\depset'$.
\end{example}

Next, we discuss the problem in the presence of consensus FDs. The
following theorem states that such FDs do not change the complexity of
the problem. Recall that, for a set $\depset$ of FDs and a
set $X$ of attributes, the set $\depset-X$ denotes the set of FDs that
is obtained from $\depset$ by removing each attribute of $X$ from the
lhs and rhs of every FD. Also recall that
$\closure_{\depset}(\emptyset)$ is the set of all consensus
attributes.

\def\thmuconsensus{
  Let $\depset$ be a set of FDs.  There is a strict reduction from
  computing an optimal U-repair for $\depset$ to computing an optimal U-repair for
  $\depset-\closure_{\depset}(\emptyset)$, and vice versa.
}
\begin{theorem}\label{thm:u-consensus} 
\thmuconsensus
\end{theorem}
The proof (given in the \hyperref[sec:thm:disjoint:proof]{Appendix})
  uses Theorem~\ref{thm:disjoint} and a special treatment of the case
  where all FDs are consensus.
%of
%$\emptyset \ra \closure_{\depset}(\emptyset)$ 
%can be computed in polynomial time when all FDs are consensus.  
As an example of applying Theorem~\ref{thm:u-consensus}, if $\depset$
consists of \e{only} consensus FDs, then an optimal U-repair can be
computed in polynomial time, since
$\depset-\closure_{\depset}(\emptyset)$ is empty. As another example,
if $\depset$ is the set $\set{\emptyset\ra D,AD\ra B,B\ra CD}$ then
$\depset-\closure_{\depset}(\emptyset)=\set{A\ra B,B\ra C}$ and,
according to Theorem~\ref{thm:u-consensus}, computing an optimal
U-repair is APX-hard, since this problem is hard for
$\set{A\ra B,B\ra C}$ due to Kolahi and
Lakshmanan~\cite{DBLP:conf/icdt/KolahiL09}, as explained in
Example~\ref{example:disjoint-u}.

\subsection{Reductions to/from Subset Repairing}
In this section we establish several results that enable us to infer
complexity results for the problem of computing an optimal U-repair from
that of computing an optimal S-repair via polynomial-time reductions.
These results are based on the following proposition, which shows that
we can transform a consistent update into a consistent subset (with no
extra cost) and, in the absence of consensus FDs, a consistent subset
into a consistent update (with some extra cost). We give the proof
here, as it shows the actual constructions.

\begin{proposition}\label{prop:U-S-transform}
  Let $\depset$ be a set of FDs and $T$ a table.  The following
  can be done in polynomial time.
  \begin{enumerate}
\item Given a consistent update $U$, construct a consistent subset
    $S$ such that  $\dists(S,T) \leq \distu(U,T)$.
  \item Given a consistent subset $S$, and assuming that $\depset$ is
    consensus free, construct a consistent update $U$ such that
    $\distu(U,T) \leq \mc(\depset) \cdot \dists(S,T)$.
\end{enumerate}
\end{proposition}
\begin{proof}
  We construct $S$ from $U$ by excluding any $i \in \ids(T)$ such that
  $T[i]$ has at least one attribute updated in $U$ (i.e.,
  $H(T[i], U[i]) \geq 1$). We construct $U$ from $S$ as follows. Let $C$ be an lhs cover 
  of minimum cardinality $\mc(\depset)$.  The
  tuple of each $i \in\ids(S)$ is left intact, and for
  $i \in \ids(T) \setminus \ids(S)$ we update the value of $T[i].A$
  for each attribute $A \in C$ to a fresh constant from our infinite domain
  $\dom$. Since $C$ is an lhs cover and there are no consensus FDs,
  for all $X\ra Y$ in $\depset$ it holds that two distinct tuples in $U$ that
  agree on $X$ must correspond to intact tuples; hence, $U$ is
  consistent (as $S$ is consistent).
\end{proof}

As we discuss later in Section~\ref{sec:approx-u}, 
Proposition~\ref{prop:U-S-transform},
combined with Proposition~\ref{prop:subset-approx},
 reestablishes the
result of Kolahi and Lakshmanan~\cite{DBLP:conf/icdt/KolahiL09},
stating the computing a U-repair is in APX.  We also establish 
the following additional consequences of
Proposition~\ref{prop:U-S-transform}. The first is an immediate
corollary (that we refer to later on) about the relationship between
optimal repairs.
\begin{corollary}\label{cor:U-S-opt-val-compare}
  Let $\depset$ be a set of FDs, $T$ a table, $S^*$ an optimal
  S-repair of $T$, and $U^*$ an optimal U-repair of $T$. Then
  $\dists(S^*,T) \leq \distu(U^*,T)$. Moreover, if $\depset$ is
  consensus free, then $\distu(U^*,T) \leq \mc(\depset) \cdot
  \dists(S^*,T)$.
\end{corollary}

The second consequence relates to FD sets $\depset$ with a common lhs,
that is, $\mc(\depset)=1$.

\begin{corollary}\label{cor:U-S-same-mc-1}
  Let $\depset$ be an FD set with a common lhs. There is a strict
  reduction from the problem of computing an optimal S-repair to that
  of computing an optimal U-repair, and vice versa.
\end{corollary}

For example, if $\depset$ consists of a single FD, then an optimal
U-repair can be computed in polynomial time. Additional examples
follow.
\begin{example}
  To illustrate the use of Corollary~\ref{cor:U-S-same-mc-1}, consider
  the FD set $\depset$ of our running example
  (Figure~\ref{fig:example}). Since $\depset$ has a common lhs, and we
  have established in Example~\ref{example:dichotomy} that an optimal
  S-repair for $\depset$ can be found in polynomial time (i.e.,
  $\depset$ passes the test of $\algname{OSRSucceeds}$), we get that an
  optimal U-repair can also be computed in polynomial time for
  $\depset$.
  
  As another illustration, consider the following FD set.
  \[\depset_1\eqdef\set{\att{id country}\ra\att{passport}\,,\,\att{id passport}\ra\att{country}}\]
  Again, as $\depset_1$ has a common lhs, and $\depset_1$ passes the
  test of $\algname{OSRSucceeds}$ (by applying common lhs followed by
  an lhs marriage), from Theorem~\ref{thm:dichotomy} we conclude that
  an optimal U-repair can be found in polynomial time.

  Finally, consider the following set of FDs.
  \[\depset_2\eqdef\set{\att{state city}\ra\att{zip}\,,\,\att{state  zip}\ra\att{country}}\]
  The reader can verify that $\depset_2$ fails 
  $\algname{OSRSucceeds}$, and therefore, from
  Theorem~\ref{thm:dichotomy} we conclude that computing an optimal
  U-repair is APX-complete. 
%As we discuss later in Section~\ref{sec:approx-u}, Kolahi and Lakshmanan have shown that computing an optimal U-repair is in APX under data complexity for any set of FDs~\cite{DBLP:conf/icdt/KolahiL09}, therefore it is also APX-complete.
\end{example}

By combining Theorem~\ref{thm:u-consensus},
Corollary~\ref{cor:U-S-same-mc-1}, and
Corollary~\ref{cor:subset-chain-tractable}, we conclude the following.
\begin{corollary}\label{cor:update-chain-tractable}
  If $\depset$ is a chain FD set, then an optimal U-repair can be
  computed in polynomial time. 
\end{corollary}
\begin{proof}
  If $\depset$ is a chain FD set, then so is
  $\depset-\closure_{\depset}(\emptyset)$.
Theorem~\ref{thm:u-consensus} states that computing an optimal
U-repair has the same complexity under the two FD sets.
   Moreover, if
  $\depset-\closure_{\depset}(\emptyset)$ is nonempty, then it has at least
  one common lhs. From Corollary~\ref{cor:U-S-same-mc-1} we conclude
  that the problem then strictly reduces to computing an S-repair,
  which, by Corollary~\ref{cor:subset-chain-tractable}, can be done in
  polynomial time.  
\end{proof}

Hence, for chain FD sets, an optimal repair can be computed for both
subset and update variants.

\subsection{Incomparability to S-Repairs}\label{sec:incomparable-u-s}

Corollaries~\ref{cor:U-S-same-mc-1}
and~\ref{cor:update-chain-tractable} state cases where computing an
optimal S-repair has the same complexity as computing an optimal
U-repair. A basic case (among others) that the corollaries do not
cover is in the following proposition, where again both variants have
the same (polynomial-time) complexity.

\def\propABBA{
  Under $\depset=\set{A\ra B,B\ra A}$, an
  optimal U-repair can be computed in polynomial time.
}

\begin{proposition}\label{prop:AB-BA}
\propABBA
\end{proposition}
For $\depset = \set{A\ra B,B\ra A}$, although $\mc(\depset) = 2$, we
show that $\distu(U^*, T) = \dists(S^*, T)$ for an optimal U-repair
$U^*$ and an optimal S-repair $S^*$ for any table $T$ over $R$. Since
$\depset$ passes the test of $\algname{OSRSucceeds}$ (by applying lhs
marriage), from Theorem~\ref{thm:dichotomy} an optimal S-repair can be
computed in polynomial time, and therefore an optimal U-repair of
$\set{A\ra B,B\ra A}$ can also be computed in polynomial time.

\par
Do the two variants of optimal repairs feature the same complexity for
\e{every} set of FDs? Next, we answer this question in a negative way.

We have already seen an example of an FD set $\depset$ where an
optimal U-repair can be computed in polynomial time, but finding an
S-repair is APX-complete.  Indeed, Example~\ref{example:disjoint-u}
shows that $\set{A\ra B,C\ra D}$ is a tractable case for optimal
U-repairs; yet, it fails the test of $\algname{OSRSucceeds}$, and is
therefore hard for optimal S-repairs (Theorem~\ref{thm:dichotomy}).

Showing an example of the other direction is more involved. The FD set
in this example is $\marriageplus$ from
Example~\ref{example:common-and-marriage}. In
Example~\ref{example:dichotomy} we showed that $\marriageplus$ passes
the test of $\algname{OSRSucceeds}$, and therefore, an optimal S-repair
can be computed in polynomial time. This is not the case for an
optimal U-repair.

\def\thmmarriageuhard{
  For the relation schema $R(A,B,C)$ and the FD set $\marriageplus$,
  computing an optimal U-repair is APX-complete, even on unweighted,
  duplicate-free tables.
}
\begin{theorem}\label{thm:marriage-u-hard}
\thmmarriageuhard
\end{theorem}
The proof of hardness in Theorem~\ref{thm:marriage-u-hard} is inspired
by, yet different from, the reduction of Kolahi and
Lakshmanan~\cite{DBLP:conf/icdt/KolahiL09} for showing hardness for
$\set{A\ra B,B\ra C}$. The reduction is from the problem of finding a
minimum vertex cover of a graph $G(V,E)$. Every edge $\set{u,v}\in E$
gives rise to the tuples $(u,v,\val{0})$ and $(v,u,\val{0})$.  In
addition, each vertex $v\in V$ gives rise to the tuple
$(v,v,\val{1})$. The proof shows that there is a consistent update of
cost at most $2|E|+k$ (assuming a unit weight for each tuple) if and
only if $G$ has a vertex cover of size at most $k$. To establish a
PTAS reduction, we use the APX-hardness of vertex cover when $G$ is of
bounded degree~\cite{DBLP:journals/tcs/AlimontiK00}.  The challenging
(and interesting) part of the proof is in showing that a consistent
update of cost $2|E|+k$ can be transformed into a vertex cover of size
$k$; this part is considerably more involved than the corresponding
proof of Kolahi and Lakshmanan~\cite{DBLP:conf/icdt/KolahiL09}.

We conclude with the following corollary, stating the existence of FD
sets where the two variants of the problem feature different
complexities.
% (as we discuss in the next section, 
%an $\alpha$-optimal U-repair can be computed in polynomial time where $\alpha$ is a constant assuming data complexity,  showing that 
%the problem is in APX \cite{DBLP:conf/icdt/KolahiL09}).

\begin{corollary}
  There exist FD sets $\depset_1$ and $\depset_2$ such that:
  \begin{enumerate}
  \item Under $\depset_1$ an optimal S-repair can be computed in
    polynomial time, but computing an optimal U-repair is APX-complete.
  \item Under $\depset_2$ an optimal U-repair can be computed in
    polynomial time, but computing an optimal S-repair is APX-complete.
\end{enumerate}
\end{corollary}

\subsection{Approximation}\label{sec:approx-u}
In this section we discuss approximations for optimal U-repairs.  The
combination of Propositions~\ref{prop:subset-approx}
and~\ref{prop:U-S-transform} gives the following.
\begin{theorem}\label{thm:approx-u}
  Let $\depset$ be a set of FDs, and $\alpha=2\cdot\mc(\depset)$. An
  $\alpha$-optimal U-repair can be computed in polynomial time.
\end{theorem}
Note that the approximation ratio can be further improved by applying
Theorem~\ref{thm:disjoint}: if $\depset$ is the union of
attribute-disjoint FD sets $\depset_1$ and $\depset_2$, then an
$\alpha$-optimal U-repair can be computed (under $\depset$) where
\[\alpha=2\cdot\max\set{\mc(\depset_1),\mc(\depset_2)}\,.\]

\begin{sloppypar}
Kolahi and Lakshmanan~\cite{DBLP:conf/icdt/KolahiL09} have also given
a constant-factor approximation algorithm for U-repairs (assuming
$\depset$ is fixed). To compare their ratio with ours, we first
explain their ratio.
\end{sloppypar}

\def\mfs{\mathit{MFS}}
\def\mci{\mathit{MCI}}

Let $\depset$ be a set of FDs and assume (without loss of generality)
that the rhs of each FD consists of a single attribute.  Kolahi and
Lakshmanan~\cite{DBLP:conf/icdt/KolahiL09} define two measures on
$\depset$.  By $\mfs(\depset)$ we denote the maximum number of
attributes involved in the lhs of any FD in $\depset$.  An \e{implicant} of an
attribute $A$ is a set $X$ of attributes such that $X\ra A$ is in
the closure of $\depset$.  A \e{core implicant} of $A$ is a set $C$ of
attributes that hits every implicant of $A$ (\ie, $X\cap
C\neq\emptyset$ whenever $X \rightarrow A$ is in the closure of
$\depset$).  A \e{minimum core implicant} of $A$ is a core implicant
of $A$ with the smallest cardinality. By $\mci(\depset)$ we denote the
size of the largest minimum core implicant over all attributes $A$.

\begin{citedtheorem}{DBLP:conf/icdt/KolahiL09}\label{thm:KL-approx-U}
  For any set $\depset$ of FDs, an $\alpha$-optimal U-repair can be
  computed in polynomial time where $\alpha=(\mci(\depset) + 2) \cdot
  (2 \mfs(\depset) - 1)$.
\end{citedtheorem}
In both Theorems~\ref{thm:approx-u} and \ref{thm:KL-approx-U}, the
approximation ratios are constants under data complexity, but depend
on $\depset$.  It is still unknown whether there is a constant
$\alpha$ that applies to all FD sets $\depset$. Yet, it is known that
a constant-ratio approximation cannot be obtained in polynomial time
under \emph{combined complexity} (where $R$, $T$, and $\depset$ are
all given as input)~\cite{DBLP:conf/icdt/KolahiL09}.
% and no polynomial time approximation scheme
%exists under data complexity.

%Therefore, a natural question is whether there exists an approximation that is an \emph{absolute constant $\alpha$} independent of the given set of FDs $\depset$. Kolahi and Lakshmanan \cite{DBLP:conf/icdt/KolahiL09} showed that there exists a schema and a set of FDs for which  finding an $\alpha$-optimal U-repair for all sets of FDs is NP-hard  for any constant $\alpha > 1$. However, the size of the schema and the FDs is not fixed in this reduction. If the size of the schema is fixed, they show that the problem is \emph{Max-SNP-hard} ruling out existence of a poly-time approximation scheme (PTAS), and whether the problem is approximable in polynomial time within an absolute constant factor independent of the FDs for a fixed schema is an open question.
\par
Although the proof of Theorem~\ref{thm:approx-u} is much simpler than
the non-trivial proof of Theorem~\ref{thm:KL-approx-U} given in
\cite{DBLP:conf/icdt/KolahiL09}, it can be noted that the
approximation ratios in these two theorems are not directly comparable.
%and can vary with different sets of FDs.
%Since 
%\[\mci(\depset) \leq
%\mc(\depset) \leq k \times \mci(\depset),\] 
%and
%$\mfs(\depset) \leq k$, and
%$\mc(\depset) \leq k$, 
If $k$ is the number of attributes, then the worst-case approximation
ratio in Theorem~\ref{thm:KL-approx-U} is quadratic in $k$, while the
worst-case approximation in Theorem~\ref{thm:approx-u} is linear in
$k$ (precisely, linear in $\min(k, |\depset|)$). Moreover, an easy
observation is that the ratio between the two approximation ratios can
be at most linear in $k$.  In the remainder of this section, we
illustrate the difference between the approximations with examples.
\par 
First, we show an infinite sequence of FD sets where the approximation ratio
of Theorem~\ref{thm:approx-u} is $\Theta(k)$ and that of Theorem~\ref{thm:KL-approx-U} is $\Theta(k^2)$.
%asymptotically better. 
For a natural
number $k\geq 1$, we define $\depset_k$ as follows.
%\[\depset_k \eqdef \set{A_1\ra B_0\,,\,\dots\,,\, A_k \ra B_0\,,\, B_0B_1\cdots B_k \ra C}\]
\[\depset_k \eqdef \set{A_0\cdots A_k \ra B_0, B_0 \rightarrow C, B_1 \rightarrow A_0,\, \ldots, \, B_k \rightarrow A_0}\]
The
approximation ratio for $\depset_k$ given by Theorem~\ref{thm:approx-u} is 
$2(k+2)$. For the approximation ratio of
Theorem~\ref{thm:KL-approx-U}, we have $\mfs(\depset_k)=k+1$ (due to the FD
$A_0\cdots A_k \ra B$) and $\mci(\depset_k)=k$ (since the core
implicant of $A_0$ is $\set{B_1,\dots,B_k}$). Hence, the approximation
ratio of Theorem~\ref{thm:KL-approx-U} grows quadratically with $k$
(i.e., it is $\Theta(k^2)$).

On the other hand, following is a sequence of FD sets in which the
approximation ratio of Theorem~\ref{thm:approx-u} grows linearly with $k$,
while that of Theorem~\ref{thm:KL-approx-U} is a constant.
\[\depset'_k \eqdef \set{A_0A_1\ra B_0,\,A_1A_2\ra B_1,\,\dots,\,A_kA_{k+1}\ra B_k}\]
Here, the approximation ratio of Theorem~\ref{thm:approx-u} is $\Theta(k)$
(since $\mc(\depset'_k)$ is $\lceil (k+1)/2 \rceil$), but that of
Theorem~\ref{thm:KL-approx-U} is constant, since $\mfs(\depset'_k)=2$
and $\mci(\depset'_k)=1$. 

The following theorem shows that computing an optimal U-repair for
both $\depset_k$ and $\depset_k'$ is a hard problem, thus an
approximation is, indeed, needed.

\def\theoremuKLquad{
  Let $k\geq 1$ be fixed. Computing an optimal U-repair is APX-complete for:
  \begin{enumerate}
  \item $R(A_0, \dots, A_k,B_0,\dots,B_k, C)$ and $\depset_k$;
  \item $R(A_0,\dots,A_{k+1},B_0,\dots,B_k)$ and $\depset_k'$.
  \end{enumerate}
  }
  
\begin{theorem}\label{theorem:u-KL-quad}
\theoremuKLquad
\end{theorem}

The proof for $\depset_k$ is by a reduction from computing an optimal
U-repair under $\set{A\ra B, B\ra C}$ (see
Example~\ref{example:disjoint-u}). For $\depset_k'$, we first show
that the problem is APX-hard for $k=1$. In this case, the FD set
contains two FDs $A_0A_1\rightarrow B_0$ and $A_1A_2\rightarrow B_1$,
thus $A_1$ is a common lhs, and Corollary~\ref{cor:U-S-same-mc-1},
combined with the fact that computing an optimal S-repair under
$\set{A\rightarrow B, C\rightarrow D}$ is APX-hard
(Theorem~\ref{thm:dichotomy}), imply that computing an optimal
U-repair is APX-hard as well. Then, we construct a reduction from
computing an optimal U-repair under $\depset'_k$ for $k=1$ to
computing an optimal U-repair under $\depset'_k$ for $k>1$.

Clearly, one can take the benefit of the approximations of both
Theorems~\ref{thm:approx-u} and~\ref{thm:KL-approx-U} by computing
U-repairs by both algorithms and selecting the one with the smaller
cost. As we showed, this combined approximation
outperforms each of its two components.

\balance
\section{Discussion and Future Work}\label{sec:conclusions}
We investigated the complexity of computing an optimal S-repair and an
optimal U-repair. For the former, we established a dichotomy over all
sets of FDs (and schemas). For the latter, we developed general
techniques for complexity analysis, showed concrete complexity
results, and explored the connection to the complexity of
S-repairs. We presented approximation results and, in the case of
U-repairs, compared to the approximation of Kolahi and
Lakshmanan~\cite{DBLP:conf/icdt/KolahiL09}. In the case of S-repairs,
we drew a direct connection to probabilistic database repairs, and
completed a dichotomy by Gribkoff et al.~\cite{GVSBUDA14} to the
entire space of FDs. Quite a few directions are left for future
investigation, and we conclude with a discussion of some
of these.

%We investigated the complexity of computing an optimal S-repair and an
%optimal U-repair. In the case of S-repairs, we established a dichotomy
%over all sets of FDs (and schemas). For U-repairs, we analyzed the
%complexity of different classes of FDs, and explored the connection to
%the complexity of finding an optimal S-repair. We also presented
%approximation algorithms and, in the case of U-repairs, compared our algorithm to
%that of Kolahi and Lakshmanan~\cite{DBLP:conf/icdt/KolahiL09}. In the
%case of S-repairs, we drew a direct connection to the probabilistic
%database repair, and completed a dichotomy by Gribkoff, Van den Broeck
%and Suciu~\cite{GVSBUDA14} to the entire space of FDs.  Quite a few
%directions are left for future investigation, and we conclude this
%paper with a description of some of these.

As our results are restricted to FDs, an obvious important direction is
to extend our study to other types of integrity constraints, such as
denial constraints~\cite{DBLP:journals/jiis/GaasterlandGM92},
conditional FDs~\cite{DBLP:conf/icde/BohannonFGJK07}, referential
constraints~\cite{Date:1981:RI:1286831.1286832}, and tuple-generating
dependencies~\cite{DBLP:journals/siamcomp/BeeriV84}. Moreover, the
repair operations we considered are either exclusively tuple deletions
or exclusively value updates. Hence, another clear direction is to
allow mixtures of deletions, insertions and updates, where the cost
depends on the operation type, the involved tuple, and the involved
attribute  (in the case of updates). 

Our understanding of the complexity of computing an optimal U-repair
is considerably more restricted than that of an optimal S-repair. We
would like to complete our complexity analysis for optimal U-repairs
into a full dichotomy.  More fundamentally, we would like to
incorporate restrictions on the allowed value updates. Our results are
heavily based on the ability to update \e{any cell} with \e{any value}
from an infinite domain. A natural restriction on the update repairs
is to allow revising only certain attributes, possibly using a finite
(small) space of possible new values. It is not clear how to
incorporate such a restriction in our results and proof techniques.

In the case of S-repairs, we are interested in incorporating
\e{preferences}, as in the framework of \e{prioritized repairing} by
Staworko et al.~\cite{DBLP:journals/amai/StaworkoCM12}. There,
priorities among tuples allow to eliminate subset repairs that are
inferior to others (where ``inferior'' has several possible
interpretations). It may be the case that priorities are rich enough
to clean the database
unambiguously~\cite{DBLP:conf/icdt/KimelfeldLP17}. A relevant question
is, then, what is the minimal number of tuples that we need to delete
in order to have an unambiguous repair? Alternatively, how many
preferences are needed for this cause?

%\section{TODOs}
%\begin{itemize}
%\end{itemize}

\newpage
\bibliographystyle{abbrv}
\bibliography{minrepair}

\newpage
\onecolumn
\appendix\label{sec:appendix}
\section{Details from Section 3}\label{app:proofs-subset-repair}

In this section, we prove Theorem~\ref{thm:dichotomy}.

\begin{reptheorem}{\ref{thm:dichotomy}}
\thmdichotomy
\end{reptheorem}

\subsection{Tractability Side}~\label{sec:s-repair-tractability}

We start by proving the positive side of Theorem~\ref{thm:dichotomy}. Figure~\ref{fig:ptime-srepair} illustrates our proof. The idea is the following: we start with a table $T_0$ and an FD set $\depset_0$, for which we want to find an optimal S-repair. Then, we apply the simplifications (represented by the red arrows in the figure) until it is no longer possible. We can find an optimal S-repair in polynomial time if we have a trivial set of FDs at this point (which can also be empty). In order to prove that, we show for each one of the simplifications (common lhs, consensus FD and lhs marriage), that if the problem after the simplification can be solved in polynomial time, then we can use this solution to solve the original problem (before the simplification). The lemmas that appear next the black arrows in the figure contain our proofs.

\begin{lemma}\label{lemma:s1-ptime}
Let $T$ be a table, and let $\depset$ be a set of FDs, such that $\depset$ has a common lhs $A$. If for each $a\in\pi_AT[*]$, $\algname{\OSR}(\depset-A,\sigma_{A=a}T)$ returns an optimal S-repair of $\sigma_{A=a}T$ w.r.t. $\depset-A$, then $\algname{\OSR}(\depset,T)$ returns an optimal S-repair of $T$ w.r.t. $\depset$. 
\end{lemma}

\begin{proof}
Assume that for each $a\in\pi_AT[*]$, it holds that $\algname{\OSR}(\depset-A, \sigma_{A=a}T)$ returns an optimal S-repair of $\sigma_{A=a}T$ w.r.t. $\depset-A$. We contend that the algorithm $\algname{\OSR}(\depset,T)$ returns an optimal S-repair of T w.r.t. $\depset$. Since the condition of line~4 of $\algname{\OSR}$ is satisfied, the subroutine $\algname{CommonLHSRep}$ will be called. Thus, we have to prove that $\algname{CommonLHSRep1}(\depset,T)$ returns an optimal S-repair of $T$.

Let $J$ be the result of $\algname{CommonLHSRep}(\depset,T)$. We will start by proving that $J$ is consistent. Let us assume, by way of contradiction, that $J$ is not consistent. Thus, there are two tuples $\tup t_1$ and $\tup t_2$ in $J$ that violate an FD $Z\rightarrow W$ in $\depset$. Since $A\in Z$ (as $A$ is a common lhs attribute), the tuples $\tup t_1$ and $\tup t_2$ agree on the value of attribute $A$. Assume that $\tup t_1.A=\tup t_2.A=a$. By definition, there is an FD $(Z\setminus\set{A)}\rightarrow (W\setminus\set{A})$ in $\depset-A$. Clearly, the tuples $\tup t_1$ and $\tup t_2$ agree on all the attributes in $Z\setminus\set{A}$, and since they also agree on the attribute $A$, there exists an attribute $B\in (W\setminus\set{A})$ such that $\tup t_1.B\neq \tup t_2.B$. Thus, $\tup t_1$ and $\tup t_2$ violate an FD in $\depset-A$, which is a contradiction to the fact that $\algname{\OSR}(\depset-A, \sigma_{A=a}T)$ returns an optimal S-repair of $\sigma_{A=a}T$ that contains both $\tup t_1$ and $\tup t_2$.

Next, we will prove that $J$ is an optimal S-repair of $T$. Let us assume, by way of contradiction, that this is not the case. That is, there is another consistent subset $J'$ of $T$, such that the weight of the tuples in $T\setminus J'$ is lower than the weight of the tuples in $T\setminus J$. In this case, there exists at least one value $a'$ of attribute $A$, such that the weight of the tuples $\tup t\in (T\setminus J')$ for which it holds that $\tup t.A=a'$ is lower than the weight of such tuples in $T\setminus J$. Let $H=\set{\tup h_1,\dots,\tup h_r}$ be the set of tuples from $T$ for which it holds that $\tup h_j.A=a'$. Let $\set{\tup f_1,\dots,\tup f_n}$ be the set of tuples in $H\cap J$, and let $\set{\tup g_1,\dots,\tup g_m}$ be the set of tuples in $H\cap J'$. It holds that $w_T(H\setminus J) > w_T(H\setminus J')$. We claim that $\set{\tup g_1,\dots,\tup g_m}$ is a consistent subset of $\sigma_{A=a'}T$, which is a contradiction to the fact that $\set{\tup f_1,\dots,\tup f_n}$ is an optimal S-repair of $\sigma_{A=a'}T$. Let us assume, by way of contradiction, that $\set{\tup g_1,\dots,\tup g_m}$ is not a consistent subset of $\sigma_{A=a'}T$. Thus, there exist two tuples $\tup g_{j_1}$ and $\tup g_{j_2}$ in $\set{\tup g_1,\dots,\tup g_m}$ that violate an FD, $Z\rightarrow W$, in $\depset-A$. By definition, there is an FD $(Z\cup \set{A})\rightarrow (W\cup Y)$ in $\depset$, where $Y\subseteq \set{A}$, and since $\tup g_{j_1}$ and $\tup g_{j_2}$ agree on the value of attribute $A$, they clearly violate this FD, which is a contradiction to the fact that they both appear in $J'$ (which is a consistent subset of $T$).
\end{proof}

\begin{lemma}\label{lemma:s2-ptime}
Let $T$ be a table, and let $\depset$ be a set of FDs, such that $\depset$ contains a consensus FD $\emptyset\rightarrow X$. If for each $a\in\pi_XT[*]$, $\algname{\OSR}(\depset-X,\sigma_{X=a}T)$ returns an optimal S-repair of $\sigma_{X=a}T$ w.r.t. $\depset-X$, then $\algname{\OSR}(\depset,T)$ returns an optimal S-repair of $T$ w.r.t. $\depset$.
\end{lemma}

\begin{proof}
Assume that for each $a\in\pi_XT[*]$, it holds that $\algname{\OSR}(\depset-X, \sigma_{X=a}T)$ returns an optimal S-repair of $\sigma_{X=a}T$ w.r.t. $\depset-X$. We contend that the algorithm $\algname{\OSR}(\depset,T)$ returns an optimal S-repair of $T$ w.r.t. $\depset$. Note that the condition of line~4 cannot be satisfied, since there is no attribute that appears on the left-hand side of $\emptyset\rightarrow A$. Since the condition of line~6 of $\algname{\OSR}$ is satisfied, the subroutine $\algname{ConsensusRep}$ will be called. Thus, we have to prove that $\algname{ConsensusRep}(\depset,T)$ returns an optimal S-repair of $T$.

Let $J$ be the result of $\algname{ConsensusRep}(\depset,T)$. We will start by proving that $J$ is consistent. Let us assume, by way of contradiction, that $J$ is not consistent. Thus, there are two tuples $\tup t_1$ and $\tup t_2$ in $J$ that violate an FD $Z\rightarrow W$ in $\depset$. Note that $\tup t_1$ and $\tup t_2$ agree on the value of the attributes in $X$ (since $\algname{ConsensusRep}(\depset,T)$ always returns a set of tuples that agree on the value of the attributes in $X$). Assume that $\tup t_1[X]=\tup t_2[X]=a$. Thus, it holds that there exists an attribute $A'\in W$, such that $\tup t_1.A'\neq \tup t_2.A'$. By definition, there is an FD $(Z\setminus X)\rightarrow (W\setminus X)$ in $\depset-X$. Clearly, the tuples $\tup t_1$ and $\tup t_2$ agree on all the attributes in $Z\setminus X$, but do not agree on the attribute $A'\in (W\setminus X)$. Thus, $\tup t_1$ and $\tup t_2$ violate an FD in $\depset-X$, which is a contradiction to the fact that $\algname{\OSR}(\depset-X, \sigma_{X=a}T)$ returns an optimal S-repair of $\sigma_{X=a}T$ that contains both $\tup t_1$ and $\tup t_2$. 

Next, we will prove that $J$ is an optimal S-repair of $T$. Let us assume, by way of contradiction, that this is not the case. That is, there is another consistent subset $J'$ of $T$, such that the weight of the tuples in $T\setminus J'$ is lower than the weight of the tuples in $T\setminus J$. Clearly, each consistent subset of $T$ only contains tuples that agree on the value of attribute $A$ (that is, tuples that belong to $\sigma_{X=a}T$ for some value $a$. The instance $J$ is an optimal S-repair of $\sigma_{X=a'}T$ for some value $a'$. If $J'\subseteq \sigma_{X=a'}T$, then we get a contradiction to the fact that $J$ is an optimal S-repair of $\sigma_{X=a'}T$. Thus, $J'\subseteq \sigma_{X=a''}T$ for some other value $a''$. In this case, the optimal S-repair of $\sigma_{X=a''}T$ has a higher weight than the optimal S-repair of $\sigma_{X=a'}T$, which is a contradiction to the fact that $\algname{ConsensusRep}$ returns an optimal S-repair of $\sigma_{X=a}T$ that has the highest weight among all such optimal S-repairs.
\end{proof}

\begin{figure}
\centering
\input{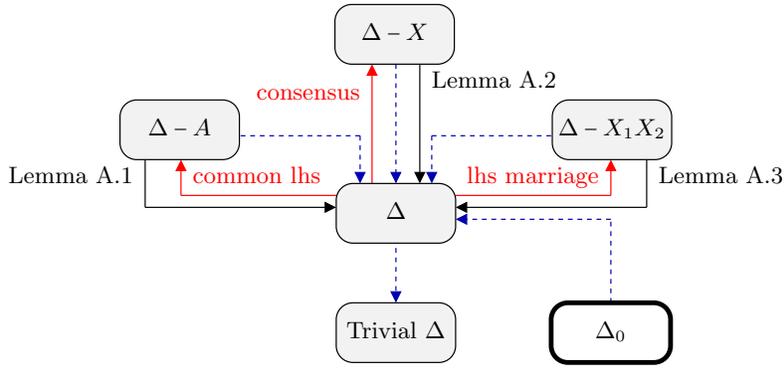}
\vskip0.5em
\caption{An illustration of our proof of the positive side of Theorem~\ref{thm:dichotomy}. We start with an FD set $\depset_0$ and apply simplifications to it, until we get a trivial set of FDs $\depset$. The red arrows represent simplifications and the dashed blue arrows represent renaming of a set of FDs. A black arrow from $\depset'$ to $\depset$ means that if we can also find an optimal S-repair for $\depset'$ in polynomial time, then we can find an optimal S-repair for $\depset$ in polynomial time. The proof is in the lemma that appears next to the corresponding black arrow.}\label{fig:ptime-srepair}
\end{figure}

\begin{lemma}\label{lemma:s3-ptime}
Let $T$ be a table, and let $\depset$ be a set of FDs, such that $\depset$ has an lhs marriage $(X_1,X_2)$ and does not have a common lhs. If it holds that $\algname{\OSR}(\depset-X_1X_2,\sigma_{X_1=a_1,X_2=a_2}T)$ returns an optimal S-repair of $\sigma_{X_1=a_1,X_2=a_2}T$ w.r.t. $\depset-X_1X_2$ for each pair $(a_1,a_2)\in\pi_{X_1X_2}T[*]$, then $\algname{\OSR}(\depset,T)$ returns an optimal S-repair of $T$ w.r.t. $\depset$.
\end{lemma}

\begin{proof}
Assume that for each $(a_1,a_2)\in\pi_{X_1X_2}T[*]$, it holds that $\algname{\OSR}(\depset-X_1X_2,\sigma_{X_1=a_1,X_2=a_2}T)$ returns an optimal S-repair of $\sigma_{X_1=a_1,X_2=a_2}T$ w.r.t. $\depset-X_1X_2$. We contend that the algorithm $\algname{\OSR}(\depset,T)$ returns an optimal S-repair of $T$ w.r.t. $\depset$. We assumed that $\depset$ does not have a common lhs, thus the condition of line~4 is not satisfied. The condition of line~6 cannot be satisfied as well, since neither $X_1\subseteq\emptyset$ nor $X_2\subseteq\emptyset$ (as we remove trivial FDs from $\depset$ in line~3). The condition of line~8 on the other hand is satisfied, thus the algorithm will call subroutine $\algname{MarriageRep}$ and return the result. Thus, we have to prove that $\algname{MarriageRep}(\depset,T)$ returns an optimal S-repair of $T$.

Let us denote by $J$ the result of $\algname{MarriageRep}(\depset,T)$. We will start by proving that $J$ is consistent. Let $\tup t_1$ and $\tup t_2$ be two tuples in $T$. Note that it cannot be the case that $\tup t_1[X_1]\neq \tup t_2[X_1]$ but $\tup t_1[X_2]= \tup t_2[X_2]$ (or vice versa), since in this case the matching that we found for the graph $G$ contains two edges $(a_1,a_2)$ and $(a_1',a_2)$, which is impossible. Moreover, if it holds that $\tup t_1[X_1]\neq \tup t_2[X_1]$ and $\tup t_1[X_2]\neq \tup t_2[X_2]$, then $\tup t_1$ and $\tup t_2$ do not agree on the left-hand side of any FD in $\depset$ (since we assumed that for each FD $Z\rightarrow W$ in $\depset$ it either holds that $X_1\subseteq Z$ or $X_2\subseteq Z$). Thus, $\set{\tup t_1,\tup t_2}$ satisfies all the FDs in $\depset$. Now, let us assume, by way of contradiction, that $J$ is not consistent. Thus, there are two tuples $\tup t_1$ and $\tup t_2$ in $J$ that violate an FD $Z\rightarrow W$ in $\depset$. That is, $\tup t_1$ and $\tup t_2$ agree on all the attributes in $Z$, but do not agree on at least one attribute $B\in W$. As mentioned above, the only possible case is that $\tup t_1[X_1]=\tup t_2[X_1]=a_1$ and $\tup t_1[X_2]=\tup t_2[X_2]=a_2$. In this case, $\tup t_1$ and $\tup t_2$ both belong to $\sigma_{X_1=a_1,X_2=a_2}T$, and they do not agree on an attribute $B\in (W\setminus(X_1\cup X_2))$. The FD $(Z\setminus(X_1\cup X_2)\rightarrow (W\setminus(X_1\cup X_2))$ belongs to $\depset-X_1X_2$, and clearly $\tup t_1$ and $\tup t_2$ also violate this FD, which is a contradiction to the fact that $J$ only contains an optimal S-repair of $\sigma_{X_1=a_1,X_2=a_2}T$ and does not contain any other tuples from $\sigma_{X_1=a_1,X_2=a_2}T$ (recall that we assumed that $\algname{\OSR}(\depset-X_1X_2,\sigma_{X_1=a_1,X_2=a_2}T)$ returns an optimal S-repair for each $(a_1,a_2)\in\pi_{X_1X_2}T[*]$).

Next, we will prove that $J$ is an optimal S-repair of $T$. Let us assume, by way of contradiction, that this is not the case. That is, there is another consistent subset $J'$ of $T$, such that the weight of the tuples in $T\setminus J'$ is lower than the weight of the tuples in $T\setminus J$, and $J'$ is an optimal S-repair of $T$. Note that the weight of the matching corresponding to $J$ is the total weight of the tuples in $J$ (since the weight of each edge $(a_1,a_2)$ is the weight of the optimal S-repair of $\sigma_{X_1=a_1,X_2=a_2}T$, and $J$ contains the optimal S-repair of each $\sigma_{X_1=a_1,X_2=a_2}T$, such that the edge $(a_1,a_2)$ belongs to the matching). Let $\tup t_1$ and $\tup t_2$ be two tuples in $J'$. Note that it cannot be the case that $\tup t_1[X_1]= \tup t_2[X_1]$ but $\tup t_1[X_2]\neq \tup t_2[X_2]$, since in this case, $\set{\tup t_1,\tup t_2}$ violates the FD $X_1\rightarrow cl_\depset(X_1)$ which is implied from $\depset$ and this is a contradiction to the fact that $J'$ satisfied $\depset$ (we recall that $cl_\depset(X_1)=cl_\depset(X_2)$, and since $X_2\subseteq cl_\depset(X_2)$, it holds that $X_2\subseteq cl_\depset(X_1)$). Similarly, it cannot be the case that $\tup t_1[X_1]\neq \tup t_2[X_1]$ but $\tup t_1[X_2]=\tup t_2[X_2]$. Hence, it either holds that $\tup t_1[X_1]= \tup t_2[X_1]$ and $\tup t_1[X_2]= \tup t_2[X_2]$ or $\tup t_1[X_1]\neq \tup t_2[X_1]$ and $\tup t_1[X_2]\neq \tup t_2[X_2]$. Therefore, $J'$ clearly corresponds to a matching of $G$ as well (the matching will contain an edge $(a_1,a_2)$ if there is a tuple $\tup t\in J'$, such that $\tup t[X_1]=a_1$ and $\tup t[X_2]=a_2$).

Next, we claim that for each edge $(a_1,a_2)$ that belongs to the above matching (the matching that corresponds to $J'$), the subinstance $J'$ contains an optimal S-repair of $\sigma_{X_1=a_1,X_2=a_2}T$ w.r.t. $\depset-X_1X_2$. Clearly, $J'$ cannot contain two tuples $\tup t_1$ and $\tup t_2$ from $\sigma_{X_1=a_1,X_2=a_2}T$ that violate an FD $Z\rightarrow W$ from $\depset-X_1X_2$ (otherwise, $\tup t_1$ and $\tup t_2$ will also violate the FD $(Z\cup Y_1)\rightarrow (W\cup Y_2)$ from $\depset$ (where $Y_1\subseteq (X_1\cup X_2)$ and $Y_2\subseteq (X_1\cup X_2)$), which is a contradiction to the fact that $J'$ is a consistent subset of $T$). Thus, $J'$ contains a consistent set of tuples from $\sigma_{X_1=a_1,X_2=a_2}T$. If this set of tuples is not an optimal S-repair of $\sigma_{X_1=a_1,X_2=a_2}T$, then we can replace this set of tuples with an optimal S-repair of $\sigma_{X_1=a_1,X_2=a_2}T$. This will not break the consistency of $J'$ since these tuples do not agree on the attributes in neither $X_1$ nor $X_2$ with any other tuple in $J'$, and each FD $Z\rightarrow W$ in $\depset$ is such that $X_1\subseteq Z$ or $X_2\subseteq Z$. The result will be a consistent subset of $T$ with a higher weight than $J'$ (that is, the weight of the removed tuples will be lower), which is a contradiction to the fact that $J'$ is an optimal S-repair of $T$. Therefore, for each edge $(a_1,a_2)$ that belongs to the above matching, $J'$ contains tuples with a total weight of $w_T(S_{a1,a2})$, where $S_{a1,a2}$ is an optimal S-repair of $\sigma_{X_1=a_1,X_2=a_2}T$, and the weight of this matching is the total weight of tuples in $J'$. In this case, we found a matching of $G$ with a higher weight than the matching corresponding to $J$, which is a contradiction to the fact that $J$ corresponds to the maximum weighted matching of $G$. 
\end{proof}

Next, we prove Theorem~\ref{thm:osr}.

\begin{reptheorem}{\ref{thm:osr}}
\thmosr
\end{reptheorem}

\begin{proof}
We will prove the theorem by induction on $n$, the number of simplifications that will be applied to $\depset$ by $\algname{\OSR}$. We start by proving the basis of the induction, that is $n=0$. In this case, $\algname{\OSR}$ will only succeed if $\depset=\emptyset$ or if $\depset$ is trivial. Clearly, in this case, $T$ is consistent w.r.t. $\depset$ and an optimal S-repair of $T$ is $T$ itself. And indeed, $\algname{\OSR}(\depset,T)$ will return $T$ in polynomial time.

For the inductive step, we need to prove that if the claim is true for all $n=1,\dots,k-1$, it is also true for $n=k$. In this case, $\algname{\OSR}(\depset,T)$ will start by applying some simplification to the schema. Clearly, the result is a set of FDs $\depset'$, such that $\algname{\OSR}(\depset',T')$ will apply $n-1$ simplifications to $\depset'$. One of the following holds:

\begin{itemize}
\item $\depset$ has a common lhs $A$. In this case, the condition of line~4
  is satisfied and the subroutine $\algname{CommonLHSRep}$ will be
  called. Note that $\algname{\OSR}(\depset,T)$ will succeed only if $\algname{\OSR}(\depset-A,\sigma_{A=a}T)$ succeeds for each $a\in\pi_AT[*]$. We know from the
  inductive assumption that if $\algname{\OSR}(\depset-A,\sigma_{A=a}T)$ succeeds, then it returns an optimal S-repair. Thus, Lemma~\ref{lemma:s1-ptime} implies that $\algname{\OSR}(\depset,T)$ returns an optimal S-repair of $T$ w.r.t. $\depset$.

\item $\depset$ has a consensus FD $\emptyset\rightarrow X$. In this case, the condition of line~4
  is not satisfied, but the condition of line~6 is satisfied and the subroutine $\algname{ConsensusRep}$ will be
  called. Again, $\algname{\OSR}(\depset,T)$ will succeed only if $\algname{\OSR}(\depset-X,\sigma_{X=a}T)$ succeeds for each $a\in\pi_XT[*]$. We know from the
  inductive assumption that if $\algname{\OSR}(\depset-X,\sigma_{X=a}T)$ succeeds, then it returns an optimal S-repair. Thus, Lemma~\ref{lemma:s2-ptime} implies that $\algname{\OSR}(\depset,T)$ returns an optimal S-repair of $T$ w.r.t. $\depset$.

\item $\depset$ does not have a common lhs, but has an lhs marriage. In this case, the conditions of line~4 and line~6
  are not satisfied, but the condition of line~8 is satisfied and the subroutine $\algname{MarriageRep}$ will be
  called. As in the previous cases, $\algname{\OSR}(\depset,T)$ will succeed only if $\algname{\OSR}(\depset-X_1X_2,\sigma_{X_1=a_1,X_2=a_2}T)$ succeeds for each $(a_1,a_2)\in\pi_{X_1X_2}T[*]$. We know from the
  inductive assumption that if $\algname{\OSR}(\depset-X_1X_2,\sigma_{X_1=a_1,X_2=a_2}T)$ succeeds, then it returns an optimal S-repair. Thus, Lemma~\ref{lemma:s3-ptime} implies that $\algname{\OSR}(\depset,T)$ returns an optimal S-repair of $T$ w.r.t.
\end{itemize}

It is only left to prove that $\algname{\OSR}$ terminates in polynomial time in $k$, $|\depset|$, and $|T|$. We will first explain how checking each one of the conditions can be done in polynomial time. Then, we will provide the recurrence relation for each one of the subroutines of the algorithm. The algorithm will first check if $\depset$ is trivial. This can be done in polynomial time as we can go over the FDs and for each one of them check if each attribute that appears on the rhs also appears on the lhs. Then, the algorithm checks if $\depset$ has a common lhs. This can also be done in polynomial time by going over the attributes in $A_1,\dots,A_k$ and checking for each one of them if the lhs of each FD in $\depset$ contains it. If this condition does not hold, then the algorithm will check if $\depset$ has a consensus FD. This can be done in polynomial time by going over the FDs in $\depset$ and checking for each one of them if the lhs is empty. Finally, if this condition does not hold as well, the algorithm will check if $\depset$ has an lhs marriage. To do that in polynomial time, we can go over the FDs in $\depset$ and for each one of them find the closure of its lhs w.r.t. $\depset$ (it is known that this can be done in polynomial time). Then, for each pair of FDs in $\depset$ that agree on the closure of their lhs, we have to check if one of their lhs is contained in the lhs of each other FD in $\depset$.

If the condition of line~4 holds, then the subroutine $\algname{CommonLHSRep}$ will be called. Since we have already found a common lhs $A$ (when we were checking if the condition holds), it is only left to divide $T$ into blocks of tuples that agree on the value of $A$ (which again can be done in polynomial time by going over the tuples in $T$ and for each one of them checking the value of attribute $A$) and make a recursive call to $\algname{\OSR}$ for each one of the blocks. The recurrence relation for this subroutine is:

\begin{equation}
F(k,|T|,|\depset|)=\sum_{a\in\pi_AT[*]} F(k-1,|\sigma_{A=a}T|,|\depset-A|)+poly(k,|T|,|\depset|)
\end{equation}

If the condition of line~6 holds, then the subroutine $\algname{ConsensusRep}$ will be called. This case is very similar to the previous one. Since we have already found a consensus FD $\emptyset\rightarrow X$, it is only left to divide $T$ into blocks of tuples that agree on the value of the attributes in $X$ and make a recursive call to $\algname{\OSR}$ for each one of the blocks. The recurrence relation for this subroutine is:

\begin{equation} 
F(k,|T|,|\depset|)=\sum_{a\in\pi_XT[*]} F(k-|X|,|\sigma_{X=a}T|,|\depset-X|)+poly(k,|T|,|\depset|)
\end{equation}

If the condition of line~8 holds, then the subroutine $\algname{MarriageRep}$ will be called. In this case, we go over each pair of values $(a_1,a_2)$ that appear in the attributes $(X_1,X_2)$ in some tuple of $T$ (note that there are at most $|T|$ such pairs). Then, we make a recursive call to $\algname{\OSR}$ for each one of these pairs and calculate the weight of the result $S_{a_1,a_2}$ (clearly, this can be done in polynomial time in the size of $S_{a_1,a_2}$). Finally, we build the graph $G$ and find a maximum weighted matching for the graph (which can be done in polynomial time with the Hungarian algorithm). The recurrence relation for this subroutine is:

\begin{equation} 
F(k,|T|,|\depset|)=\sum_{(a_1,a_2)\in\pi_{X_1X_2}T[*]} F(k-|X_1\cup X_2|,|\sigma_{X_1=a_1,X_2=a_2}T|,|\depset-X_1X_2|)+poly(k,|T|,|\depset|)
\end{equation}

Since finding an optimal S-repair for a trivial FD set can be done in polynomial time in the size of $T$, and since in each one of the three recurrence relations, the tables in the middle argument form a partition of $T$, a standard analysis of $F$ shows that it is bounded by a polynomial.
\end{proof}

\subsection{Hardness Side}

\begin{figure}
\centering
\input{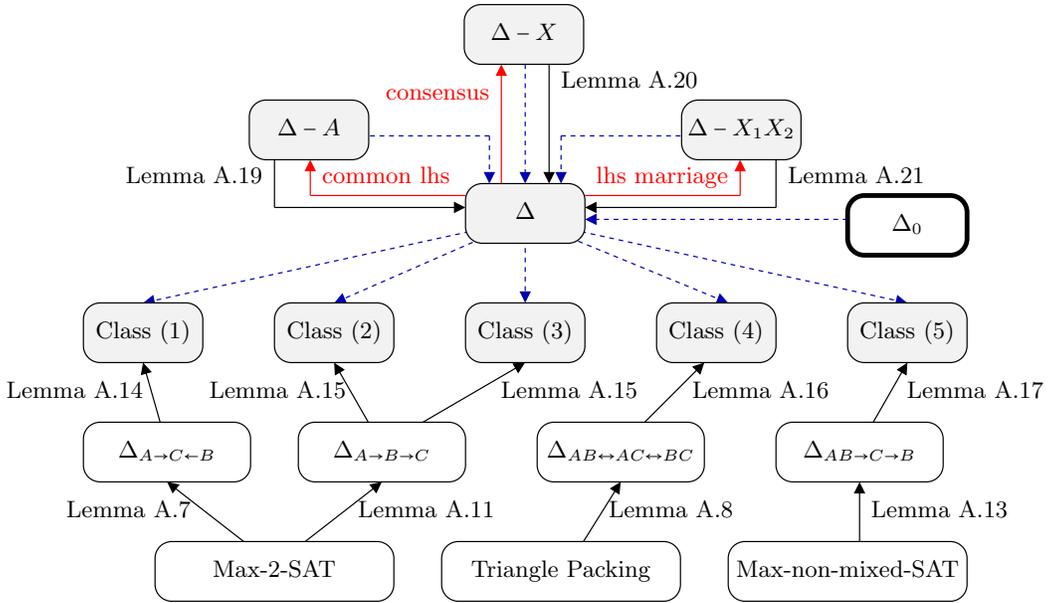}
\vskip0.5em
\caption{An illustration of our proof of the negative side of Theorem~\ref{thm:dichotomy}. We start with an FD set $\depset_0$ and apply simplifications to it, until we get a non trivial set of FDs $\depset$, that we classify into one of five classes. The red arrows represent simplifications and the dashed blue arrows represent renaming. A black arrow represents a reduction that we construct in the lemma that appears next to the arrow.}\label{fig:hardness-srepair}
\end{figure}

As mentioned above, our proof of hardness is based on the concept of a \e{fact-wise reduction}~\cite{DBLP:conf/pods/Kimelfeld12}. We first prove, for each one of the FD sets in Table~\ref{table:special-schemas} (over the schema $R(A,B,C)$), that computing an optimal S-repair is APX-complete. Then, we prove the existence of fact-wise reductions from these FD sets over the schema $R(A,B,C)$ to other sets of FDs over other schemas.

Figure~\ref{fig:hardness-srepair} illustrates our proof. The idea is the following: we start with a table $T_0$ and an FD set $\depset_0$, for which we want to find an optimal S-repair. Then, we apply the simplifications (represented by the red arrows in the figure) until it is no longer possible. We prove, for each one of the simplifications (common lhs, consensus FD and lhs marriage), that there is a fact-wise reduction from the problem after the simplification to the problem before the simplification. The problem is APX-complete if at this point we have a set of FDs that is not trivial. In this case, we pick two local minima from the FD set. We recall that an FD $X\rightarrow Y$ in $\depset$ is a local minimum if there is no FD $Z\rightarrow W$ in $\depset$ such that $Z\subset X$. These two local minima belong to one of five classes, that we discuss in Section~\ref{sec:subset-repairs} and later in this section.  Finally, we prove for each one of the classes, that there is a fact-wise reduction to this class from one of four FD sets. To prove APX-hardness for these four sets (over the relation schema $R(A,B,C)$) we construct reductions from problems that are known to be APX-hard. A black arrow in Figure~\ref{fig:hardness-srepair} represents a reduction that we construct in the lemma that appears next to the arrow.

\subsubsection{Hard Schemas}

We start by proving that computing an optimal S-repair for the FD sets in Table~\ref{table:special-schemas} is APX-complete. Proposition~\ref{prop:subset-approx} implies that the problem is in APX for each one of these sets. Thus, it is only left to show that the problem is also APX-hard. Gribkoff et al.~\cite{GVSBUDA14} prove that the MPD problem is NP-hard for both $\sabc$ and $\stfd$. Their hardness proof also holds for the problem of computing an optimal S-repair. More formally, the following hold.

\begin{citedlemma}{GVSBUDA14}\label{lemma:ac-bc-np-hard}
  Computing an optimal S-repair for $\stfd$ is NP-hard.
\end{citedlemma}

\begin{citedlemma}{GVSBUDA14}\label{lemma:ab-bc-np-hard}
  Computing an optimal S-repair for $\sabc$ is NP-hard.
\end{citedlemma}

They prove both of these result by showing a reduction from the MAX-$2$-SAT problem: given a $2$-CNF formula $\phi$, determine what is the maximum number of clauses in $\phi$ which can be simultaneously satisfied. In their reductions it holds that the maximum number of clauses that can be simultaneously satisfied is exactly the size of an optimal S-repair of the constructed table $T$ (which is unweighted and duplicate-free). The problem MAX-$2$-SAT is known to be APX-hard. However, this is not enough to prove APX-hardness for our problem, as we would like to approximate the minimum number of tuples to delete, rather than the size of the optimal S-repair. Thus, we will now strengthen their results by proving that the complement problem of MAX-$2$-SAT is APX-hard. Clearly, their reductions are strict reductions from the complement problem (as the number of clauses that are not satisfied is exactly the number of tuples that are deleted from the table), thus this will complete our proof of APX-completeness for $\sabc$ and $\stfd$.

\begin{lemma}\label{lemma:max-2-sat-complement}
The complement problem of MAX-$2$-SAT is APX-hard.
\end{lemma}

\begin{proof}
It is known that there is always an assignment that satisfies at least half of the clauses in the formula. Thus, the solution to the complement problem contains at most half of the clauses. Let $\psi$ be a $2$-CNF formula. Let $m$ be the number of clauses in $\psi$, let $OPT_{P}$ be an optimal solution to MAX-$2$-SAT and let $OPT_{CP}$ be an optimal solution to the complement problem. Then, the following hold: \e{(a)} $|OPT_{P}|\ge \frac{1}{2}\cdot m$, \e{(b)} $|OPT_{CP}|\le \frac{1}{2}\cdot m$, and \e{(c)} $\frac{|OPT_{CP}|}{|OPT_{P}|}\le 1$. Now, let us assume, by way of contradiction, that the complement problem is not APX-hard. That is, for every $\epsilon>0$, there exists a $(1+\epsilon)$-optimal solution to that problem. Let $S_{CP}$ be such a solution. Then,
\begin{equation}
    \begin{aligned} 
        |S_{CP}|-|OPT_{CP}|\le \epsilon\cdot |OPT_{CP}|
    \end{aligned}
\end{equation}

We will now show that in this case, $S_{P}$ (which consists of the clauses that do not belong to $S_{CP}$, thus $|S_P|=m-|S_{CP}|$) is an $\epsilon$-optimal solution for MAX-$2$-SAT.

\begin{equation}
    \begin{aligned} 
        |S_{P}|-|OPT_{P}|=(m-|S_{CP}|)-(m-|OPT_{CP}|)=|OPT_{CP}|-|S_{CP}|\ge -\epsilon\cdot |OPT_{CP}|\ge -\epsilon\cdot |OPT_{P}|
    \end{aligned}
\end{equation}

It follows that for each $\epsilon>0$, we can get a $(1+\epsilon)$-optimal solution for MAX-$2$-SAT by finding a $(1+\epsilon)$-optimal solution to the complement problem and selecting all of the clauses that do not belong to this solution, which is a contradiction to the fact that MAX-$2$-SAT is APX-hard.
\end{proof}

We can now conclude that the following hold.

\begin{lemma}\label{lemma:ac-bc-hard}
  Computing an optimal S-repair for $\stfd$ is APX-complete.
\end{lemma}
\begin{proof}
This is straightforward based on Lemma~\ref{lemma:ac-bc-np-hard}, Lemma~\ref{lemma:max-2-sat-complement} and our observation that the reduction of Gribkoff et al.~\cite{GVSBUDA14} is a strict reduction from the complement problem of MAX-$2$-SAT.
\end{proof}

\begin{lemma}\label{lemma:ab-bc-hard}
  Computing an optimal S-repair for $\sabc$ is APX-complete.
\end{lemma}
\begin{proof}
This is straightforward based on Lemma~\ref{lemma:ab-bc-np-hard}, Lemma~\ref{lemma:max-2-sat-complement} and our observation that the reduction of Gribkoff et al.~\cite{GVSBUDA14} is a strict reduction from the complement problem of MAX-$2$-SAT.
\end{proof}

\begin{figure}
\centering
\input{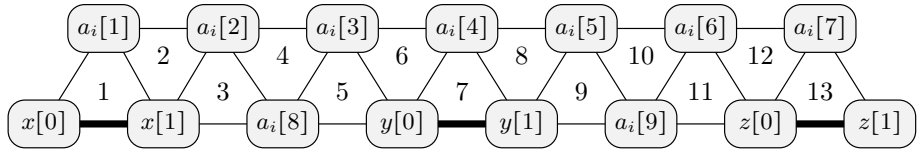}
\vskip0.5em
\caption{An illustration of the reduction used by Amini et al.~\cite{DBLP:journals/tcs/AminiPS09} to prove APX-hardness for the problem of finding the maximum number of edge disjoint triangles in a tripartite graph with a bounded degree.}\label{fig:triangles}
\end{figure}

Next, we prove that computing an optimal S-repair for $\str$ is APX-complete as well. To do that, we construct a reduction from the problem of finding the maximum number of edge-disjoint triangles in a tripartite graph with a bounded degree B. Amini et al.~\cite{DBLP:journals/tcs/AminiPS09} (who refer to this problem as MECT-B) proved that this problem is APX-complete. Again, in our reduction, it holds that the maximum number edge-disjoint triangles is exactly the size of an optimal S-repair of the constructed table $T$. Thus, our reduction is a strict reduction from the complement problem. Hence, we first prove that the complement problem of MECT-B is APX-hard for tripartite graphs that satisfy a specific property. We use the reduction of Amini et al.~\cite{DBLP:journals/tcs/AminiPS09} to prove that. They build a reduction from the problem of finding a maximum bounded covering by $3$-sets: given a collection of subsets of a given set that contain exactly three elements each, such that each element appears in at most B subsets, find the maximum number of disjoint subsets. In their reduction, they construct a tripartite graph, such that for each subset $S_i=(x,y,z)$, they add to the graph the structure from Figure~\ref{fig:triangles}. Note that the nodes $a_i[1]\dots a_i[9]$ are unique for this subset, while the nodes $x[0],x[1],y[0],y[1],z[0],z[1]$ will appear only once in the graph, even if they appear in more than one subset. Thus, we can build a set of edge-disjoint triangles for the constructed tripartite graph by selecting, for each subset, six out of the thirteen triangles (the even ones). This is true since the even triangles do not share an edge with any other triangle. We can now conclude, that in their reduction, they construct a tripartite graph with the following property: the maximum number of edge-disjoint triangles in the graph is at least $\frac{6}{13}$ of the total number of triangles. We denote a graph that satisfies this property by $\frac{6}{13}$-tripartite graph. Thus, we can conclude the following:

\begin{lemma}\label{lemma:6-13-apx-appendix}
The problem MECT-B for $\frac{6}{13}$-tripartite graphs is APX-hard.
\end{lemma}

We will now prove that the complement problem is APX-hard as well.

\begin{lemma}\label{lemma:MECT-B-complement}
The complement problem of MECT-B for $\frac{6}{13}$-tripartite graphs is APX-hard.
\end{lemma}

\begin{proof}
Let $g$ be a $\frac{6}{13}$-tripartite graph. Then, the solution to the complement problem contains at most $\frac{7}{13}$ of the total number of triangles. Let $m$ be the number of triangles in $g$, let $OPT_{P}$ be an optimal solution to MECT for $g$ and let $OPT_{CP}$ be an optimal solution to the complement problem. Thus, the following hold: \e{(a)} $|OPT_{P}|\ge \frac{6}{13}\cdot m$, \e{(b)} $|OPT_{CP}|\le \frac{7}{13}\cdot m$, and \e{(c)} $\frac{|OPT_{CP}|}{|OPT_{P}|}\le \frac{7}{6}$. Now, let us assume, by way of contradiction, that the complement problem is not APX-hard. That is, for every $\epsilon>0$, we can find a $(1+\epsilon)$-optimal solution to that problem. Let $S_{CP}$ be such a solution. Then,
\begin{equation}
    \begin{aligned} 
        |S_{CP}|-|OPT_{CP}\|le \epsilon\cdot |OPT_{CP}|
    \end{aligned}
\end{equation}

We will now show that in this case, $S_{P}$ (which consists of the triangles that do not belong to $S_{CP}$, thus $|S_P|=m-|S_{CP}|$) is a $\frac{7}{6}\epsilon$-optimal solution for MECT-B.

\begin{equation}
    \begin{aligned} 
        |S_{P}|-|OPT_{P}|=(m-|S_{CP}|)-(m-|OPT_{CP}|)=|OPT_{CP}|-|S_{CP}|\ge -\epsilon\cdot |OPT_{CP}|\ge -\frac{7}{6}\epsilon\cdot |OPT_{P}|
    \end{aligned}
\end{equation}

It follows that for each $\epsilon'>0$, we can get an $\epsilon'$-optimal solution for MECT by finding a $(1+\frac{6}{7}\epsilon')$-optimal solution to the complement problem, which is a contradiction to the fact that MECT-B is APX-hard for $g$, as implied by Lemma~\ref{lemma:6-13-apx-appendix}.
\end{proof}

Next, we introduce our reduction from MECT-B to the problem of computing an optimal S-repair for $\str$.

\begin{lemma}\label{lemma:abc-acb-bca-hard}
  Computing an optimal S-repair for $\str$ is APX-complete.
\end{lemma}

\begin{proof}
We construct a reduction from the problem of finding the maximum number of edge-disjoint triangles in a tripartite graph with a bounded degree. The input to this problem is a tripartite graph $g$ with a bounded degree B.
 The goal is to determine what is the maximum number of edge-disjoint triangles in $g$ (that is, no two triangles share an edge).
 We assume that $g$ contains three sets of nodes: $\set{a_1,\dots,a_n}$, $\set{b_1,\dots,b_l}$ and $\set{c_1,\dots,c_r}$. Given such an input, 
 we will construct the input 
 $T$ for our problem as follows.
 For each triangle in $g$ that consists of the nodes $a_i$, $b_j$, and $c_k$, $I$ will contain a tuple $(a_i,b_j,c_k)$.
 We will now prove that there are at least $m$ edge-disjoint triangles in $g$ if and only if there is a consistent subset of $T$ that contains at least $m$ tuples.
 
\paragraph*{The ``if'' direction}
there is a consistent subset $J$ of $T$ that contains at least $m$ tuples. The FD $AB\rightarrow C$ implies that a consistent subset cannot contain two tuples $(a_i,b_j,c_{k_1})$ and $(a_i,b_j,c_{k_2})$ such that $c_{k_1}\neq c_{k_2}$. Moreover, the FD $AC\rightarrow B$ implies that it cannot contain two tuples $(a_i,b_{j_1},c_k)$ and $(a_i,b_{j_2},c_k)$ such that $b_{j_1}\neq b_{j_2}$, and the FD $BC\rightarrow A$ implies that it cannot contain two tuples $(a_{i_1},b_j,c_k)$ and $(a_{i_2},b_j,c_k)$ such that $a_{i_1}\neq a_{i_2}$. Thus, the two triangles $(a_{i_1},b_{j_1},c_{k_1})$ and $(a_{i_2},b_{j_2},c_{k_2})$ in $g$ that correspond to two tuples $(a_{i_1},b_{j_1},c_{k_1})$ and $(a_{i_1},b_{j_1},c_{k_1})$ in $J$, will not share an edge (they can only share a single node). Hence, there are at least $m$ edge-disjoint triangles in $g$.

\paragraph*{The ``only if'' direction}
Assume that there are at least $m$ edge-disjoint triangles $\set{t_1,\dots,t_m}$ in $g$. We can build a consistent subset $J$ of $T$ as follows: for each triangle $(a_i,b_j,c_k)$  in $\set{t_1,\dots,t_m}$ we will add the tuple $(a_i,b_j,c_k)$ to $J$. Thus, $J$ will contain $m$ tuples. It is only left to show that $J$ is consistent. Let us assume, by way of contradiction, that $J$ is not consistent. That is, there are two tuples $(a_1,b_1,c_1)$ and $(a_2,b_2,c_2)$ in $J$ that violate an FD in $\str$. If the tuples violate the FD $AB\rightarrow C$, it holds that $a_1=a_2$ and $b_1=b_2$. Thus, the corresponding two triangles from $\set{t_1,\dots,t_m}$ share the edge $(a_1,b_1)$, which is a contradiction to the fact that $\set{t_1,\dots,t_m}$ is a set of edge-disjoint triangles. Similarly, if the tuples violate the FD $AC\rightarrow B$, then the corresponding two triangles share an edge $(a_1,c_1)$, and if they violate the FD $BC\rightarrow A$, the corresponding two triangles share an edge $(b_1,c_1)$. To conclude, there exists a consistent subset of $T$ that contains at least $m$ facts, and that concludes our proof.

Clearly, our reduction is a strict reduction from the complement problem of MECT-B (as the number of triangles that are not part of the set of edge-disjoint triangles is exactly the number of tuples that are deleted from $T$), thus Lemma~\ref{lemma:MECT-B-complement} implies that our problem is indeed APX-complete.
\end{proof}

Finally, we construct a reduction from the problem MAX-non-mixed-SAT to the problem of computing an optimal S-repair for $\stk$. Note that the following holds.

\begin{lemma}\label{lemma:max-non-mixed-sat-complement}
The complement problem of MAX-non-mixed-SAT is APX-hard.
\end{lemma}

\begin{proof}
The proof is identical to the proof of Lemma~\ref{lemma:max-2-sat-complement}.
\end{proof}

Thus, if we construct a reduction that is a strict reduction from the complement problem, this will conclude our proof of APX-completeness for $\stk$.

\begin{lemma}\label{lemma:abc-cb-hard}
Computing an optimal S-repair for $\stk$ is APX-complete.
\end{lemma}

\begin{proof}
We construct a reduction from MAX-Non-Mixed-SAT to the problem of computing an optimal S-repair for $\stk$.  The input to the first problem is a formula $\psi$ with the free variables 
 $x_1,\ldots,x_n$, such that  $\psi$
 has the form $c_1 \wedge  \cdots \wedge c_m$ where each $c_j$ is a clause. Each clause is a conjunction of variables from one of the following sets: \e{(a)} $\{x_i : i=1,\ldots,n\}$ or \e{(b)} $\{\neg x_i : i=1,\ldots,n\}$ (that is, each clause either contains only positive variables or only negative variables).
 The goal is to determine what is the maximum number of clauses in the formula $\psi$ that can be simultaneously satisfied.
 Given such an input, 
 we will construct the input 
 $T$ for our problem as follows.
 For each $i=1,\ldots, n$ and $j=1,\ldots, m$, $T$ will contain the following tuples:
 \begin{itemize}
 \item
 $(c_j,\val{1},x_i)$, if $c_j$ contains only positive variables and $x_i$ appears in $c_j$.
 \item
$(c_j,\val{0},x_i)$, if $c_j$ contains only negative variables and $\neg x_i$ appears in $c_j$.
 \end{itemize}
 The weight of each tuple will be $1$ (that is, $T$ is an unweighted, duplicate-free table).
 We will now prove that there is an assignment that satisfies at least $k$ clauses in $\psi$ if and only if there is a consistent subset of the constructed table $T$ that contains at least $k$ tuples.
 
\paragraph*{The ``if'' direction}
Assume that there is a consistent subset $J$ of $T$ that contains at least $k$ tuples. The FD $AB\rightarrow C$ implies that no consistent subset of $T$ contains two tuples $(c_j,b_j,x_{i_1})$ and $(c_j,b_j,x_{i_2})$ such that $x_{i_1}\neq x_{i_2}$. Thus, each consistent subset contains at most one tuple $(c_j,b_j,x_i)$ for each $c_j$. We will now define an assignment $\tau$ as follows: $\tau(x_i)\eqdef b_j$ if there exists a tuple $(c_j,b_j,x_i)$ in $J$ for some $c_j$. Note that the FD $C\rightarrow B$ implies that no consistent subset contains two tuples $(c_{j_1},\val{1},x_i)$ and $(c_{j_2},\val{0},x_i)$, thus the assignment is well defined. Finally, as mentioned above, $J$ contains a tuple $(c_j,b_j,x_i)$ for $k$ clauses $c_j$ from $\psi$. If $x_i$ appears in $c_j$ without negation, it holds that $b_1=1$, thus $\tau(x_i)\eqdef 1$ and $c_j$ is satisfied. Similarly, if $x_i$ appears in $c_j$ with negation, it holds that $b_j=0$, thus $\tau(x_i)\eqdef 0$ and $c_j$ is satisfied. Thus, each one of these $k$ clauses is satisfied by $\tau$ and we conclude that there exists an assignment that satisfies at least $k$ clauses in $\psi$.

\paragraph*{The ``only if'' direction}
Assume that $\tau: \{x_1,\ldots,x_n\} \rightarrow \{0,1\}$ is an assignment that satisfies at least $k$ clauses in $\psi$.
 We claim that there exists a consistent subset of $T$ that contains at least $k$ tuples. Since $\tau$ satisfies at least $k$ clauses, for each one of these clauses $c_j$ there exists a variable $x_i\in c_j$, such that $\tau(x_i)=1$ if $x_i$ appears in $c_j$ without negation or $\tau(x_i)=0$ if it appears in $c_j$ with negation. Let us build a consistent subset $J$ as follows. For each $c_j$ that is satisfied by $\tau$ we will choose exactly one variable $x_i$ that satisfies the above and add the tuple $\rtk(c_j,b_j,x_i)$ (where $\tau(x_i)=b_j$) to $J$. Since there are at least $k$ satisfied clauses, $J$ will contain at least $k$ tuples, thus it is only left to prove that $J$ is consistent. Let us assume, by way of contradiction, that $J$ is not consistent. Since $J$ contains one tuple for each satisfied $c_j$, no two tuples violate the FD $AB\rightarrow C$. Thus, $J$ contains two tuples $(c_{j_1},\val{1},x_i)$ and $(c_{j_2},\val{0},x_i)$, but this is a contradiction to the fact that $\tau$ is an assignment (that is, it cannot be the case that $\tau(x_i)=1$ and $\tau(x_i)=0$ as well). 
 
Clearly, our reduction is a strict reduction from the complement problem of MAX-non-mixed-SAT (as the number of clauses that are not satisfied is exactly the number of tuples that are deleted from $T$), thus Lemma~\ref{lemma:abc-cb-hard} implies that our problem is indeed APX-complete.
\end{proof}

\subsubsection{Fact-Wise Reductions}

Let $R$ be a schema and let $\depset$ be a set of FDs over $R$. Recall that an FD set $\depset$ is a
chain if for every two FDs $X_1\rightarrow Y_1$ and
$X_2 \rightarrow Y_2$ it is the case that $X_1 \subseteq X_2$ or
$X_2 \subseteq X_1$. Note that as long as $\depset$ is a chain, it either contains a common lhs, or a consensus FD. Thus, if we reach a point where we cannot apply any simplifications to the problem, the set of FDs is not a chain. In the rest of this section, we will use the following definition: an FD $X\rightarrow Y$ is a \e{local minimum} of $\depset$ if there is no other FD $Z\rightarrow W$ in $\depset$ such that $Z\subset X$. If $\depset$ is not a chain set of FDs, then it contains at least two distinct local minima $X_1\rightarrow Y_1$ and $X_2\rightarrow Y_2$ (by distinct we mean that $X_1\neq X_2$). Thus, if no simplification can be applied to $\depset$, one of the following holds:
\begin{itemize}
\item $(\closure_{\depset}(X_1)\setminus X_1)\cap \closure_{\depset}(X_2)= \emptyset$ and $(\closure_{\depset}(X_2)\setminus X_2)\cap \closure_{\depset}(X_1) = \emptyset$.
\item $(\closure_{\depset}(X_1)\setminus X_1) \cap (\closure_{\depset}(X_2) \setminus X_2)\neq \emptyset$, $(\closure_{\depset}(X_1)\setminus X_1)\cap X_2 = \emptyset$ and $(\closure_{\depset}(X_2)\setminus X_2)\cap X_1 = \emptyset$. 
\item $(\closure_{\depset}(X_1)\setminus X_1)\cap X_2 \neq \emptyset$ and $(\closure_{\depset}(X_2) \setminus X_2)\cap X_1 = \emptyset$. 
\item $(\closure_{\depset}(X_1)\setminus X_1)\cap X_2 \neq \emptyset$ and $(\closure_{\depset}(X_2) \setminus X_2)\cap X_1 \neq \emptyset$ and also $(X_1\setminus X_2)\subseteq (\closure_{\depset}(X_2)\setminus X_2)$ and $(X_2\setminus X_1)\subseteq (\closure_{\depset}(X_1)\setminus X_1)$. In this case, $\depset$ contains at least one more local minimum. Otherwise, for every FD $Z\rightarrow W$ in $\depset$ it holds that either $X_1\subseteq Z$ or $X_2\subseteq Z$. If $X_1\cap X_2\neq\emptyset$, then $\depset$ contains a common lhs (an attribute from $X_1\cap X_2$). If $X_1\cap X_2=\emptyset$, then $\depset$ contains an lhs marriage. 
\item $(\closure_{\depset}(X_1)\setminus X_1)\cap X_2 \neq \emptyset$ and $(\closure_{\depset}(X_2) \setminus X_2)\cap X_1 \neq \emptyset$ and also $(X_2\setminus X_1)\not\subseteq (\closure_{\depset}(X_1)\setminus X_1)$. 
\end{itemize}
We will now prove that for each one of these cases there is a fact-wise reduction from one of the hard schemas we discusses above.

\begin{lemma}\label{lemma:disjoint-reduction}
Let $R$ be a schema and let $\depset$ be an FD set over $R$ that does not contain trivial FDs. Suppose that $\depset$ contains two distinct local minima $X_1\rightarrow Y_1$ and $X_2\rightarrow Y_2$, and the following hold:
\begin{itemize}
\item $(\closure_{\depset}(X_1)\setminus X_1)\cap \closure_{\depset}(X_2)= \emptyset$,
\item $(\closure_{\depset}(X_2)\setminus X_2)\cap \closure_{\depset}(X_1) = \emptyset$.
\end{itemize} 
Then, there is a fact-wise reduction from $(R(A,B,C),\stfd)$ to $(R,\depset)$.
\end{lemma}

\begin{proof}
We define a fact-wise reduction $\Pi:(R(A,B,C),\stfd) \rightarrow (R,\depset)$, using the FDs $X_1\rightarrow Y_1$ and $X_2\rightarrow Y_2$ and the constant $\odot \in \consts$.
Let $\tup t= (a,b,c)$ be a tuple over $R(A,B,C)$ and let $\set{A_1,\dots,A_n}$ be the set of attributes in $R$.
We define $\Pi $ as follows:
\[
\Pi (\tup t).A_k \eqdef
\begin{cases}
\odot & \mbox{$A_k\in X_1\cap X_2$} \\ 
 a& \mbox{$A_k\in X_1 \setminus X_2$}\\ 
 b& \mbox{$A_k\in X_2 \setminus X_1$} \\
 \langle a,c\rangle& \mbox{$A_k\in \closure_{\depset}(X_1)\setminus X_1$} \\
 \langle b,c\rangle& \mbox{$A_k\in \closure_{\depset}(X_2)\setminus X_2$} \\
 \langle a,b\rangle& \mbox{otherwise}
\end{cases}
\]
It is left to show that $\Pi$ is a fact-wise reduction.
To do so, we prove that $\Pi$ is well defined, injective and preserves consistency and inconsistency.

\partitle{$\mathbf{\Pi}$ is well defined}
This is straightforward from the definition and the fact that $(\closure_{\depset}(X_1)\setminus X_1)\cap \closure_{\depset}(X_2)= \emptyset$ and $(\closure_{\depset}(X_2)\setminus X_2)\cap \closure_{\depset}(X_1) = \emptyset$.

\partitle{$\mathbf{\Pi}$ is injective}
Let $\tup t,\tup t'$ be two tuples, such that $\tup t=(a,b,c)$ and $\tup t' = (a',b',c')$.
Assume that $\Pi (\tup t) = \Pi (\tup t')$.  Let us denote
$\Pi (\tup t)=(x_1,\dots, x_n)$ and $\Pi (\tup t')=(x'_1,\dots, x'_n)$.
Note that $X_1 \setminus X_2$ and $X_2 \setminus X_1$ are not empty since $X_1\neq X_2$. Moreover, since both FDs are minimal, $X_1\not\subset X_2$ and $X_2\not\subset X_1$.
Therefore, there are $l$ and $p$ such that $\Pi (\tup t).A_l = a$, $\Pi (\tup t).A_p= b$. Furthermore, since $X_1\rightarrow Y_1$ and $X_2\rightarrow Y_2$ are not trivial, there are $m$ and $n$ such that $\Pi (\tup t).A_m=\langle a,c\rangle$ and $\Pi (\tup t).A_n=\langle b,c\rangle$.
Hence, $\Pi (\tup t) = \Pi (\tup t')$ implies that $\Pi (\tup t).A_l = \Pi (\tup t').A_l$, $\Pi (\tup t).A_p = \Pi (\tup t').A_p$, $\Pi (\tup t).A_m = \Pi (\tup t').A_m$ and also $\Pi (\tup t).A_n= \Pi (\tup t').A_n$. We obtain that
$a=a'$, $b=b'$ and $c=c'$, which implies $\tup t=\tup t'$.

\partitle{$\mathbf{\Pi}$ preserves consistency}
Let $\tup t=(a,b,c)$ and $\tup t' = (a',b',c')$ be two distinct tuples.
We contend that  the set $\{\tup t,\tup t'\}$ is consistent w.r.t. $\stfd$ if and only if the set $\{\Pi(\tup t),\Pi(\tup t')\}$ is consistent w.r.t. $\depset$.
\paragraph*{The ``if'' direction}
Assume that $\{\tup t,\tup t'\}$ is consistent w.r.t $\stfd$. We prove that  $\{\Pi(\tup t),\Pi(\tup t')\}$ is consistent w.r.t $\depset$. First, note that each FD that contains an attribute $A_k\not\in (\closure_{\depset}(X_1)\cup \closure_{\depset}(X_2))$ on its left-hand side is satisfied by $\{\Pi(\tup t),\Pi(\tup t')\}$, since $\tup t$ and $\tup t'$ cannot agree on both $A$ and $B$ (otherwise, the FD $A\rightarrow C$ implies that $\tup t=\tup t'$). Thus, from now on we will only consider FDs that do not contain an attribute $A_k\not\in (\closure_{\depset}(X_1)\cup \closure_{\depset}(X_2))$ on their left-hand side. The FDs in $\stfd$ imply that if $\tup t$ and $\tup t'$ agree on one of $\set{A,B}$ then they also agree on $C$, thus one of the following holds:
\begin{itemize}
\item $a\neq a'$, $b= b'$ and $c= c'$. In this case, $\Pi(\tup t)$ and $\Pi(\tup t')$ only agree on the attributes $A_k$ such that $A_k\in X_1\cap X_2$ or $A_k\in X_2\setminus X_1$ or $A_k\in \closure_{\depset}(X_2)\setminus X_2$. That is, they only agree on the attributes $A_k$ such that $A_k\in \closure_{\depset}(X_2)$. Thus, each FD that contains an attribute $A_k\not\in \closure_{\depset}(X_2)$ on its left-hand side is satisfied. Moreover, any FD that contains only attributes $A_k\in \closure_{\depset}(X_2)$ on its left-hand side, also contains only attributes $A_k\in \closure_{\depset}(X_2)$ on its right-hand side (by definition of a closure), thus $\Pi(\tup t)$ and $\Pi(\tup t')$ agree on both the left-hand side and the right-hand side of such FDs and $\{\Pi(\tup t),\Pi(\tup t')\}$ satisfies all the FDs in $\depset$.
\item $a= a'$, $b\neq b'$ and $c= c'$. This case is symmetric to the previous one, thus a similar proof applies for this case as well.
\item $a\neq a'$, $b\neq b'$. In this case, $\Pi(\tup t)$ and $\Pi(\tup t')$ only agree on the attributes $A_k$ such that $A_k\in X_1\cap X_2$. Since $X_1\rightarrow Y_1$ and $X_2\rightarrow Y_2$ are local minima, there is no FD in $\depset$ that contains only attributes $A_k$ such that $A_k\in X_1\cap X_2$ on its left-hand side (as if there is an FD $Z\rightarrow W$ in $\depset$, such that $Z\subseteq X_1\cap X_2$, then $Z\subset X_1$ in contradiction to the fact that $X_1$ is a local minimum). Thus, $\Pi(\tup t)$ and $\Pi(\tup t')$ do not agree on the left-hand side of any FD in $\depset$ and $\{\Pi(\tup t),\Pi(\tup t')\}$ is consistent w.r.t. $\depset$.
\end{itemize}
This concludes our proof of the ``if'' direction.

\paragraph*{The ``only if'' direction}
Assume $\set{\tup t,\tup t'}$ is inconsistent w.r.t. $\stfd$. We prove that $\{\Pi(\tup t),\Pi(\tup t')\}$ is inconsistent w.r.t. $\depset$.
Since $\set{\tup t,\tup t'}$ is inconsistent w.r.t. $\stfd$ it either holds that $a=a'$ and $c\neq c'$ or $b=b'$ and $c\neq c'$ (or both). In the first case, $\Pi(\tup t)$ and $\Pi(\tup t')$ agree on the attributes on the left-hand side of the FD $X_1\rightarrow Y_1$, but do not agree on at least one attribute on its right-hand side (since the FD is not trivial). Similarly, in the second case, $\Pi(\tup t)$ and $\Pi(\tup t')$ agree on the attributes on the left-hand side of the FD $X_2\rightarrow Y_2$, but do not agree on at least one attribute on its right-hand side. Thus, $\{\Pi(\tup t),\Pi(\tup t')\}$ does not satisfy at least one of these FDs and $\{\Pi(\tup t),\Pi(\tup t')\}$ is inconsistent w.r.t. $\depset$.
\end{proof}

\begin{lemma}\label{lemma:disjoint-left-reduction}
Let $R$ be a schema and let $\depset$ be an FD set over $R$ that does not contain trivial FDs. Suppose that $\depset$ contains two distinct local minima $X_1\rightarrow Y_1$ and $X_2\rightarrow Y_2$, and one of the following holds:
\begin{itemize}
\item $((\closure_{\depset}(X_1)\setminus X_1) \cap ((\closure_{\depset}(X_2)\setminus X_2)\neq \emptyset$, $((\closure_{\depset}(X_1)\setminus X_1)\cap X_2 = \emptyset$ and $((\closure_{\depset}(X_2) \setminus X_2)\cap X_1 = \emptyset$,
\item $((\closure_{\depset}(X_1)\setminus X_1)\cap X_2 \neq \emptyset$ and $((\closure_{\depset}(X_2) \setminus X_2)\cap X_1 = \emptyset$.
\end{itemize} 
Then, there is a fact-wise reduction from $(R(A,B,C),\sabc)$ to $(R,\depset)$.
\end{lemma}

\begin{proof}
We define a fact-wise reduction $\Pi:(R(A,B,C),\sabc) \rightarrow (R,\depset)$, using $X_1\rightarrow Y_1$ and $X_2\rightarrow Y_2$ and the constant $\odot \in \consts$.
Let $\tup t = (a,b,c)$ be a tuple over $R(A,B,C)$ and let $\set{A_1,\dots,A_n}$ be the set of attributes in $R$.
We define $\Pi $ as follows:
\[
\Pi (\tup t).A_k \eqdef
\begin{cases}
\odot & \mbox{$A_k\in X_1\cap X_2$} \\ 
 a& \mbox{$A_k\in X_1 \setminus X_2$}\\ 
 b& \mbox{$A_k\in X_2 \setminus X_1$} \\
 \langle a,c\rangle& \mbox{$A_k\in \closure_{\depset}(X_1)\setminus X_1\setminus \closure_{\depset}(X_2)$} \\
 \langle b,c\rangle& \mbox{$A_k\in \closure_{\depset}(X_2)\setminus X_2$} \\
 a& \mbox{otherwise}
\end{cases}
\]
It is left to show that $\Pi$ is a fact-wise reduction.
To do so, we prove that $\Pi$ is well defined, injective and preserves consistency and inconsistency.

\partitle{$\mathbf{\Pi}$ is well defined}
This is straightforward from the definition and the fact that $(\closure_{\depset}(X_2) \setminus X_2)\cap X_1 = \emptyset$ in both cases.

\partitle{$\mathbf{\Pi}$ is injective}
Let $\tup t,\tup t'$ be two tuples, such that $\tup t=(a,b,c)$ and $\tup t' = (a',b',c')$.
Assume that $\Pi (\tup t) = \Pi (\tup t')$.  Let us denote
$\Pi (\tup t)=(x_1,\dots, x_n)$ and $\Pi (\tup t')=(x'_1,\dots, x'_n)$.
Note that $X_1 \setminus X_2$ and $X_2 \setminus X_1$ are not empty since $X_1\neq X_2$. Moreover, since both FDs are minimal, $X_1\not\subset X_2$ and $X_2\not\subset X_1$.
Therefore, there are $l$ and $p$ such that $\Pi (\tup t).A_l= a$, $\Pi (\tup t).A_p = b$. Furthermore, since $X_2\rightarrow Y_2$ is not trivial, there is at least one $m$ such that $\Pi (\tup t).A_m=\langle b,c\rangle$.
Hence, $\Pi (\tup t) = \Pi (\tup t')$ implies that $\Pi (\tup t).A_l = \Pi (\tup t').A_l$, $\Pi (\tup t).A_p = \Pi (\tup t').A_p$ and $\Pi (\tup t).A_m = \Pi (\tup t').A_m$. We obtain that
$a=a'$, $b=b'$ and $c=c'$, which implies $\tup t=\tup t'$.

\partitle{$\mathbf{\Pi}$ preserves consistency}
Let $\tup t=(a,b,c)$ and $\tup t'=(a',b',c')$ be two distinct tuples.
We contend that  the set $\{t,t'\}$ is consistent w.r.t. $\sabc$ if and only if the set $\{\Pi(\tup t),\Pi(\tup t')\}$ is consistent w.r.t. $\depset$.
\paragraph*{The ``if'' direction}
Assume that $\set{\tup t,\tup t'}$ is consistent w.r.t $\sabc$. We prove that  $\{\Pi(\tup t),\Pi(\tup t')\}$ is consistent w.r.t $\depset$. First, note that each FD that contains an attribute $A_k\not\in (\closure_{\depset}(X_1)\cup \closure_{\depset}(X_2))$ on its left-hand side is satisfied by $\{\Pi(\tup t),\Pi(\tup t')\}$, since $\tup t$ and $\tup t'$ cannot agree on $A$ (otherwise, the FDs $A\rightarrow B$ and $B\rightarrow C$ imply that $\tup t=\tup t'$). Thus, from now on we will only consider FDs that do not contain an attribute $A_k\not\in (\closure_{\depset}(X_1)\cup \closure_{\depset}(X_2))$ on their left-hand side. One of the following holds:
\begin{itemize}
\item $a\neq a'$, $b= b'$ and $c= c'$. In this case, $\Pi(\tup t)$ and $\Pi(\tup t')$ only agree on the attributes $A_k$ such that $A_k\in X_1\cap X_2$ or $A_k\in X_2\setminus X_1$ or $A_k\in \closure_{\depset}(X_2)\setminus X_2$. That is, they only agree on the attributes $A_k$ such that $A_k\in \closure_{\depset}(X_2)$. Thus, each FD that contains an attribute $A_k\not\in \closure_{\depset}(X_2)$ on its left-hand side is satisfied. Moreover, any FD that contains only attributes $A_k\in \closure_{\depset}(X_2)$ on its left-hand side, also contains only attributes $A_k\in \closure_{\depset}(X_2)$ on its right-hand side (by definition of a closure), thus $\Pi(\tup t)$ and $\Pi(\tup t')$ agree on both the left-hand side and the right-hand side of such FDs and $\{\Pi(\tup t),\Pi(\tup t')\}$ satisfies all the FDs in $\depset$.
\item $a\neq a'$, $b\neq b'$. In this case, $\Pi(\tup t)$ and $\Pi(\tup t')$ only agree on the attributes $A_k$ such that $A_k\in X_1\cap X_2$. Since $X_1\rightarrow Y_1$ and $X_2\rightarrow Y_2$ are minimal, there is no FD in $\depset$ that contains only attributes $A_k$ such that $A_k\in X_1\cap X_2$ on its left-hand side. Thus, $\Pi(\tup t)$ and $\Pi(\tup t')$ do not agree on the left-hand side of any FD in $\depset$ and $\{\Pi(\tup t),\Pi(\tup t')\}$ is consistent w.r.t. $\depset$.
\end{itemize}
This concludes our proof of the ``if'' direction.

\paragraph*{The ``only if'' direction}
Assume $\set{\tup t,\tup t'}$ is inconsistent w.r.t. $\sabc$. We prove that $\{\Pi(\tup t),\Pi(\tup t')\}$ is inconsistent w.r.t. $\depset$.
Since $\set{\tup t,\tup t'}$ is inconsistent w.r.t. $\sabc$, one of the following holds:
\begin{itemize}
\item $a=a'$ and $b\neq b'$. For the first case of this lemma, since $(\closure_{\depset}(X_1)\setminus X_1)\cap (\closure_{\depset}(X_2)\setminus X_2) \neq \emptyset$, at least one attribute $A_k\in (\closure_{\depset}(X_1)\setminus X_1)$ also belongs to $\closure_{\depset}(X_2)\setminus X_2$ and it holds that $\Pi(\tup t).A_k=\langle b,c\rangle$. For the second case of this lemma, since $(\closure_{\depset}(X_1)\setminus X_1)\cap X_2 \neq \emptyset$, at least one attribute $A_k\in (\closure_{\depset}(X_1)\setminus X_1)$ also belongs to $X_2$ and it holds that $\Pi(\tup t).A_k=b$. Moreover, by definition of a closure, the FD $X_1\rightarrow A_k$ is implied by $\depset$. In both cases, the tuples $\Pi(\tup t)$ and $\Pi(\tup t')$ agree on the attributes on the left-hand side of the FD $X_1\rightarrow A_k$, but do not agree on the right-hand side of this FD. If two tuples do not satisfy an FD that is implied by a set $\depset$ of FDs, they also do not satisfy $\depset$, thus $\{\Pi(\tup t),\Pi(\tup t')\}$ is inconsistent w.r.t. $\depset$.
\item $a=a'$, $b=b'$ or $c\neq c'$. For the first case of this lemma, as mentioned above, there is an attribute $A_k\in (\closure_{\depset}(X_1)\setminus X_1)$ such that $\Pi(\tup t).A_k=\langle b,c\rangle$. Moreover, by definition of a closure, the FD $X_1\rightarrow A_k$ is implied by $\depset$. The tuples $\Pi(\tup t)$ and $\Pi(\tup t')$ agree on the attributes on the left-hand side of the FD $X_1\rightarrow A_k$, but do not agree on the right-hand side of this FD. For the second case of this lemma, since $(\closure_{\depset}(X_2)\setminus X_2)\cap X_1 = \emptyset$ and since the FD $X_2\rightarrow Y_2$ is not trivial, there is at least one attribute $A_k\in (\closure_{\depset}(X_2)\setminus X_2)$ such that $\Pi(\tup t).A_k=\langle b,c \rangle$. Furthermore, the FDs $X_2\rightarrow A_k$ is implied by $\depset$. The tuples $\Pi(\tup t)$ and $\Pi(\tup t')$ again agree on the attributes on the left-hand side of the FD $X_2\rightarrow A_k$, but do not agree on the right-hand side of this FD. If two tuples do not satisfy an FD that is implied by a set $\depset$ of FDs, they also do not satisfy $\depset$, thus $\{\Pi(\tup t),\Pi(\tup t')\}$ is inconsistent w.r.t. $\depset$.
\item $a\neq a'$, $b=b'$ and $c\neq c'$. In this case, $\Pi(\tup t)$ and $\Pi(\tup t')$ agree on the attributes on the left-hand side of the FD $X_2\rightarrow Y_2$, but do not agree on the right-hand side of this FD (since the FD is not trivial and contains at least one attribute $A_k$ such that $\Pi(\tup t).A_k=\langle b,c\rangle$ on its right-hand side). Therefore, $\{\Pi(\tup t),\Pi(\tup t')\}$ is inconsistent w.r.t. $\depset$.
\end{itemize}
This concludes our proof of the ``only if'' direction.
\end{proof}

\begin{lemma}\label{lemma:triangle-reduction}
Let $R$ be a schema and let $\depset$ be an FD set over $R$ that does not contain trivial FDs. Suppose that $\depset$ contains three distinct local minima $X_1\rightarrow Y_1$, $X_2\rightarrow Y_2$ and $X_3\rightarrow Y_3$.
Then, there is a fact-wise reduction from $(R(A,B,C),\str)$ to $(R,\depset)$.
\end{lemma}

\begin{proof}
We define a fact-wise reduction $\Pi:(R(A,B,C),\str) \rightarrow (R,\depset)$, using $X_1\rightarrow Y_1$, $X_2\rightarrow Y_2$ and $X_3\rightarrow Y_3$ and the constant $\odot \in \consts$.
Let $\tup t=(a,b,c)$ be a tuple over $R(A,B,C)$ and let $\set{A_1,\dots,A_n}$ be the set of attributes in $R$.
We define $\Pi $ as follows:
\[
\Pi (\tup t).A_k \eqdef
\begin{cases}
 \odot& \mbox{$A_k\in X_1\cap X_2\cap X_3$}\\
 a& \mbox{$A_k\in (X_1\cap X_2)\setminus X_3$}\\
 b& \mbox{$A_k\in (X_1\cap X_3)\setminus X_2$}\\ 
 c& \mbox{$A_k\in (X_2\cap X_3)\setminus X_1$}\\ 
 \langle a,b\rangle& \mbox{$A_k\in X_1\setminus X_2\setminus X_3$} \\
 \langle a,c\rangle& \mbox{$A_k\in X_2\setminus X_1\setminus X_3$} \\
 \langle b,c\rangle& \mbox{$A_k\in X_3\setminus X_1\setminus X_2$} \\
 \langle a,b,c\rangle& \mbox{otherwise}
\end{cases}
\]
It is left to show that $\Pi$ is a fact-wise reduction.
To do so, we prove that $\Pi$ is well defined, injective and preserves consistency and inconsistency.

\partitle{$\mathbf{\Pi}$ is well defined}
This is straightforward from the definition.

\partitle{$\mathbf{\Pi}$ is injective}
Let $\tup t,\tup t'$ be two tuples, such that $\tup t=(a,b,c)$ and $\tup t'=(a',b',c')$.
Assume that $\Pi (\tup t) = \Pi (\tup t')$.  Let us denote
$\Pi (\tup t)=(x_1,\dots, x_n)$ and $\Pi (\tup t')=(x'_1,\dots, x'_n)$.
Note that $X_1$ contains at least one attribute that does not belong to $X_3$ (otherwise, it holds that $X_1\subseteq X_3$, which is a contradiction to the fact that $X_3$ is minimal). Thus, there exists an attribute $A_l$ such that either $\Pi (\tup t).A_l=a$ or $\Pi (\tup t).A_l=\langle a,b\rangle$. Similarly, $X_3$ contains at least one attribute that does not belong to $X_2$. Thus, there exists an attribute $A_p$ such that either $\Pi (\tup t).A_p=b$ or $\Pi (\tup t).A_p=\langle b,c\rangle$. Finally, $X_2$ contains at least one attribute that does not belong to $X_1$. Thus, there exists an attribute $A_r$ such that either $\Pi (\tup t).A_r=c$ or $\Pi (\tup t).A_r=\langle a,c\rangle$.
Hence, $\Pi (\tup t) = \Pi (\tup t')$ implies that $\Pi (\tup t).A_l = \Pi (\tup t').A_l$, $\Pi (\tup t).A_p = \Pi (\tup t').A_p$ and $\Pi (\tup t).A_r = \Pi (\tup t').A_r$. We obtain that
$a=a'$, $b=b'$ and $c=c'$, which implies $\tup t=\tup t'$.

\partitle{$\mathbf{\Pi}$ preserves consistency}
Let $\tup t=(a,b,c)$ and $\tup t'=(a',b',c')$ be two distinct tuples.
We contend that  the set $\{t,t'\}$ is consistent w.r.t. $\str$ if and only if the set $\{\Pi(\tup t),\Pi(\tup t')\}$ is consistent w.r.t. $\depset$.
\paragraph*{The ``if'' direction}
Assume that $\{t,t'\}$ is consistent w.r.t $\str$. We prove that  $\{\Pi(\tup t),\Pi(\tup t')\}$ is consistent w.r.t $\depset$. Note that $\tup t$ and $\tup t'$ cannot agree on more than one attribute (otherwise, they will violate at least one FD in $\str$). Thus, $\Pi(\tup t)$ and $\Pi(\tup t')$ may only agree on attributes that appear in $X_1\cap X_2\cap X_3$ and in one of $(X_1\cap X_2)\setminus X_3$, $(X_1\cap X_3)\setminus X_2$ or $(X_2\cap X_3)\setminus X_1$. As mentioned above, $X_1$ contains at least one attribute that does not belong to $X_3$, thus no FD in $\depset$ contains only attributes from $X_1\cap X_2\cap X_3$ and $(X_1\cap X_3)\setminus X_2$ on its left-hand side (otherwise, $X_1$ will not be minimal). Similarly, no FD in $\depset$ contains only attributes from $X_1\cap X_2\cap X_3$ and $(X_2\cap X_3)\setminus X_1$ on its left-hand side and no FD in $\depset$ contains only attributes from $X_1\cap X_2\cap X_3$ and $(X_1\cap X_2)\setminus X_3$ on its left-hand side. Therefore $\Pi(\tup t)$ and $\Pi(\tup t')$ do not agree on the left-hand side of any FD in $\depset$, and $\{\Pi(\tup t),\Pi(\tup t')\}$ is consistent w.r.t. $\depset$.

\paragraph*{The ``only if'' direction}
Assume $\set{\tup t,\tup t'}$ is inconsistent w.r.t. $\str$. We prove that $\{\Pi(\tup t),\Pi(\tup t')\}$ is inconsistent w.r.t. $\depset$.
Since $\set{\tup t,\tup t'}$ is inconsistent w.r.t. $\str$, $\tup t$ and $\tup t'$ agree on two attributes, but do not agree on the third one. Thus, one of the following holds:
\begin{itemize}
\item $a=a'$, $b=b'$ and $c\neq c'$. In this case, $\Pi(\tup t)$ and $\Pi(\tup t')$ agree on all of the attributes that appear on the left-hand side of $X_1\rightarrow Y_1$. Since this FD is not trivial, it must contain on its right-hand side an attribute $A_k$ such that $A_k\not\in X_1$. That is, there is at least one attribute $A_k$ that appears on the right-hand side of $X_1\rightarrow Y_1$ such that one of the following holds: \e{(a)} $\Pi(\tup t).A_k=c, $\e{(b)} $\Pi(\tup t).A_k=\langle a,c\rangle$, \e{(c)} $\Pi(\tup t).A_k=\langle b,c\rangle$ or \e{(d)} $\Pi(\tup t).A_k=\langle a,b,c\rangle$. Hence, $\Pi(\tup t)$ and $\Pi(\tup t')$ do not satisfy the FD $X_1\rightarrow Y_1$ and $\{\Pi(\tup t),\Pi(\tup t')\}$ is inconsistent w.r.t. $\depset$.
\item $a=a'$, $b\neq b'$ and $c= c'$. This case is symmetric to the first one.  $\Pi(\tup t)$ and $\Pi(\tup t')$ agree on all of the attributes that appear on the left-hand side of $X_2\rightarrow Y_2$, but do not agree on at least one attribute that appears on the right-hand side of the FD.
\item $a\neq a'$, $b=b'$ and $c= c'$. This case is also symmetric to the first one. Here, $\Pi(\tup t)$ and $\Pi(\tup t')$ agree on the left-hand side, but not on the right-hand side of the FD $X_3\rightarrow Y_3$.
\end{itemize}
\end{proof}

\begin{lemma}\label{lemma:last-reduction}
Let $R$ be a schema and let $\depset$ be an FD set over $R$ that does not contain trivial FDs. Suppose that $\depset$ contains two distinct local minima $X_1\rightarrow Y_1$ and $X_2\rightarrow Y_2$, and the following hold.
\begin{itemize}
\item $(\closure_{\depset}(X_1)\setminus X_1)\cap X_2\neq\emptyset$ and $(\closure_{\depset}(X_2)\setminus X_2)\cap X_1\neq\emptyset$,
\item $(X_2\setminus X_1)\not\subseteq (\closure_{\depset}(X_1)\setminus X_1)$.
\end{itemize}
Then, there is a fact-wise reduction from $(R(A,B,C),\stk)$ to $(R,\depset)$.
\end{lemma}

\begin{proof}
We define a fact-wise reduction $\Pi:(R(A,B,C),\stk) \rightarrow(R,\depset)$, using $X_1\rightarrow Y_1$, $X_2\rightarrow Y_2$ and the constant $\odot \in \consts$.
Let $\tup t=(a,b,c)$ be a tuple over $R(A,B,C)$ and let $\set{A_1,\dots,A_n}$ be the set of attributes in $R$.
We define $\Pi $ as follows:
\[
\Pi (\tup t).A_k \eqdef
\begin{cases}
\odot & \mbox{$A_k\in X_1\cap X_2$} \\ 
 c& \mbox{$A_k\in X_1\setminus X_2$}\\ 
 b& \mbox{$A_k\in (X_2\setminus X_1) \cap (\closure_{\depset}(X_1)\setminus X_1)$}\\ 
 \langle a,b \rangle& \mbox{$A_k\in (X_2\setminus X_1) \setminus (\closure_{\depset}(X_1)\setminus X_1)$}\\ 
 \langle b,c\rangle& \mbox{$A_k\in (\closure_{\depset}(X_1)\setminus X_1)\setminus (X_2\setminus X_1)$}\\
 \langle a,b,c\rangle& \mbox{otherwise}
\end{cases}
\]
It is left to show that $\Pi$ is a fact-wise reduction.
To do so, we prove that $\Pi$ is well defined, injective and preserves consistency and inconsistency.

\partitle{$\mathbf{\Pi}$ is well defined}
This is straightforward from the definition.

\partitle{$\mathbf{\Pi}$ is injective}
Let $\tup t,\tup t'$ be two tuples, such that $\tup t=(a,b,c)$ and $\tup t'=(a',b',c')$.
Assume that $\Pi (\tup t) = \Pi (\tup t')$.  Let us denote
$\Pi (\tup t)=(x_1,\dots, x_n)$ and $\Pi (\tup t')=(x'_1,\dots, x'_n)$.
Since the FD $X_2\rightarrow Y_2$ is a local minimum, it holds that $X_1\not\subseteq X_2$. Thus, there is an attribute that appears in $X_1$, but does not appear in $X_2$. Moreover, it holds that $(X_2\setminus X_1)\not\subseteq (\closure_{\depset}(X_1)\setminus X_1)$, thus $X_2\setminus X_1$ contains at least one attribute that does not appear in $\closure_{\depset}(X_1)\setminus X_1$.
Therefore, there are $l$ and $p$ such that $\Pi (\tup t).A_l = c$, $\Pi (\tup t).A_p = \langle a,b\rangle$. 
Hence, $\Pi (\tup t) = \Pi (\tup t')$ implies that $\Pi (\tup t).A_l = \Pi (\tup t').A_l$ and $\Pi (\tup t).A_p = \Pi (\tup t').A_p$. We obtain that
$a=a'$, $b=b'$ and $c=c'$, which implies $\tup t=\tup t'$.

\partitle{$\mathbf{\Pi}$ preserves consistency}
Let $\tup t=(a,b,c)$ and $\tup t'=(a',b',c')$ be two distinct tuples.
We contend that the set $\{t,t'\}$ is consistent w.r.t. $\stk$ if and only if the set $\{\Pi(\tup t),\Pi(\tup t')\}$ is consistent w.r.t. $\depset$.
\paragraph*{The ``if'' direction}
Assume that $\{t,t'\}$ is consistent w.r.t $\stk$. We prove that  $\{\Pi(\tup t),\Pi(\tup t')\}$ is consistent w.r.t $\depset$. One of the following holds:
\begin{itemize}
\item $b\neq b'$ and $c\neq c'$. In this case, $\Pi(\tup t)$ and $\Pi(\tup t')$ only agree on the attributes $A_k$ such that $A_k\in X_1\cap X_2$. Since $X_1\rightarrow Y_1$ and $X_2\rightarrow Y_2$ are local minima, there is no FD in $\depset$ that contains on its left-hand side only attributes $A_k$ such that $A_k\in X_1\cap X_2$. Thus, $\Pi(\tup t)$ and $\Pi(\tup t')$ do not agree on the left-hand side of any FD in $\depset$ and $\{\Pi(\tup t),\Pi(\tup t')\}$ is consistent w.r.t. $\depset$.
\item $a\neq a'$, $b= b'$ and $c=c'$. Note that in this case $\Pi(\tup t)$ and $\Pi(\tup t')$ agree on all of the attributes that belong to $\closure_{\depset}(X_1)$, and only on these attributes. Any FD in $\depset$ that contains only attributes from $\closure_{\depset}(X_1)$ on its left-hand side, also contains only attributes from $\closure_{\depset}(X_1)$ on its right-hand side (by definition of closure). Thus, $\Pi(\tup t)$ and $\Pi(\tup t')$ satisfy all the FDs in $\depset$.
\item $a\neq a'$, $b= b'$ and $c\neq c'$. In this case, $\Pi(\tup t)$ and $\Pi(\tup t')$ only agree on the attributes $A_k$ such that $A_k\in X_1\cap X_2$ or $A_k\in (X_2\setminus X_1) \cap (\closure_{\depset}(X_1)\setminus X_1)$. Since the FD $X_2\rightarrow Y_2$ is a local minimum, and since $X_2$ contains an attributes that does not belong to $\closure_{\depset}(X_1)$, no FD in $\depset$ contains on its left-hand side only attributes $A_k$ such that $A_k\in X_1\cap X_2$ or $A_k\in (X_2\setminus X_1) \cap (\closure_{\depset}(X_1)\setminus X_1)$. Thus, $\Pi(\tup t)$ and $\Pi(\tup t')$ do not agree on the left-hand side of any FD in $\depset$ and $\{\Pi(\tup t),\Pi(\tup t')\}$ is consistent w.r.t. $\depset$.
\end{itemize}
This concludes our proof of the ``if'' direction.

\paragraph*{The ``only if'' direction}
Assume $\set{\tup t,\tup t'}$ is inconsistent w.r.t. $\stk$. We prove that $\{\Pi(\tup t),\Pi(\tup t')\}$ is inconsistent w.r.t. $\depset$.
Since $\set{\tup t,\tup t'}$ is inconsistent w.r.t. $\stk$, one of the following holds:
\begin{itemize}
\item $a=a$, $b= b'$ and $c\neq c'$. In this case, $\Pi(\tup t)$ and $\Pi(\tup t')$ agree on all of the attributes that appear in $X_2$. Since $X_2\rightarrow Y_2$ is not trivial, there is an attribute $A_k$ in $Y_2$ that does not belong to $X_2$. That is, one of the following holds: \e{(a)} $\Pi(\tup t).A_k=c$, \e{(b)} $\Pi(\tup t).A_k=\langle b,c \rangle$ or \e{(c)} $\Pi(\tup t).A_k=\langle a,b,c \rangle$. Thus, $\{\Pi(\tup t),\Pi(\tup t')\}$ violates the FD $X_2\rightarrow Y_2$ and it is inconsistent w.r.t. $\depset$.
\item $b\neq b'$ and $c= c'$. In this case, $\Pi(\tup t)$ and $\Pi(\tup t')$ agree on all of the attributes that appear in $X_1$. Since $X_1\rightarrow Y_1$ is not trivial, there is an attribute $A_k$ in $Y_1$ that does not belong to $X_1$. That is, one of the following holds: \e{(a)} $\Pi(\tup t).A_k=b$, \e{(b)} $\Pi(\tup t).A_k=\langle a,b \rangle$, \e{(c)} $\Pi(\tup t).A_k=\langle b,c \rangle$ or \e{(d)} $\Pi(\tup t).A_k=\langle a,b,c \rangle$. Thus, $\{\Pi(\tup t),\Pi(\tup t')\}$ violates the FD $X_1\rightarrow Y_1$ and it is again inconsistent w.r.t. $\depset$.
\end{itemize}
This concludes our proof of the ``only if'' direction.
\end{proof}

Next, we show that whenever
$\algname{\OSR}$ simplifies an FD set $\depset$ into an FD set
$\depset'$, there is a fact-wise reduction from $(R,\depset')$ to
$(R,\depset)$, where $R$ is the underlying relation schema.
Note that for each one of the simplifications, we just remove a set of attributes from the schema and the FDs in $\depset$. Thus, we will show that there is a fact-wise reduction from $(R,\depset-X)$ to $(R,\depset)$, where $X$ is a set of attributes. 

\begin{lemma}\label{lemma:fact-wise-simplifications}
Let $R(A_1,\dots,A_m)$ be a relation schema and let $\depset$ be an FD set over $R$. Let $X\subseteq \set{A_1,\dots,A_m}$ be a set of attributes. Then, there is a fact-wise reduction from $(R,\depset-X)$ to $(R,\depset)$.
\end{lemma}

\begin{proof}
We define a fact-wise reduction $\Pi:(R,\depset-X) \rightarrow (R,\depset)$, using the constant $\odot \in \consts$.
Let $\tup t$ be a tuple over $R$.
We define $\Pi $ as follows:
\[
\Pi (\tup t).A_k \eqdef
\begin{cases}
\odot & \mbox{$A_k\in X$} \\ 
t.A_k& \mbox{otherwise}
\end{cases}
\]
It is left to show that $\Pi$ is a fact-wise reduction.
To do so, we prove that $\Pi$ is well defined, injective and preserves consistency and inconsistency.

\partitle{$\mathbf{\Pi}$ is well defined}
This is straightforward from the definition.

\partitle{$\mathbf{\Pi}$ is injective}
Let $\tup t,\tup t'$ be two distinct tuples over $R$. Since $t\neq t'$, there exists an attribute $A_j$ in $\set{A_1,\dots,A_k}\setminus X$, such that $\tup t.A_j\neq \tup t'.A_j$. Thus, it also holds that $\Pi(\tup t).A_j\neq \Pi(\tup t').A_j$ and $\Pi(\tup t)\neq \Pi(\tup t')$.

\partitle{$\mathbf{\Pi}$ preserves consistency}
Let $\tup t,\tup t'$ be two distinct tuples over $R$.
We contend that the set $\{\tup t,\tup t'\}$ is consistent w.r.t. $\depset-X$ if and only if the set $\{\Pi(\tup t),\Pi(\tup t')\}$ is consistent w.r.t. $\depset$.
\paragraph*{The ``if'' direction}
Assume that $\{t,t'\}$ is consistent w.r.t $\depset-X$. We prove that  $\{\Pi(\tup t),\Pi(\tup t')\}$ is consistent w.r.t $\depset$.  Let us assume, by way of contradiction, that $\{\Pi(\tup t),\Pi(\tup t')\}$ is inconsistent w.r.t $\depset$. That is, there exists an FD $Z\rightarrow W$ in $\depset$, such that $\Pi(\tup t)$ and $\Pi(\tup t')$ agree on all the attributes in $Z$, but do not agree on at least one attribute $B$ in $W$. Clearly, it holds that $B\not\in X$ (since $\Pi(\tup t).A_i=\Pi(\tup t').A_i=\odot$ for each attribute $A_i\in X$). Note that the FD $(Z\setminus X)\rightarrow (W\setminus X)$ belongs to $\depset-X$. Since $\Pi(\tup t).A_k=t.A_k$ and $\Pi(\tup t').A_k=t'.A_k$ for each attribute $A_k\not\in X$, the tuples $\tup t$ and $\tup t'$ also agree on all the attributes on the left-hand side of the FD $(Z\setminus X)\rightarrow (W\setminus X)$, but do not agree on the attribute $B$ that belongs to $W\setminus X$. Thus, $\tup t$ and $\tup t'$ violate an FD in $\depset$, which is a contradiction to the fact that the set $\{\tup t,\tup t'\}$ is consistent w.r.t. $\depset-X$.

\paragraph*{The ``only if'' direction}
Assume $\set{\tup t,\tup t'}$ is inconsistent w.r.t. $\depset-X$. We prove that $\{\Pi(\tup t),\Pi(\tup t')\}$ is inconsistent w.r.t. $\depset$.
Since $\set{\tup t,\tup t'}$ is inconsistent w.r.t. $\depset-X$, there exists an FD $Z\rightarrow W$ in $\depset-X$, such that $\tup t$ and $\tup t'$ agree on all the attributes on the left-hand side of the FD, but do not agree on at least one attribute $B$ on its right-hand side. The FD $(Z\cup X)\rightarrow (W\cup X')$ (for some $X'\subseteq X$) belongs to $\depset$, and since $\Pi(\tup t).A_k=t.A_k$ and $\Pi(\tup t').A_k=t'.A_k$ for each attribute $A_k\not\in X$, and $\Pi(\tup t).A_k=\Pi(\tup t').A_k=\odot$ for each attribute $A_k\in X$, it holds that $\Pi(\tup t)$ and $\Pi(\tup t')$ agree on all the attributes in $Z\cup X$, but do not agree on the attribute $B\in (W\cup X')$, thus $\{\Pi(\tup t),\Pi(\tup t')\}$ is inconsistent w.r.t. $\depset$.
\end{proof}

The following lemmas are straightforward based on Lemma~\ref{lemma:fact-wise-simplifications}.

\begin{lemma}\label{lemma:s1-hard}
Let $R$ be a relation schema and let $\depset$ be an FD set over $R$. If $\depset$ has a common lhs $A$, then there is a fact-wise reduction from $(R,\depset-A)$ to $(R,\depset)$.
\end{lemma}

\begin{lemma}\label{lemma:s2-hard}
Let $R$ be a relation schema and let $\depset$ be an FD set over $R$. If $\depset$ has a consensus FD, $\emptyset\rightarrow X$, then there is a fact-wise reduction from $(R,\depset-X)$ to $(R,\depset)$.
\end{lemma}

\begin{lemma}\label{lemma:s3-hard}
Let $R$ be a relation schema and let $\depset$ be an FD set over $R$. If $\depset$ has an lhs marriage, $(X_1,X_2)$, then there is a fact-wise reduction from $(R,\depset-X_1X_2)$ to $(R,\depset)$.
\end{lemma}

\begin{lemma}\label{lemma:ind-basis}
Let $R$ be a relation schema and let $\depset$ be an FD set over $R$. If no simplification can be applied to $\depset$ and $\algname{OSRSucceeds}(\depset)$ returns false, then computing an optimal S-repair for $\depset$ is APX-complete.
\end{lemma}

\begin{proof}
If no simplification can be applied to $\depset$ and $\algname{OSRSucceeds}(\depset)$ returns false, then $\depset$ is not trivial and also not empty. Note that in this case, $\depset$ cannot be a chain. Otherwise, $\depset$ contains a global minimum, which is either an FD of the form $\emptyset\rightarrow X$, in which case $\depset$ has a consensus FD, or an FD of the form $X\rightarrow Y$, where $X\neq \emptyset$, in which case it holds that $X\subseteq Z$ for each FD $Z\rightarrow W$ in $\depset$ and $\depset$ has a common lhs. Thus, $\depset$ contains at least two local minima $X_1\rightarrow Y_1$ and $X_2\rightarrow Y_2$ (that is, no FD $Z\rightarrow W$ in $\depset$ is such that $Z\subset X_1$ or $Z\subset X_2$). Note that we always remove trivial FDs from $\depset$ before applying a simplification, thus we can assume that $\depset$ does not contain trivial FDs. One of the following holds:
\begin{enumerate}
\item $(\closure_\depset(X_2)\setminus X_2)\cap X_1 = \emptyset$. We divide this case into three subcases:
\begin{itemize}
\item $(\closure_\depset(X_1)\setminus X_1)\cap \closure_\depset(X_2)= \emptyset$. In this case, Lemma~\ref{lemma:ac-bc-hard} and Lemma~\ref{lemma:disjoint-reduction} imply that computing an optimal S-repair is APX-hard.
\item $(\closure_\depset(X_1)\setminus X_1) \cap (\closure_\depset(X_2)\setminus X_2)\neq \emptyset$ and $(\closure_\depset(X_1)\setminus X_1)\cap X_2 = \emptyset$. In this case, Lemma~\ref{lemma:ab-bc-hard} and Lemma~\ref{lemma:disjoint-left-reduction} imply that computing an optimal S-repair is APX-hard.
\item $(\closure_\depset(X_1)\setminus X_1)\cap X_2 \neq \emptyset$. In this case, Lemma~\ref{lemma:ab-bc-hard} and Lemma~\ref{lemma:disjoint-left-reduction} imply that computing an optimal S-repair is APX-hard.
\end{itemize}
\item $(\closure_\depset(X_2)\setminus X_2)\cap X_1 \neq \emptyset$. We divide this case into three subcases:
\begin{itemize}
\item $(\closure_\depset(X_1)\setminus X_1)\cap X_2 \neq \emptyset$ and it holds that $(X_1\setminus X_2)\subseteq (\closure_\depset(X_2)\setminus X_2)$ and $(X_2\setminus X_1)\subseteq (\closure_\depset(X_1)\setminus X_1)$. In this case, $\depset$ contains at least one more local minimum. Otherwise, for every FD $Z\rightarrow W$ in $\depset$ it holds that either $X_1\subseteq Z$ or $X_2\subseteq Z$. If $X_1\cap X_2\neq\emptyset$, then $\depset$ has a common lhs, which is an attribute from $X_1\cap X_2$. If $X_1\cap X_2=\emptyset$, then $\depset$ does not have a common lhs, but has an lhs marriage. In both cases, we get a contradiction to the fact that no simplifications can be applied to the $\depset$. Thus, Lemma~\ref{lemma:abc-acb-bca-hard} and Lemma~\ref{lemma:triangle-reduction} imply that computing an optimal S-repair is APX-hard.
\item $(\closure_\depset(X_1)\setminus X_1)\cap X_2 \neq \emptyset$ and it holds that $(X_2\setminus X_1)\not\subseteq (\closure_\depset(X_1)\setminus X_1)$. In this case, Lemma~\ref{lemma:abc-cb-hard} and Lemma~\ref{lemma:last-reduction} imply that computing an optimal S-repair is APX-hard.
\end{itemize}
\end{enumerate}
Proposition~\ref{prop:subset-approx} implies that computing an optimal S-repair is always in APX, thus the problem is actually APX-complete in each one of these cases. This concludes our proof of the lemma.
\end{proof}

Finally, we prove the following.

\begin{lemma}
Let $R$ be a relation schema and let $\depset$ be an FD set over $R$. If $\algname{OSRSucceeds}(\depset)$ returns false, then computing an optimal S-repair for $\depset$ is APX-complete.
\end{lemma}
\begin{proof}
We will prove the lemma by induction on $n$, the number of simplifications that will be applied to $\depset$ by $\algname{OSRSucceeds}$. The basis of the induction is $n=0$. In this case, Lemma~\ref{lemma:ind-basis} implies that computing an optimal S-repair is indeed APX-complete. For the inductive step, we need to prove that if the claim is true for all $n=1,\dots,k-1$, it is also true for $n=k$. In this case, $\algname{OSRSucceeds}(\depset)$ will start by applying a simplification to the problem. One of the following holds:

\begin{itemize}
\item $\depset$ has a common lhs $A$. Note that $\algname{OSRSucceeds}(\depset)$ will return false only if $\algname{OSRSucceeds}(\depset-A)$ returns false. From the inductive step we know that if $\algname{OSRSucceeds}(\depset-A)$ returns false, then computing an optimal S-repair for $\depset-A$ is APX-complete. Thus, Lemma~\ref{lemma:s1-hard} implies that computing an optimal S-repair for $\depset$ is APX-hard.
\item $\depset$ has a consensus FD $\emptyset\rightarrow X$. The algorithm $\algname{OSRSucceeds}(\depset)$ will return false only if $\algname{OSRSucceeds}(\depset-X)$ returns false. From the inductive step we know that if $\algname{OSRSucceeds}(\depset-X)$ returns false, then computing an optimal S-repair for $\depset-X$ is APX-complete. Thus, Lemma~\ref{lemma:s2-hard} implies that computing an optimal S-repair for $\depset$ is APX-hard.
\item $\depset$ has an lhs marriage $(X_1,X_2)$. Again, $\algname{OSRSucceeds}(\depset)$ will return false only if $\algname{OSRSucceeds}(\depset-X_1X_2)$ returns false. From the inductive step we know that if $\algname{OSRSucceeds}(\depset-X_1X_2)$ returns false, then computing an optimal S-repair for $\depset-X_1X_2$ is APX-complete. Thus, Lemma~\ref{lemma:s1-hard} implies that computing an optimal S-repair for $\depset$ is APX-hard.
\end{itemize}
Proposition~\ref{prop:subset-approx} implies that computing an optimal S-repair is always in APX, thus the problem is APX-complete in each one of these cases. This concludes our proof of the lemma.
\end{proof}

\def\refsecupdate{\ref{sec:update-repairs}}
\section{Details from Section~\refsecupdate}\label{app:proofs-value-repair}

\subsection{Proof of Theorem~\ref{thm:disjoint}}\label{sec:thm:disjoint:proof}
In this section we prove Theorem~\ref{thm:disjoint}.

\begin{reptheorem}{\ref{thm:disjoint}}
\thmdisjoint
\end{reptheorem}

First we show the following proposition:
\begin{proposition}\label{prop:optimal-disjoint}
For a schema $R$ and a table $T$ over $R$, if $U^*, U_1^*, U_2^*$ denote optimal U-repairs for $\depset, \depset_1, \depset_2$ respectively, where $\depset = \depset_1 \cup \depset_2$ and $\depset_1, \depset_2$ are attribute disjoint, then
$$\distu(U^*, T) = \distu(U_1^*, T) + \distu(U_2^*, T)$$ 
\end{proposition}
\begin{proof}
%Consider an $\alpha$-optimal U-repair $U$ under $\depset$ for a table $T$ such that $\distu(U, T) \leq \alpha \times distu(U^*, T)$. 
Let us denote the Hamming distance of a tuple id $i \in \ids(T)$ in $U^*$ and $T$ with respect to a subset of attributes $P \subseteq \attr(\depset)$ as $H_P(T[i], U^*[i])$ (\ie, the the number of attributes in $P$ where $T[i]$ and $U^*[i]$ differ). Since $U^*$ is an optimal U-repair for $\depset$, any attribute $\notin \attr(\depset)$ is not updated in $U^*$ (otherwise we can change $U^*$ in poly-time such that all attributes $\notin \attr(\depset)$ which will reduce the distance, contradicting that $U^*$ is optimal). Clearly, for any $i \in \ids(T)$, $H(T[i], U^*[i]) = H_{\attr(\depset)}(T[i], U^*[i]) = H_{\attr(\depset_1)}(T[i], U^*[i]) + H_{\attr(\depset_2)}(T[i], U^*[i])$.
Hence from $U^*$ we can create two repairs $U_1^*, U_2^*$, one for $\depset_1$ and one for $\depset_2$, by taking the attribute values in $\attr(\depset_1)$ (resp. $\attr(\depset_2)$ ) from $U$, and keeping the remaining unchanged as in $T$, leading to $\distu(U^*, T) = \distu(U_1^*, T) + \distu(U_2^*, T)$.  Note that $U_1^*$ is an optimal U-repair for $\depset_1$, otherwise combining with $U_2^*$ we get a U-repair $U_0^*$ for $\depset$ of smaller distance than $U^*$ contradicting the optimality of $U^*$. Similarly, $U_2^*$ is an optimal U-repair of $\depset_2$.
\end{proof}

\begin{proof}[of Theorem~\ref{thm:disjoint}]
We now prove the theorem.
%\emph{
%  Suppose that $\depset=\depset_1\cup\depset_2$ where $\depset_1$ and
%  $\depset_2$ are attribute disjoint. The following are equivalent for
%  all $\alpha\geq 1$.
%  \begin{enumerate}
%  \item An $\alpha$-optimal U-repair can be computed in polynomial
%    time under $\depset$.
%  \item An $\alpha$-optimal U-repair can be computed in polynomial
%    time under \e{each of} $\depset_1$ and $\depset_2$.
%\end{enumerate}
%}

\paragraph*{$(1)\rightarrow (2)$}
Suppose an $\alpha$-optimal U-repair can be computed in polynomial time under $\depset$ for any schema $R$ and any table $T$ on $R$ by an algorithm $P_{\depset}(T)$. We create another table $T_0$ from $T$ by keeping the values of attributes of $\attr(\depset_1)$ the same for all tuples, and changing the values of all other attributes of all tuples to 0.  Let $U_0^*, U_{01}^*, U_{02}^*$ denote optimal solutions for $\depset, \depset_1, \depset_2$ for $T_0$.
  From Proposition~\ref{prop:optimal-disjoint}, for $T'$, 
\begin{equation}
\distu(U_0^*, T_0) = \distu(U_{01}^*, T_0) + \distu(U_{02}^*, T_0)
\end{equation}
 Since no FDs are violated in $\depset_2$ in $T_0$, $T_0 = U_{02}^*$, and $\distu(U_{02}^*, T_0) = 0$. Therefore,
 \begin{equation}
 \distu(U_0^*, T_0) = \distu(U_{01}^*, T_0)
\end{equation}
Then we run the algorithm $P_{\depset}(T_0)$ on $T_0$ to obtain an $\alpha$-optimal U-repair $U_0$ for $T_0$ in polynomial time with respect to $\depset$,
\ie, 
 \begin{equation}
 \distu(U_0, T_0) \leq \alpha \distu(U_{0}^*, T_0)
\end{equation}
We obtain a U-repair $U_{01}$ for $\depset_1$ and $T_0$ from $U_0$ by copying the attribute values in $\attr(\depset_1)$ from $U_0$ and keeping the other attributes unchanged as in $T_0$.. 
Hence
  \begin{equation}
 \distu(U_{0_1}, T_0) \leq  \distu(U_{0}, T_0) \leq \alpha \distu(U_{0}^*, T_0) = \alpha \distu(U_{01}^*, T_0)
\end{equation}
Since the attributes $\notin \attr(\depset_1)$ are unchanged in $U_{01}$ and $U_{01}^*$, and the attribute values of $\attr(\depset_1)$ are the same in $T_0$ and $T$, $U_{01}$ is also a U-repair of $T$ with $ \distu(U_{0_1}, T_0) =  \distu(U_{0_1}, T)$, and    $U_{01}^*$ is also an optimal U-repair of $T$ with $ \distu(U^*_{0_1}, T_0) =  \distu(U^*_{0_1}, T)$, and hence,
  \begin{equation}
 \distu(U_{01}, T) \leq \alpha \distu(U_{01}^*, T)
\end{equation}
This gives an $\alpha$-optimal U-repair of $T$ for $\depset_1$ in polynomial time (similar argument holds for $\depset_2$). 

\paragraph*{$(2)\rightarrow (1)$}
This direction is simpler, and we argue that an $\alpha$-optimal U-repair of $\depset$ can be obtained by composing $\alpha$-optimal U-repairs $U_1, U_2$ of $\depset_1, \depset_2$  for a table $T$ over schema $R$. Suppose $U^*, U_1^*, U_2^*$ are optimal U-repairs of $\depset, \depset_1, \depset_2$ respectively. Since $U_1, U_2$ are $\alpha$-optimal U-repairs,
\begin{equation}\label{equn:2-1-disjoint-1}
\distu(U_1, T) \leq \alpha \distu(U_1^*, T) \textrm{ and } \distu(U_2, T) \leq \alpha \distu(U_2^*, T)
\end{equation}
We construct a U-repair $U$ of $\depset$ by choosing attribute values in $\attr(\depset_1)$ from $U_1$ values in $\attr(\depset_2)$ from $U_2$, and choosing all other attribute values from $T$.
Hence 
\begin{eqnarray*}
\distu(U, T) & = & \distu(U_1, T) + \distu(U_2, T)\\
 &  \leq & \alpha (\distu(U_1^*, T) + \distu(U_2^*, T))~~~~~ \textrm{ from (\ref{equn:2-1-disjoint-1})}\\
 & = & \alpha \distu(U^*, T)~~~~~ \textrm{ from Proposition~\ref{prop:optimal-disjoint})}
\end{eqnarray*}
Hence $U$ is an $\alpha$-optimal U-repair for $\depset$. 
\end{proof}

\subsection{Proof of Theorem~\ref{thm:u-consensus}}
In this section we prove Theorem~\ref{thm:u-consensus}.
\begin{reptheorem}{\ref{thm:u-consensus}}
\thmuconsensus
\end{reptheorem}

%For a  set of FDs $\depset$, th goal is to show that there is a strict reduction from
%  computing an optimal U-repair for $\depset$ to that for
%  $\depset-\closure_{\depset}(\emptyset)$, and vice versa. 
  To prove this theorem, first we show the following proposition:
 
 \begin{proposition}\label{prop:FD-fixed-single}
For a consensus FD 
  $\emptyset \rightarrow A$, an optimal U-repair can be computed in polynomial time. 
 \end{proposition} 
 \begin{proof}
Consider a table $T$ over a schema $R$ that violates the FD $\emptyset \rightarrow A$, \ie, at least two tuples in $T$ have two distinct values of $A$. For every distinct value $a$ of $A$ in $T$, let $M_a$ denote the set of tuples $\tup t$ such that $\tup t.A = a$, \ie, $M_a = \sigma_{A=a}T$. We obtain a repair of $T$ as follows: for every distinct value $a$ of $A$, compute the total weight of the tuples in $M_a$, \ie, compute $W_a = \sum_{\tup t \in M_a} w_{\tup t}$ (slightly abusing the notation for weights). Choose the value of $a$, say $a_0$, having the maximum value of $W_a$ among all such $a$'s. Keep the tuples in $M_{a_0}$ unchanged, and update every other tuple $\tup t \in M_a$, $a \neq a_0$, such that $\tup t.A = a_0$. Clearly, this new table, say $U$, is a consistent update since now every tuple $\tup t$ in $T$ will have $\tup t.A = a_0$.
 \par
 To see that $U$ is an optimal U-repair, first note that a repair with a better distance cannot be obtained by setting the $A$ values to a fresh constant from the infinite domain, since instead choosing a value from the active domain saves the cost of the repair for at least one tuple in $T$. Now assume some other value $a_1 \neq a_0$ has been chosen for all tuples in $T$ in a repair $U_1$ such that $\distu(U_1,T) < \distu(U,T)$. Then $\distu(U_1,T) = \sum_{i\in\ids(T)}w_T(i)\cdot H(T[i],U_1[i]) = \sum_{a \neq a_1} \sum_{\tup t \in M_a} w_{\tup t}$ (since the Hamming distance is 1 for all $a \neq a_1$, and is 0 for $a_1$)
 $  =  \sum_{a \neq a_1} W_a = (\sum_{\tup t \in T} w_{\tup t} ) - W_{a_1}$, whereas $\distu(U,T) = (\sum_{\tup t \in T} w_{\tup t} ) - W_{a_0}$. Since $W_{a_0} \geq W_{a_1}$, $\distu(U_1,T) \geq \distu(U,T)$ contradicting the assumption that $\distu(U_1,T) < \distu(U,T)$. Hence $U$ is an optimal U-repair.
 \end{proof}
 
 From Theorem~\ref{thm:disjoint} and ~\ref{prop:FD-fixed-single}, the corollary below follows:
 %(any FD of the form $\phi \rightarrow Z$ for subset of attributes $Z$ can be decomposed into a set of disjoint FDs $\phi \rightarrow A_j$ with single attributes $A_j$):
 \begin{corollary}\label{cor:FD-fixed-multiple}
For the consensus FD $\emptyset \rightarrow \closure_{\depset}(\emptyset)$, where $\closure_{\depset}(\emptyset)$ is the set of all consensus attributes, an optimal U-repair can be computed in polynomial time. 
 \end{corollary}
 
 In addition, the following proposition is straightforward.
 \begin{proposition}\label{prop:closure-emptyset-equiv}
The two sets of FDs $\depset$ and $\depset' = \set{\emptyset \ra \closure_{\depset}(\emptyset)} \cup (\depset - \closure_{\depset}(\emptyset))$ are equivalent, and for any table $T$, and $\alpha$-optimal solution for $\depset$ is also an $\alpha$-optimal
 solution for $\depset'$, and vice versa.
  \end{proposition}
Using the above results, we prove Theorem~\ref{thm:u-consensus}.

% \emph{
%  Let $\depset$ be a set of FDs.  There is a strict reduction from
%  computing an optimal U-repair for $\depset$ to computing an optimal U-repair for
%  $\depset-\closure_{\depset}(\emptyset)$, and vice versa.} 

\begin{proof}
To show strict reductions, we show the following two directions:

\paragraph*{$\alpha$-optimal U-repair of $\depset$ to an $\alpha$-optimal U-repair of  $\depset - \closure_{\depset}(\emptyset)$} 
Using (i) Proposition~\ref{prop:closure-emptyset-equiv} and the fact that (ii) $\set{\emptyset \ra \closure_{\depset}(\emptyset)}$ and  $(\depset - \closure_{\depset}(\emptyset))$ are attribute disjoint, from Theorem~\ref{thm:disjoint}, an $\alpha$-optimal U-repair of $\depset - \closure_{\depset}(\emptyset)$ can be computed in polynomial time if an $\alpha$-optimal U-repair of $\depset$ can be computed in polynomial time. 

\paragraph*{$\alpha$-optimal U-repair of $\depset - \closure_{\depset}(\emptyset)$ to an $\alpha$-optimal U-repair of  $\depset$} 
 Note that an optimal U-repair of $\depset - \closure_{\depset}(\emptyset)$ (which is also an $\alpha$-optimal U-repair for any $\alpha \geq 1$) can be computed in polynomial time (by Corollary~\ref{cor:FD-fixed-multiple}). Therefore, using Proposition~\ref{prop:closure-emptyset-equiv} and the fact that  $\set{\emptyset \ra \closure_{\depset}(\emptyset)}$ and  $(\depset - \closure_{\depset}(\emptyset))$ are attribute disjoint, this direction also follows from Theorem~\ref{thm:disjoint}.
\end{proof}

\subsection{Proof of Proposition~\ref{prop:AB-BA}}

\begin{repproposition}{\ref{prop:AB-BA}}
\propABBA
\end{repproposition}

  For the FD set $\depset  = \set{A\ra B,B\ra A}$, the goal is to show that an optimal U-repair can be computed in polynomial time. Note that the FDs $\set{A\ra B,B\ra A}$ imply that in a consistent update, one value of $A$ cannot be associated with multiple values of $B$ and vice versa.
  
\begin{proof}
For $\depset = \set{A\ra B,B\ra A}$, although $\mc(\depset) = 2$, we argue below that still  $\distu(U^*, T) = \dists(S^*, T)$ for an optimal U-repair $U^*$ and an optimal S-repair $S^*$  for a table $T$ over $R$.
From Corollary~\ref{cor:U-S-opt-val-compare}, 
\begin{equation}\label{equn:AB-BA-1}
\dists(S^*, T) \leq \distu(U^*, T)
\end{equation}
Given $S^*$, we can create a consistent update $U$ by keeping the tuples in $S^*$ unchanged, and updating the $A$ or $B$ values of all deleted tuples in $S^*$ as follows.
Consider any tuple $\tup t \in T \setminus S^*$. There must exist a tuple $\tup s \in S^*$ with either $\tup t.A = \tup s.A$ or $\tup t.B = \tup s.B$, otherwise, $\tup t$ could be included in $S^*$ violating the optimality of $S^*$. Suppose there is a tuple $\tup s \in S^*$ with  $\tup t.A = \tup s.A$ without loss of generality. Then we change the value of $\tup t.B$ to $\tup s.B$ with the Hamming distance = 1. Hence we get a consistent update  $U$ with 
\begin{equation}\label{equn:AB-BA-2}
\dists(S^*, T) \geq \distu(U, T)
\end{equation}
Combining (\ref{equn:AB-BA-1}) and (\ref{equn:AB-BA-2}), $U$ is an optimal U-repair. 
\par 
 Since $\depset$ passes the
  test of $\algname{OSRSucceds}$ (by applying lhs marriage), from Theorem~\ref{thm:dichotomy} an optimal S-repair $S^*$ can be computed in polynomial time, from which 
  an optimal U-repair of $\set{A\ra B,B\ra A}$ can be computed in polynomial time as described above.
  \end{proof}

\subsection{Proof of Theorem~\ref{thm:marriage-u-hard}}
In this section we prove Theorem~\ref{thm:marriage-u-hard}.
\begin{reptheorem}{\ref{thm:marriage-u-hard}}
\thmmarriageuhard
\end{reptheorem}

%  For the relation schema $R(A,B,C)$ and the FD set $\marriageplus$, the goal is to show that 
%  computing an optimal U-repair is APX-complete, even on unweighted,
%  duplicate-free tables. 

  As discussed in Section~\ref{sec:incomparable-u-s}, we give a reduction from the vertex cover problem in bounded-degree graphs (known to be APX-complete) showing that the size of the minimum vertex cover is $k$ if and only if the distance of the optimal U-repair for the constructed instance is $2|E|+k$ (where $E$ is the set of edges in the original graph). This proof is inspired by the NP-hardness proof for  $\depset=\set{A\rightarrow B, B\rightarrow C}$ in~\cite{DBLP:conf/icdt/KolahiL09}, but needs some additional machinery to extend the reduction to $\marriageplus$.

\cut{
In this section we prove the following lemma.

\begin{lemma}\label{lemma:updates-ab-ba-bc}
  Let $(\signature, \depset)$ be an FD schema, such that $\signature$ consists of a single relation $R(A,B,C)$ and $\depset=\set{A\rightarrow B, B\rightarrow A, B\rightarrow C}$. Then, the problem $\maxvrepfd{\signature}{\depset}$ is NP-hard.
\end{lemma}

The idea of our proof of Lemma~\ref{lemma:updates-ab-ba-bc} is similar to the idea of the proof of NP-hardness for the set of FDs $\depset=\set{A\rightarrow B, B\rightarrow C}$ in~\cite{DBLP:conf/icdt/KolahiL09}. We construct a reduction from the minimum vertex cover problem and prove that the size of the minimum vertex cover is $k$ if and only if the distance of the optimal U-repair for the constructed instance is $2|E|+k$ (where $E$ is the set of edges in the original graph). However, in this case, we have to prove some nontrivial results before we can use a similar proof to theirs for the reduction.
}

\paragraph*{Construction}
The input to minimum vertex cover problem is a graph $G(V, E)$. The goal is to find a vertex cover $C$ for $G$ (a set of vertices that touches every edge in $G$), such that no other vertex cover $C'$ of $G$ contains less vertices than $C$ (that is, $|C'|<|C|$). This problem is known to be NP-hard even for bounded degree graphs. 
 Given such an input, 
 we will construct the input 
 $I$ for our problem as follows.
 Let $V=\set{v_1,\dots,v_n}$ be the set of vertices in $g$ and let $E=\set{e_1,\dots,e_m}$ be the set of edges in $g$. For each edge $e_r=(v_i,v_j)$ in $E$, the instance $I$ will contain the following tuples:
 \begin{itemize}
 \item
 $(v_i,v_j,\val{0})$,
 \item
$(v_j,v_i,\val{0})$.
 \end{itemize}
 In addition, for each vertex $v_i\in V$, the instance $I$ will contain the following tuple:
  \begin{itemize}
 \item
 $(v_i,v_i,\val{1})$,
 \end{itemize}
 
Let $C_{min}$ be a minimum vertex cover of $G$. We will now prove the following results:
 \begin{enumerate}
     %\item The cost of each U-repair of $I$ is at least $2|E|$.
     \item For each vertex cover $C$ of $G$ of size $k$, there is a U-repair of $I$ with a distance $2|E|+k$.
     \item No U-repair of $I$ has a distance $2|E|+k$ for some $k<|C_{min}|$.
 \end{enumerate}
  Then, we can conclude that the size of the minimum vertex cover of $G$ is $m$ if and only if the distance of the optimal U-repair of $I$ is $2|E|+m$.
 
% \begin{lemma}\label{lemma:updates-ab-ba-bc-min2E}
% The cost of each U-repair of $I$ is at least $2|E|$.
% \end{lemma}
% \begin{proof}
% Let $J$ be a U-repair of $I$, and let $\set{f_1,\dots f_l}$ be the set of facts of the form $R(v_i,v_j,\val{0})$ in $I$ that are not updated by $J$. Clearly, if this set is empty, then the cost of $J$ is at least $2|E|$, since there are exactly $2|E|$ facts of this form. If this set is not empty, then $J$ updates $2|E|-l$ facts of this form, and does not update $l$ facts of this form, for some $l\in\set{1,\dots,2|E|-1}$. Let $f_r$ be the fact $R(v_i,v_j,\val{0})$ for some $v_i,v_j\in V$. Note that there is a conflict between this fact and the fact $R(v_i,v_i,\val{1})$ (since they agree on the value of attribute $A$, but do not agree on the value of attribute $B$, thus they violate the FD $A\rightarrow B$). Therefore, if the fact $R(v_i,v_j,\val{0})$ is not updated, the fact $R(v_i,v_i,\val{1})$ must be updated. Hence, we update $2|E|-l$ facts of the form $R(v_i,v_j,\val{0})$ and at least $l$ facts of the form $R(v_i,v_i,\val{1})$, and the cost of $J$ is at least $2|E|-l+l=2|E|$.
% \end{proof}

\paragraph*{Proof of (1)}
We will now prove that for each vertex cover $C$ of $g$ of size $k$, there is a U-repair of $I$ with distance $2|E|+k$. Assume that $C$ is a vertex cover of $g$ that contain $k$ vertices. We can build a U-repair $J$ of $I$ as follows: for each edge $e_r=(v_i,v_j)$, if $v_i\in C$, then we will change the tuples $(v_i,v_j,\val{0})$ and $(v_j,v_i,\val{0})$ to $(v_i,v_i,\val{0})$. Otherwise, we will change these two tuples to $(v_j,v_j,\val{0})$ (note that in this case, since $v_i$ does not belong to the vertex cover, $v_j$ does). Moreover, for each vertex $v_i$ that belongs to $C$, we will change the tuple $(v_i,v_i,\val{1})$ to $R(v_i,v_i,\val{0})$. Note that if a vertex $v_i$ does not belong to $C$, the instance $J$ will not contain any tuple of the form $(v_i,v_i,\val{0})$, thus the tuple $(v_i,v_i,\val{1})$ does not agree on the value of $A$ or $B$ with any other tuple in $J$, and is not in conflict with any other tuple. On the other hand, if a node $v_i$ does belong to $C$, the instance $J$ will only contain tuples of the form $(v_i,v_i,\val{0})$ for $v_i$, since in this case we change the tuple $(v_i,v_i,\val{1})$ to $(v_i,v_i,\val{0})$. Thus, $J$ is indeed a U-repair. Finally, for each edge $(v_i,v_j)\in E$ we update one cell in each of the tuples $(v_i,v_j,\val{0})$ and $(v_j,v_i,\val{0})$ ($2|E|$ cell updates), and for each vertex $V_i\in C$ we update one cell in the tuple $(v_i,v_i,\val{1})$ ($k$ cell updates), thus the distance of $J$ is $2|E|+k$.

\paragraph*{Proof of (2)}
\newcounter{counter}

In order to prove this part of the lemma (that is, that no U-repair of $I$ is of distance $2|E|+k$ for some $k<|C_{min}|$), we first have to prove the following lemma, which is a nontrivial part of the proof.

\begin{lemma}\label{lemma:comp-proof}
Let $J$ be a U-repair of $I$ that updates $t$ tuples of the form $(v_i,v_j,\val{0})$ for some $t<2|E|$. Then, there is a U-repair $J'$ of $I$, that updates at least $t+1$ tuples of the form $(v_i,v_j,\val{0})$, such that the distance of $J'$ is lower or equal to the distance of $J$.
\end{lemma}
\begin{proof}
Since it holds that $t<2|E|$, there is at least one tuple $(v_1,v_2,\val{0})$, for some $v_1,v_2\in V$ that is not updated by $J$. We will now show that we can build another U-repair $J'$, that updates the tuple $(v_1,v_2,\val{0})$ (as well as every tuple that is updated by $J$), such that the distance of $J'$ is lower or equal to the distance of $J$. That is, $J'$ will update at least $t+1$ tuples of the form $(v_i,v_j,\val{0})$.

The tuple $\tup t_1=(v_1,v_2,\val{0})$ is in conflict with the tuples $\tup t_2=(v_1,v_1,\val{1})$ and $\tup t_3=(v_2,v_2,\val{1})$. Moreover, these two tuples are in conflict with the tuple $\tup t_4=(v_2,v_1,\val{0})$. We will now show that since the tuple $\tup t_1$ is not updated by $J$ and $J$ is consistent, $J$ must make at least four changes to the tuples in $\set{\tup t_2,\tup t_3,\tup t_4}$. The following holds: 
\begin{enumerate}
    \item To resolve the conflict between $\tup t_1$ and $\tup t_2$ (that agree only on the value of attribute $A$), we have to either change the value of attribute $A$ in $\tup t_2$ to another value, or change the values of both attributes $B$ and $C$ to $v_2$ and $\val{0}$, respectively. In the first case, since $\tup t_2$ is still in conflict with $\tup t_4$ (as they only agree on the value of attribute $B$), we also have to change the value of attribute $B$ in one of $\tup t_2$ or $\tup t_4$. In the second case, $\tup t_2$ and $\tup t_4$ are no longer in conflict, thus we do not have to make any more changes. In both cases, the distance of the update is $2$.
    \item To resolve the conflict between $\tup t_1$ and $\tup t_3$ (that agree only on the value of attribute $B$), we have to either change the value of attribute $B$ in $\tup t_3$ to another value, or change the values of both attributes $A$ and $C$ to $v_1$ and $\val{0}$, respectively. In the first case, since $\tup t_3$ is still in conflict with $\tup t_4$ (as they only agree on the value of attribute $A$), we also have to change the value of attribute $A$ in one of $\tup t_3$ or $\tup t_4$. In the second case, $\tup t_3$ and $\tup t_4$ are no longer in conflict, thus we do not have to make any more changes. In both cases, the cost of the update is $2$.
\end{enumerate} 
Note that the updates presented in $(1)$ and $(2)$ above are independent (as they change different attributes in different tuples), thus we have to combine them to resolve all of the conflicts among the tuples in $\set{\tup t_1,\tup t_2,\tup t_3,\tup t_4}$, and the total cost of updating the tuples $\tup t_2$, $\tup t_3$ and $\tup t_4$ in $J$ is at least $4$. Moreover, note that it is not necessary to change the value of attribute $A$ in both $\tup t_3$ and $\tup t_4$ or change the value of attribute $B$ in both $\tup t_2$ and $\tup t_4$, since changing just one of them resolves the conflict between the tuples. If $J$ updates both of these values, then the cost of updating the tuples $\tup t_2$, $\tup t_3$ and $\tup t_4$ is increased by one for each pair of tuples ($\tup t_2$ and $\tup t_4$, or $\tup t_3$ and $\tup t_4$) updated together. We will use this observation later in the proof.

Now, we will build an instance $J'$ that updates each one of the tuples updated by $J$, as well as the tuple $\tup t_1$. We will prove that the distance of $J'$ is not higher than the distance of $J$. We will start by adding these four tuples to $J'$:
\begin{enumerate}
    \item Instead of the tuple $\tup t_1$ ($(v_1,v_2,\val{0})$), we will insert the tuple $(v_1,v_1,\val{0})$.
    \item Instead of the tuple $\tup t_2$ ($(v_1,v_1,\val{1})$), we will insert the tuple $(v_1,v_1,\val{0})$.
    \item Instead of the tuple $\tup t_3$ ($(v_2,v_2,\val{1})$), we will insert the tuple $(v_2,v_2,\val{0})$.
    \item Instead of the tuple $\tup t_4$ ($(v_2,v_1,\val{0})$), we will insert the tuple $(v_2,v_2,\val{0})$.
    \setcounter{counter}{\value{enumi}}
\end{enumerate}
So far we only updated one cell in each of the four tuples discussed above. Thus, the current cost is $4$. As mentioned above, the cost of changing these tuples in $J$ is at least $4$, thus so far we did not exceed the cost of $J$. From now on, we will look at the rest of the tuples in $J$ and insert each one of them to $J'$ (in some cases, only after updating them to resolve some conflict). 
\begin{enumerate}
    \setcounter{enumi}{\value{counter}}
    \item For each tuple $\tup t$ that appears in $J$ that is not in conflict with one of $(v_1,v_1,\val{0})$ or $(v_2,v_2,\val{0})$, we insert $\tup t$ to $J'$.
    \setcounter{counter}{\value{enumi}}
\end{enumerate}
Clearly, $J'$ is consistent at this point, and since those tuples are updated by $J'$ in exactly the same way they were updated by $J$, we did not exceed the distance of $J$. Each one of the remaining tuples is in conflict with one of $(v_1,v_1,\val{0})$ or $(v_2,v_2,\val{0})$, thus we have to update it before inserting it to $J'$.
\begin{enumerate}
    \setcounter{enumi}{\value{counter}}
    \item Each tuple $\tup t$ such that $\tup t.A=v_1$ is in conflict with the tuple $(v_1,v_1,\val{0})$. Since the tuple $(v_1,v_2,\val{0})$ belongs to $J$ and $J$ is consistent, it holds that $\tup t=(v_1,v_2,\val{0})$. 
    There are a few possible cases:
        \begin{itemize}
            \item If the original tuple $\tup t$ (the tuple in the instance $I$) was of the form $(v_i,v_j,\val{0})$ for some $v_i\neq v_1$ and $v_j\neq v_2$, then $J$ updated two cells in this tuple. We will replace this tuple in $J'$ with the tuple $(v_1,v_1,\val{0})$. Note that the cost is the same (since we again update two cells).
            \item If the original tuple $\tup t$ was of the form $(v_1,v_j,\val{0})$ for some $v_j\neq v_2$, then $J$ updated one cell in this tuple. We will replace this tuple in $J'$ with the tuple $(v_1,v_1,\val{0})$. Again, the cost is the same.
            \item If the original tuple was of the form $(v_i,v_2,\val{0})$ for some $v_i\neq v_1$, we will replace it with the tuple $(v_2,v_2,\val{0})$ with no additional cost.
            \item If the original tuple was of the form $(v_i,v_i,\val{1})$ for some $v_i\not\in\set{v_1,v_2}$ (we already handled the case where $v_i$ is one of $v_1$ or $v_2$), then $J$ updated three cells in this tuple (the values in all of the attributes are different from the values in the tuple $(v_1,v_2,\val{0})$), and we will replace it with the tuple $(v_2,v_2,\val{0})$ with no additional cost.
        \end{itemize}
    \item Each tuple $\tup t$ such that $\tup t.B=v_2$ is in conflict with the tuple $(v_2,v_2,\val{0})$. Again, since the tuple $(v_1,v_2,\val{0})$ belongs to $J$ and $J$ is consistent, it holds that $\tup t=(v_1,v_2,\val{0})$. This case is symmetric to the previous one, thus we will not elaborate on all of the possible cases.
    \setcounter{counter}{\value{enumi}}
\end{enumerate}
So far, we inserted to $J'$ one tuple (either $(v_1,v_1,\val{0})$ or $(v_2,v_2,\val{0})$) for each of the four tuples in $\set{\tup t_1,\tup t_2,\tup t_3,\tup t_4}$, one tuple for each tuple in $J$ that is not in conflict with neither $(v_1,v_1,\val{0})$ nor $(v_2,v_2,\val{0})$, and one tuple (again either $(v_1,v_1,\val{0})$ or $(v_2,v_2,\val{0})$) for each tuple in $J$ that agrees with $(v_1,v_1,\val{0})$ only on the value of attribute $A$ or agrees with $(v_2,v_2,\val{0})$ only on the value of attribute $B$. It is left to insert each tuple that agrees with $(v_1,v_1,\val{0})$ on the value of attribute $B$ or agrees with $(v_2,v_2,\val{0})$) on the value of attribute $A$.
\begin{enumerate}
    \setcounter{enumi}{\value{counter}}
    \item Each tuple $\tup t$ such that $\tup t.A=v_2$ is in conflict with the tuple $(v_2,v_2,\val{0})$. Note that in this case it holds that $\tup t.B\neq v_2$, since $J$ is consistent and also contains the tuple $(v_1,v_2,\val{0})$. Assume that $\tup t$ is the tuple $(v_2,v_i,x)$ for some $v_i\in V$. Again, there are a few possible cases:
        \begin{itemize}
            \item  If the original tuple $\tup t$ (the tuple in the instance $I$) was of the form $(v_j,v_i',y)$ for some $v_j\neq v_2$ and $v_i'\neq v_i$, then $J$ updated at least two cells in this tuple in $J$. We will replace this tuple in $J'$ with the tuple $(v_2,v_2,\val{0})$, with no additional cost.
            \item If the original tuple $\tup t$ was of the form $(v_2,v_i',\val{0})$ for some $v_i'\neq v_i$, then $J$ updated at least one cell in this tuple (the value of attribute $B$). We will replace this tuple in $J'$ with the tuple $(v_2,v_2,\val{0})$, with no additional cost.
            \item If the original tuple $\tup t$ was of the form $(v_j,v_i,x)$ for some $v_j\neq v_2$, then $J$ updated at least one cell in this tuple (the value of attribute $A$). We will replace this tuple in $J'$ with the tuple $(y_i,v_i,\val{0})$, for some new value $y_i$ that depends on $i$, with no additional cost.
        \end{itemize}
    \item Each tuple $\tup t$ such that $\tup t.B=v_1$ is in conflict with the tuple $(v_1,v_1,\val{0})$. This case is symmetric to the previous one, thus we will not elaborate on all of the possible cases.
    \setcounter{counter}{\value{enumi}}
\end{enumerate}
At this point, $J'$ contains tuples of the form $(v_1,v_1,\val{0})$ and $(v_2,v_2,\val{0})$, tuples from $J$ that are not in conflict with these tuples, and tuples of the form $(y_i,v_i,\val{0})$. Note that we insert a tuple $(y_i,v_i,\val{0})$ to $J'$ only if there is a tuple $(v_2,v_i,x)$ in $J$ for some $v_i\neq v_2$. Since the tuple $(v_2,v_i,x)$ belongs to $J$, no other tuple $\tup t'$ such that $\tup t'.B=v_i$ also belongs to $J$, unless the tuples are identical, and all of these identical tuples are replaced by the same tuple. Thus, there is no tuple $\tup t''$ in $J$, that is not in conflict with one of $(v_1,v_1,\val{0})$ or $(v_2,v_2,\val{0})$, such that $\tup t''.B=v_i$, and the tuple $(y_i,v_i,\val{0})$ agrees on the value of attribute $B$ only with other tuples that look exactly the same. Moreover, the tuple $(y_i,v_i,\val{0})$ agrees on the value of attribute $A$ only with other tuples that look exactly the same, since we use a new value $y_i$ that does not appear in any other tuple in the instance. Thus, no tuple $(y_i,v_i,\val{0})$ is in conflict with the rest of the tuples in $J'$, and $J'$ is consistent. Moreover, since we only replaced tuples with no additional cost, so far we did not exceed the distance of $J$.

Finally, there are just two more cases that we have not covered yet. Assume that $\tup t$ is the tuple $(v_2,v_i,x)$ for some $v_i\neq v_2$, and assume that the original tuple $\tup t$ was the tuple $(v_2,v_i,\val{0})$ and $v_i\neq v_1$ (if $v_i=v_1$ then we have already covered this case). Then, it may be the case that the tuple was not updated by $J$ at all. Since the tuple $(v_2,v_i,x)$ belongs to $J$ and $J$ is consistent, no other tuple $\tup t'\in J$ is such that $\tup t'.A=v_2$ (unless they are identical). Thus, $J$ updated the value of attribute $A$ in both $\tup t_3$ (the tuple $(v_2,v_2,\val{1})$) and $\tup t_4$ (the tuple $(v_2,v_1,\val{0})$). As mentioned in the observation at the beginning of this proof, if $J$ updated both these values, then the cost of updating the tuples $\tup t_2$, $\tup t_3$ and $\tup t_4$ is at least $5$ (and not at least $4$ as we assumed so far). Thus, in this case, we can use the additional cost that we have, to change one more cell in the tuple $(v_2,v_i,x)$. 

Similarly, if $J$ contains a tuple $\tup t$ of the form $(v_i,v_1,x)$ for some $v_i\neq v_1$, and the original tuple $\tup t$ was the tuple $(v_i,v_1,\val{0})$ and $v_i\neq v_2$, then $J$ updated the value of attribute $B$ in in both $\tup t_2$ (the tuple $(v_1,v_1,\val{1})$) and $\tup t_4$ (the tuple $(v_2,v_1,\val{0})$). This again gives us an additional cost to update one more cell in the tuple $(v_i,v_1,x)$. Thus, we will do the following:
\begin{enumerate}
    \setcounter{enumi}{\value{counter}}
    \item If the tuple $(v_2,v_i,x)$ for some $v_i\not\in\set{v_1,v_2}$ belongs to $J$, and the original tuple was $(v_2,v_i,\val{0})$, then we replace this tuple in $J'$ with the tuple $(v_2,v_2,\val{0})$ with an additional cost of at most one.
    \item If the tuple $(v_j,v_1,x)$ for some $v_j\not\in\set{v_1,v_2}$ belongs to $J$, and the original tuple was $(v_j,v_1,\val{0})$, then we replace this tuple in $J'$ with the tuple $(v_1,v_1,\val{0})$ with an additional cost of at most one.
    \setcounter{counter}{\value{enumi}}
\end{enumerate}
At this point, we have inserted to $J'$ one tuple for each tuple in $J$, and as explained above, $J'$ is now consistent. Moreover, we never exceeded to the of $J$, thus we found a new U-repair $J'$ that updates each tuple updated by $J$, as well as the tuple $(v_1,v_2,\val{0})$ and that concludes our proof of the lemma.
\end{proof}

Lemma~\ref{lemma:comp-proof} implies that for each U-repair $J$ of $I$ that updates $t$ tuples of the form $(v_i,v_j,\val{0})$ for some $t<2|E|$, there is a U-repair $J'$ of $I$, that updates all of the tuples of the form $(v_i,v_j,\val{0})$, such that the distance of $J'$ is lower or equal to the distance of $J$. Thus, each U-repair of $I$ is of distance at least $2|E|$.

Finally, we can prove that no U-repair of $I$ is of distance $2|E|+k$ for some $k<|C_{min}|$. The proof of this part of the lemma is very similar to the proof of~\cite{DBLP:conf/icdt/KolahiL09}. Each U-repair of $I$ has a distance $2|E|+k$ for some $k$. Let us assume, by way of contradiction, that there is a U-repair of $I$ that updates $2|E|+k$ cells, for some $k<|C_{min}|$ (where $C_{min}$ is a minimum vertex cover of $g$). Let $J$ be the optimal U-repair of $I$ (clearly, in this case, $J$ updates $x=2|E|+k'$ cells, for some $k'<|C_{min}|$). Lemma~\ref{lemma:comp-proof} implies that we can assume that $J$ updates all of the tuples of the form $(v_i,v_j,\val{0})$.
Let $\set{f_1,\dots,f_l}$ be the set of tuples of the form $(v_i,v_i,\val{1})$ that are updated by $J$, and let $C'=\set{v_{i_1},\dots,v_{i_l}}$ be the set of the vertices that appear in those facts. Since $J$ updates all of the $2|E|$ tuples of the form $(v_i,v_j,\val{0})$, the distance of $J$ is greater or equal to $2|E|+|C'|$, thus we get that $2|E|+|C_{min}|>2|E|+k'>=2|E|+|C'|$, and $|C'|<|C_{min}|$. Therefore, $C'$ cannot be a cover of $g$ (since the size of the minimum vertex cover is $|C_{min}|$. Thus, there are at least $|C_{min}|-|C'|$ edges $(v_i,v_j)$ in $g$, such that neither $v_i\in C'$ nor $v_j\in C'$ (that is, $J$ doesn't update the tuple $(v_i,v_i,\val{1})$ and also doesn't update the tuple $(v_j,v_j,\val{1})$). For each one of these $|C_{min}|-|C'|$ edges, we have to change at least $4$ cells in the tuples $(v_i,v_j,\val{0})$ and $(v_j,v_i,\val{0})$ (the values of attributes $A$ and $B$), thus the distance of $J$ is greater or equal to $4(|C_{min}|-|C'|)+2(|E|-(|C_{min}|-|C'|))+|C'|=2|C_{min}|+2|E|-|C'|$. Finally, we get that $2|E|+|C_{min}|>=2|C_{min}|+2|E|-|C'|$, that is $|C'|>=|C_{min}|$, which is a contradiction to the tuple that $|C'|<|C_{min}|$.

The last step is to show that our reduction is a PTAS reduction. Since our reduction is from the problem of finding a minimum vertex cover for a graph of a bounded degree $B$, the size of the minimum vertex cover is at least $\frac{|E|}{B}$ (as each vertex covers at most $B$ edges). Now, let us assume, by way of contradiction, that finding an optimal U-repair under $\set{A\rightarrow B, B\rightarrow A, A\rightarrow C}$ is not APX-hard. That is, for every $\epsilon>0$, there exists a $(1+\epsilon)$-optimal solution to that problem. Let $S_U$ be the such a solution and let $OPT_U$ be an optimal solution. Then,
\begin{equation}
    \begin{aligned} 
        dist(S_U)-dist(OPT_U)\le \epsilon\cdot dist(OPT_U)
    \end{aligned}
\end{equation}

We will now show that in this case, $S_C$ (which is the vertex cover corresponding to the the U-repair $S_U$) is a $(1+2B)\epsilon$-optimal solution to the minimum vertex cover problem. Let $OPT_C$ be an optimal solution to the problem. We recall that $|OPT_C|\ge \frac{|E|}{B}$. Then,

\begin{equation}
    \begin{aligned} 
        |S_{C}|-|OPT_{C}|=(dist(S_{U})-2|E|)-(dist(OPT_U)-2|E|)=dist(S_U)-dist(OPT_U)\le \\
        \epsilon\cdot dist(OPT_U)=\epsilon\cdot (|OPT_C|+2|E|)\le \epsilon\cdot (|OPT_C|+2B|OPT_C|)=(1+2B)\epsilon\cdot |OPT_C|
    \end{aligned}
\end{equation}

It follows that for each $\epsilon>0$, we can get an $\epsilon$-optimal solution for the minimum vertex cover problem by finding a $(1+\frac{1}{1+2B})\epsilon$-optimal solution to the problem of computing an optimal U-repair, which is a contradiction to the fact that the minimum vertex cover problem is APX-hard.

\subsection{Proof of Theorem~\ref{theorem:u-KL-quad}}
In this section we prove Theorem~\ref{theorem:u-KL-quad}.
\begin{reptheorem}{\ref{theorem:u-KL-quad}}
\theoremuKLquad
\end{reptheorem}

%Given 
%\[\depset_k \eqdef \set{A_0\cdots A_k \ra B_0, B_0 \rightarrow C, B_1 \rightarrow A_0\, \, \cdots, \, B_k \rightarrow A_0}\]
%The goal is to show that 
%for the schema $R(A_0, \cdots, A_k,
%  B_0,\cdots,B_k, C)$ and FD set $\depset_k$, for any fixed $k \geq 1$, computing an optimal
%  U-repair is APX-complete.
We start by proving the first part of the theorem, that is, we prove the following.

\begin{lemma}
Let $k\geq 1$ be fixed. For the schema $R(A_0, \dots, A_k,
  B_0,\dots,B_k,C)$ and FD set $\depset_k=\set{A_0\cdots A_k \ra B_0, B_0 \rightarrow C, B_1 \rightarrow A_0,\, \ldots, \, B_k \rightarrow A_0}$, computing an optimal
  U-repair is APX-complete.
\end{lemma}

\begin{proof}
We reduce it from the problem of finding an optimal U-repair for $\set{A \rightarrow B, B \rightarrow C}$ on a schema $S(A, B, C)$ that has been shown to be NP-hard by Kolahi and Lakshmanan even if all tuples have the same weight and there are no duplicate tuples \cite{DBLP:conf/icdt/KolahiL09} by a reduction from vertex cover. Since there exists a constant factor approximation for optimal U-repair and degree-bounded vertex cover is known to be APX-hard, this reduction suffices to prove that $\depset_k$ is APX-complete. 
\par  
\paragraph*{Construction} 
Given a table $T_S$ on $S(A, B, C)$, we construct a table $T_R$ on $R(A_0, \cdots, A_k,
  B_0,\dots,B_k, C)$, where for  every tuple $\tup s = (a, b, c)$ in $T_S$, we create a tuple $\tup r = (0, a, 0, 0, \cdots 0, b, 0, 0, \cdots, 0, c)$ in $T_R$, i.e., $\tup r.A_1 = \tup s.A$, $\tup  r.B_0 = \tup s.B, \tup r.C = \tup s.C$, and $\tup t$ has value 0 in the remaining columns $A_0, A_2, \cdots, A_k$ and $B_1, \cdots, B_k$. We claim that $T_S$ has a consistent update of distance $\leq M$  if and only if $T_R$ has a consistent update of distance $\leq M$.
\par
\paragraph*{The ``if'' direction} 
Suppose $S$ has a consistent update $U_S$ of distance $M$. We can obtain a consistent update $U_R$ of $R$ of the same distance by updating the $A_1, B_0, C$ values of $T_R$ as in $T_S$, and leave $A_0, A_2, \cdots, A_k$ and $B_1, \cdots, B_k$ unchanged in $T_R$. If $B_0 \rightarrow C$ is violated in $U_R$ for two tuples $\tup r_1, \tup r_2$, $B \ra C$ will also be violated for the corresponding tuples $s_1, s_2$ in $U_S$. If $A_0 \cdots A_k \rightarrow B_0$ is violated for two tuples $\tup r_1, \tup r_2$ in $U_R$, they must have different values of  $B_0$ but the same value of $A_0\cdots A_k$, \ie, the same value of $A_1 = A$  in $U_S$ (since $A_0, A_2, \cdots, A_k$ are all 0 in $U_R$), violating the FD $A \rightarrow B_0$ in $U_S$ and contradicting that it is a consistent update of $T_S$. The FDs $B_1 \rightarrow A_0\, \, \ldots, \, B_k \rightarrow A_0$ are satisfied in $U_R$ since all of $A_0$ and $B_1, \cdots, B_k$ have values 0.
\par
\paragraph*{The ``only if'' direction} 
Suppose we have a consistent update $U_R$ of $T_R$ of distance $M$. We will transform $U_R$ to another consistent update $U_R'$ such that only $A_1, B_0, C$ attributes are updated, and all of $B_1, \cdots, B_k$ and $A_0, A_2, \cdots, A_k$ are unchanged from $T_R$ (all tuples have values 0 in all these attributes in $U_R'$).
\par
To achieve this, consider the subset of tuple identifiers $M_R \subseteq \ids(T_R)$ such that for any tuple $\tup t = T_R[i]$, $i \in M_R$, at least one attribute from $B_1, \cdots, B_k$ or $A_0, A_2, \cdots, A_k$ has been updated in $U_R$ (not equal to 0, and therefore $H(T_R[i], U_R[i]) \geq 1$).  For all tuples with identifiers in $M_R$, we update $U_R$ to $U_R'$ by (i) keeping all attributes in      $B_1, \cdots, B_k$ and $A_0, A_2, \cdots, A_k$ as 0 (unchanged from $T_R$), and (ii) assigning a fresh constant from the infinite domain $\dom$ to attribute $A_1$ in each tuple (changed from $T_R$); therefore for all $i \in M_R$, $H(T_R[i], U_R'[i]) = 1$. Since tuples with identifiers $\notin M_R$ remain the same in $U_R$ and $U_R'$, and the tuples with identifiers $i \in M_R$ have
$H(T_R[i], U_R[i]) \geq 1$ and $H(T_R[i], U_R'[i]) = 1$, therefore, the $\distu(U_R', T_R) \leq \distu(U_R, T_R)$, \ie, this transformation does not increase the distance. 
\par
Next we argue that $U_R'$ is a consistent update for $T_R$ where all tuples have values 0 in all all attributes in $B_1, \cdots, B_k$ and $A_0, A_2, \cdots, A_k$. (a) Since $U_R$ was a consistent update, it satisfied the FD $B_0 \ra C$, and since neither of $B_0$ and $C$ is changed in $U_R'$ from $U_R$, this FD is still satisfied in $U_R'$. (b) For the FDs  $B_1 \rightarrow A_0\, \, \cdots, \, B_k \rightarrow A_0$, since all tuples in $U_R'$ have values 0 in all these attributes, so all of these FDs are satisfied in $U_R'$ as well. (c) For the final FD $A_0\cdots A_k \ra B_0$, for tuples with identifiers $\notin M_R$, this FD is satisfied since they have the same value of the attributes $A_0\cdots A_k, B_0$ in $U_R$ and $U_R'$. For tuples with identifiers $\notin M_R$, they received a fresh constant in $U_R'$ for attribute $A_1$, so they do not interfere with each other and with the tuples with identifiers $\notin M_R$, and therefore the FD  $A_0\cdots A_k \ra B_0$ is satisfied in $U_R'$. Hence $U_R'$ is a consistent update for $T_R$.

Hence, we have a new consistent update $U_R'$ of $R$ with distance $\leq M$ where only $A_1, B_0, C$ are updated. Suppose $U_S = \rho_{(A, B, C)} [\pi_{A_1 B_0 C} U_R]$ (project $U_R'$ to $A_1B_0C$ and rename $A_1$ to $A$ and $B_0$ to $B$). It  can be seen that $U_S$ satisfies the FD $B \rightarrow C$, since $B_0 \ra C$ is maintained in $U_R'$. It also satisfies the FD $A \rightarrow B$, since otherwise there are two tuples $\tup s_1, \tup s_2 \in U_S$ with the same value of $A$ and different values of $B$. This implies that the corresponding tuples $\tup r_1, \tup r_2$ in $U_R$ have the same value of $A_1$ and different values of $B_0$. Since both $\tup r_1, \tup r_2$ have the same values of all attributes in $A_0, A_2, \cdots A_k$ (all 0), together with $A_1$, they have the same values of $A_0 \cdots A_k$ but different values of $B_0$, violating the FD  $A_0 \cdots A_k \rightarrow B$, and contradicting the tuple that $U_R'$ is a consistent update of $R$. Hence we get a consistent update$U_S$ of $S$ of at most distance $M$.
\end{proof}

Next, we prove the second part of the theorem. That is, we prove the following.

\begin{lemma}
Let $k\geq 1$ be fixed. For the schema $R(A_0, \dots, A_{k+1},
  B_0,\dots,B_k)$ and FD set $\depset'_k =\set{A_0A_1\ra B_0,\,A_1A_2\ra B_1,\,\dots,\,A_{k}A_{k+1}\ra B_k}$, computing an optimal
  U-repair is APX-complete.
\end{lemma}

\begin{proof}
We have previously established membership in APX for the problem of computing an optimal U-repair (see Section~\ref{sec:update-repairs}), thus it is only left to show that the problem is APX-hard.
We start by proving that computing an optimal U-repair under $\depset'_k$ for $k=1$ is APX-hard. In this case, the set of FDs contains two FDs $A_0A_1\rightarrow B_0$ and $A_1A_2\rightarrow B_1$. This FD set has a common lhs $A_1$, thus Corollary~\ref{cor:U-S-same-mc-1}, combined with the fact that computing an optimal S-repair for an FD set of the form $\set{A\rightarrow B, C\rightarrow D}$ is APX-hard, imply that computing an optimal U-repair is APX-hard as well.

Next, we construct a reduction from computing an optimal U-repair under $\depset'_k$ for $k=1$ to computing an optimal U-repair under $\depset'_k$ for $k>1$. Given a table $T$ over the schema $R(A_0,A_1,A_2,B_0,B_1)$, we construct a table $T'$ over the schema $R(A_0, \dots, A_{k+1},
  B_0,\dots,B_k)$, where for every tuple $\tup t = (a_0, a_1, a_2, b_0, b_1)$ in $T$, we create a tuple $\tup t' = (a_0,a_1,a_2,\odot,\dots,\odot,b_0,b_1,\odot,\dots,\odot)$ in $T'$. That is, $\tup t'.A=\tup t.A$ for every $A\in\set{A_0,A_1,A_2,B_0,B_1}$, and $\tup t'=\odot$ for the rest of the attributes. We claim that $T$ has a consistent update of distance $\leq m$  if and only if $T'$ has a consistent update of distance $\leq m$.
\par
\paragraph*{The ``if'' direction} 
Suppose that $T$ has a consistent update $U$ of distance $m$. We can obtain a consistent update $U'$ of $T'$ that has the same distance by updating the values of $A_0, A_1,A_2,B_0,B_1$ in $T'$ in exactly the same way we update these values in $T$, and leave the values in the rest of the attributes in $T'$ unchanged. Clearly, each FD that is not one of $A_0A_1\rightarrow B_0$ or $A_1A_2\rightarrow B_1$ is satisfied by $U'$ (since we did not change the values of the attributes in $\set{B_2,\dots,B_k}$, thus all the tuples in $U'$ have the same value $\odot$ in these attributes, and they agree on the rhs of each of these FDs). Moreover, if two tuples $\tup t$ and $\tup t'$ violate one of $A_0A_1\rightarrow B_0$ or $A_1A_2\rightarrow B_1$, then the corresponding two tuples in $U$ also violate these FDs, which is a contradiction to the fact that $U$ is a consistent update of $T$. Clearly, the distance of both updates is the same.
\par
\paragraph*{The ``only if'' direction} 
Suppose we have a consistent update $U'$ of $T'$ of distance $m$. We can obtain a consistent update $U$ of $T$ that has a lower or equal distance by updating the values of $A_0, A_1,A_2,B_0,B_1$ in $T$ in exactly the same way we update these values in $T'$. Let us assume, by way of contradiction, that $U$ is inconsistent. In this case, there are two tuples $\tup t_1$ and $\tup t_2$ in $U$ that violate one of $A_0A_1\rightarrow B_0$ or $A_1A_2\rightarrow B_1$. There is a tuple $\tup t_1'$ in $U'$ that agrees with $\tup t_1$ on the value of each one of the attributes in $\set{A_0, A_1,A_2,B_0,B_1}$. Similarly, there is a tuple $\tup t_2'$ in $U'$ that agrees with $\tup t_2$ on the value of each one of the attributes in $\set{A_0, A_1,A_2,B_0,B_1}$. Clearly, these two tuples also violate the FDs $A_0A_1\rightarrow B_0$ or $A_1A_2\rightarrow B_1$, which is a contradiction to the fact that $U'$ is a consistent update of $T'$. Clearly, the distance of $U$ is at most $m$ (it can be lower than $M$ if $U'$ also updates values in the attributes not in  $\set{A_0, A_1,A_2,B_0,B_1}$).
\end{proof}

\end{document}